\documentclass[USenglish,11pt]{scrartcl}

\usepackage[utf8]{inputenc}
\usepackage[T1]{fontenc}
\usepackage{lmodern}
\usepackage{a4wide}

\usepackage{amsmath,amsthm,amsfonts,amssymb}
\usepackage{thmtools}
\usepackage{graphicx,caption,subcaption,wrapfig}
\usepackage[]{hyperref}
\usepackage{paralist}
\usepackage{enumitem}
\usepackage{csquotes}
\usepackage[ruled,linesnumbered]{algorithm2e}
\usepackage{booktabs}
\usepackage{xkeyval,xspace}
\usepackage{authblk}

\graphicspath{{./figures/}}

\newtheorem{theorem}{Theorem}[section]
\newtheorem{lemma}[theorem]{Lemma}
\newtheorem{corollary}[theorem]{Corollary}
\theoremstyle{definition}
\newtheorem{definition}[theorem]{Definition}

\numberwithin{equation}{section}

\newcommand{\NN}{\mathbb{N}}

\newcommand{\ZZ}{\mathbb{Z}}

\DeclareMathOperator{\OPT}{OPT}
\DeclareMathOperator{\SIZE}{SIZE}

\newcommand{\card}[1]{\left\vert{#1}\right\vert }
\newcommand{\bigO}[1]{\ensuremath{\mathcal{O}\left(#1\right)}}

\newcommand{\floor}[1]{\left\lfloor{#1}\right\rfloor}
\newcommand{\ceil}[1]{\left\lceil{#1}\right\rceil}
\newcommand{\floorS}[1]{\lfloor{#1}\rfloor}
\newcommand{\ceilS}[1]{\lceil{#1}\rceil}

\newcommand{\norm}[1]{\left\lVert#1\right\rVert_1}
\newcommand{\nnz}[1]{\mathrm{nnz}(#1)}

\DeclareMathOperator{\LIN}{LIN}
\DeclareMathOperator{\LOAD}{LOAD}

\newcommand{\con}{\mathit{con}}	
\newcommand{\C}[1]{C_{#1}}	
\newcommand{\CR}[2]{C_{#1}^{#2}}	

\newcommand{\wmax}{w_{max}}
\newcommand{\wmin}{w_{min}}

\newcommand{\fract}[1]{\mathit{frac}\left({#1}\right)}

\newcommand{\kGeneric}{\floor{\frac{\epsilon}{4 \omega h_B} \SIZE(I_L)}}
\newcommand{\wGeneric}{\floor{\log \frac{1}{\epsilon}} +1}
\newcommand{\ILg}{{I_L}^g}	

\newcommand{\Repack}{\mathrm{Repack}}

\newcommand{\invA}{\ref{inv:a:categories}\xspace}
\newcommand{\invB}{\ref{inv:b:sorting}\xspace}
\newcommand{\invC}{\ref{inv:c:nrContainersBlockA}\xspace}
\newcommand{\invD}{\ref{inv:d:nrContainersBlockB}\xspace}
\newcommand{\invE}{\ref{inv:e:heightInterval}\xspace}

\newcommand*{\abbrdot}{.\@\xspace}

\newcommand*{\etal}{et~al\abbrdot}
\newcommand*{\ie}{i.\@\,e.\@\xspace}
\newcommand*{\eg}{e.\@\,g.\@\xspace}

\newcommand*{\sut}{s.t.\@\xspace} 

\author[1]{Klaus Jansen}
\author[1]{Kim-Manuel Klein}
\author[1]{Maria Kosche}
\author[2]{Leon Ladewig}
\affil[1]{Department of Computer Science, Kiel University, Germany\\
	\texttt{\{kj,kmk,mkos\}@informatik.uni-kiel.de}}
\affil[2]{Department of Computer Science, Technical University of Munich, Germany\\
	\texttt{ladewig@in.tum.de}}

\title{\vspace{-2.0cm} Online Strip Packing with Polynomial Migration 
	\footnote{An extended abstract of the paper has been published in APPROX 2017.
		This work was partially supported by DFG Project, "Robuste Online-Algorithmen f\"ur Scheduling- und Packungsprobleme", JA 612 /19-1, and by GIF-Project "Polynomial Migration for Online Scheduling"}}

\date{}

\begin{document}

\maketitle
\begin{abstract}
	\vspace{-10ex}
We consider the relaxed online strip packing problem: Rectangular items arrive online and have to be packed without rotations into a strip of fixed width such that the packing height is minimized. Thereby, repacking of previously packed items is allowed. The amount of repacking is measured by the migration factor, defined as the total size of repacked items divided by the size of the arriving item.
First, we show that no algorithm with constant migration factor can produce solutions with 
asymptotic ratio better than 4/3.
Against this background, we allow amortized migration, \ie to save migration for a later time step. As a main result, we present an AFPTAS with asymptotic ratio $1 + \bigO{\epsilon}$ for any $\epsilon > 0$ 
and amortized migration factor polynomial in $1 / \epsilon$. 
To our best knowledge, this is the first algorithm for online strip packing 
considered in a repacking model.
\end{abstract}

\section{Introduction}

In the classical \textit{strip packing} problem we are given a set of  
two-dimensional items with heights and widths bounded by~1 and a strip of 
infinite height and width~1. The goal is to find a packing
of all items into the strip without rotations such that no items overlap 
and the height of the packing is minimal.
In many practical scenarios the entire input is not known in advance.
Therefore, an interesting field of study is 
the online variant of the problem.  Here, items arrive over time and have 
to be packed immediately without knowing future items.
Following the terminology of \cite{gambosi2000algorithms} for the online
bin packing problem, in the \textit{relaxed online strip packing} problem
previous items may be repacked when a new item arrives.

There are different ways to measure the amount of repacking in a relaxed
online setting. 
Sanders, Sivadasan, and Skutella introduced the \textit{migration model} in \cite{sanders}
for online job scheduling on identical parallel machines as follows: 
When a new job of size $p_j$ arrives, jobs of total size $\mu p_j$ can be reassigned,
where $\mu$ is called the \textit{migration factor}.
In the context of online strip packing the migration factor $\mu$ ensures that
the total area of repacked items is at most $\mu$ times the area of the
arrived item. We call this the \textit{strict migration model}.

By a well known relation between strip packing and parallel job scheduling
\cite{hurink2007online}, any (online) strip packing algorithm applies to (online) scheduling
of parallel jobs. The latter problem is highly relevant \eg in computer systems
\cite{hurink2007online,steiger2003online,pruhs2004online}.

\paragraph{Preliminaries}
Since strip packing is NP-hard \cite{baker1980orthogonal},
research focuses on efficient approximation algorithms.
Let $A(I)$ denote the packing height of algorithm $A$ on input $I$
and $\OPT(I)$ the minimum packing height.
The \textit{absolute (approximation) ratio} is defined as
$\sup_{I} A(I) / \OPT(I)$  while the
\textit{asymptotic (approximation) ratio} as
$\lim \sup_{\OPT(I) \rightarrow \infty} \allowbreak A(I) / \OPT(I)$.
A family of algorithms $\{A_\epsilon\}_{\epsilon > 0}$ is called \textit{polynomial-time approximation
scheme (PTAS)}, when $A_\epsilon$ runs in polynomial-time in the input length and has absolute ratio $1+\epsilon$.
If the running time is also polynomial in $1 / \epsilon$, we call $\{A_\epsilon\}_{\epsilon > 0}$
\textit{fully polynomial-time approximation scheme (FPTAS)}. Similarly, the terms \textit{APTAS} and \textit{AFPTAS}
are defined using the asymptotic ratio.

All ratios of online algorithms in the following are \textit{competitive}, \ie online algorithms are compared with an 
optimal offline algorithm.

\subsection{Related Work}

\paragraph*{Offline}
Strip packing is one of the classical packing problems and receives ongoing 
research interest in the field of combinatorial optimization.
Since Baker, Coffman and Rivest \cite{baker1980orthogonal} gave the first 
algorithm with
asymptotic ratio 3, strip packing was investigated in many studies,
considering both asymptotic and absolute approximation ratios.
We refer the reader to \cite{christensen2017approximation} for a survey.
A well-known result is the AFPTAS by Kenyon and R{\'e}mila \cite{kenyon2000near}.
Concerning the absolute ratio, currently the best known algorithm of
ratio $5/3 + \epsilon$ for any $\epsilon > 0$ is by Harren \etal \cite{harren20145}.

An interesting result was given by Han \etal in 2007.
In \cite{han2007strip}, they studied the relation between bin packing
and strip packing and developed a framework between both problems.
For the offline case it is shown that any bin packing algorithm can be applied to
strip packing while maintaining the same asymptotic ratio.

\paragraph*{Online}
The first algorithm for online strip packing was given by Baker and Schwarz 
\cite{baker1983shelf} in 1983. Using the concept of shelf algorithms 
\cite{baker1980orthogonal}, they derived the algorithm \textit{First-Fit-Shelf}
with asymptotic ratio arbitrary close to $1.7$ and 
absolute ratio 6.99 (where all rectangles have height at most $h_{\max}~=~1$).
Later, Csirik and Woeginger \cite{csirik1997shelf}
showed a lower bound of $h_\infty \approx 1.69$ on the asymptotic ratio 
of shelf algorithms and gave an improved shelf algorithm with 
asymptotic ratio $h_\infty + \epsilon$ for any $\epsilon > 0$.

The framework of Han \etal \cite{han2007strip} is applicable in the online setting
if the online bin packing algorithm belongs to the class \textit{Super Harmonic}.
Thus, Seiden's online bin packing algorithm \textit{Harmonic++} \cite{seiden2002online}
implies an algorithm for online strip packing with asymptotic ratio 1.58889.
In 2007 and 2009, the concept of \textit{First-Fit-Shelf} by Baker and Schwarz 
was improved independently by two research groups, 
Hurink and Paulus \cite{hurink2007online} and Ye, Han, and Zhang \cite{ye2009note}. 
Both improve the absolute competitive ratio of from 6.99 to 6.6623 without a 
restriction on $h_{\max}$.
Further results on special variants of online strip packing were given by
Imeh \cite{imreh2001online} and Ye, Han, and Zhang~\cite{ye2011multi}.

On the negative side, there is no algorithm for online strip packing (without repacking)
with an asymptotic ratio better then 1.5404 since the lower bound in \cite{baloghBoundsJournal}
for online bin packing is also valid for online strip packing.
Regarding the absolute ratio, the first lower bound of 2
from \cite{brown1982lower}
was improved in several studies \cite{johannes2006scheduling,hurink2008online,kern2009note}.
Currently, the best known lower bound by Yu, Mao, and Xiao \cite{yu2016new}
is $(3+\sqrt{2})/2 \approx 2.618$.

\paragraph{Related results on the migration model}
Since its introduction by Sanders, Sivadasan, and Skutella \cite{sanders}, 
the migration model became increasingly popular.
In the context of online scheduling on identical machines,
Sanders, Sivadasan, and Skutella \cite{sanders} gave a PTAS
with migration factor $2^{\bigO{(1/\epsilon) \log^2 (1 / \epsilon)}}$ for the 
objective of minimizing the makespan.
Thereby, the migration factor in \cite{sanders} depends only on the
approximation ratio $\epsilon$ and not on the input size. Such algorithms
are called \textit{robust}.

Skutella and Verschae \cite{skutellaVerschaeJournal} studied scheduling
on identical machines while maximizing the minimum machine load, called \textit{machine covering}. 
They considered the \textit{fully dynamic} setting in which jobs are also allowed
to depart.
Skutella and Verschae showed that there is no PTAS for this problem in the migration model,
which is due to the presence of very small jobs.
Instead, they introduced
the \textit{reassignment cost model}, in which the migration factor is defined amortized. 
Using the reassignment cost model, they gave
a robust PTAS for the problem with amortized migration factor
$2^{\bigO{(1/\epsilon) \log^2 (1 / \epsilon)}}$.

Also online bin packing has been investigated in the migration model 
in a sequence of papers, inspired by the work of Sanders, Sivadasan, and Skutella \cite{sanders}:
The first robust APTAS for relaxed online bin
packing was given in 2009 by Epstein and Levin \cite{epstein2009robust}.
They obtained an exponential migration factor 
$2^{\bigO{(1/\epsilon^2) \log 1 / \epsilon}}$.
In 2013, Jansen and Klein \cite{jansen2013robust} improved this result
and gave an AFPTAS with polynomial migration factor
$\bigO{\frac{1}{\epsilon^3} \log \frac{1}{\epsilon^4}}$. The development of advanced
LP/ILP-techniques made this notable improvement possible.
Furthermore, in \cite{fullyDynamicBP} Berndt, Jansen, and Klein built upon the 
techniques developed in \cite{jansen2013robust} to give an AFPTAS
for fully dynamic bin packing with a similar migration factor.

Recently, Gupta \etal \cite{DBLP:journals/corr/abs-1711-02078} studied fully dynamic bin packing with several repacking measures.
In the amortized migration model, they presented an APTAS with amortized migration factor 
$\bigO{1/\epsilon}$. Further, this is shown to be optimal by providing a lower bound of 
of $\Omega ( 1 / \epsilon)$ for any such algorithm in the amortized migration model.

\paragraph*{Our contribution}
To the authors knowledge, there exists currently no algorithm for online 
strip packing in the migration or any other repacking model. 
Therefore, we present novel ideas to obtain the following results:
First, a relatively simple argument shows that in the strict migration model
it is not possible to maintain solutions that are close to optimal.
We prove the following theorem in Section~\ref{abs:sec:lowerBound}:
\begin{theorem}
\label{abs:theo:lowerBoundMarten}
	In the strict migration model, there is no approximation algorithm
	for relaxed online strip packing with asymptotic competitive ratio better
	than $4/3$.
\end{theorem}
Therefore, it is natural to relax the strict migration model such that
amortization is allowed in order to obtain an asymptotic approximation scheme.
We say that an algorithm has an \textit{amortized} migration factor of 
$\mu$ if for every time step $t$ the total migration (\ie the total area of repacked items)
up to time $t$ is bounded by $\mu \sum_{j=1}^t \SIZE(i_j)$, where
$\SIZE(i_j)$ is the area of item $i_j$ arrived at time $j$. 
Adapted to scheduling problems this corresponds with the reassignment cost model
introduced by Skutella and Verschae in \cite{skutellaVerschaeJournal}.
We adapt several offline and online techniques and combine them with our novel approaches
to obtain the following main result:
\begin{theorem}
\label{abs:theo:mainResult}
	There is a robust AFPTAS for relaxed online strip packing with an amortized
	migration factor polynomial in $1/\epsilon$.
\end{theorem}

\subsection{Technical Contribution}
\label{abs:sec:technicalContribution}
A general approach in the design of robust online algorithms is to rethink
existing algorithmic strategies that work for the corresponding offline problem in a way that the algorithm can adapt to a changing problem instance.
The experiences that were made so far in the design of robust algorithms 
(see \cite{jansen2013robust,fullyDynamicBP,skutellaVerschaeJournal}) are to design the algorithm in a way such that the generated solutions fulfill very tight structural properties. Such solutions can then be adapted more easily as new items arrive.

A first approach would certainly be do adapt the well known algorithm for 
(offline) strip packing by Kenyon and R{\'e}mila \cite{kenyon2000near} to the online setting.
However, we can argue that the solutions generated by this algorithm do not fulfill sufficient structural properties. 
In the algorithm by Kenyon and R{\'e}mila, the strip is divided vertically into
segments, where each segment is configured with a set of items. 
Thereby, each segment can have a different height.
Now consider the online setting, where we are
asked for a packing for the enhanced instance that maintains most parts of the existing
packing. Obviously, it is not enough to place new items on top of the packing as this would exceed the approximation guarantee.
Instead, existing configurations of the segments need to be changed to guarantee a good competitive ratio. 
However, this seems to be hard to do as the height of a configuration can change. 
Gaps can occur in the packing as a segment might decrease in height or vice versa a segment might increase in height and therefore does not fit anymore in its current position. Over time this can lead to a very fragmented packing.
On the other hand, closing gaps in a fragmented packing can cause a huge amount of repacking.

\begin{figure}
	\captionsetup[subfigure]{justification=centering}
	\centering
	\begin{subfigure}[t]{0.47\textwidth}
		\centering
		\includegraphics[height=3cm]{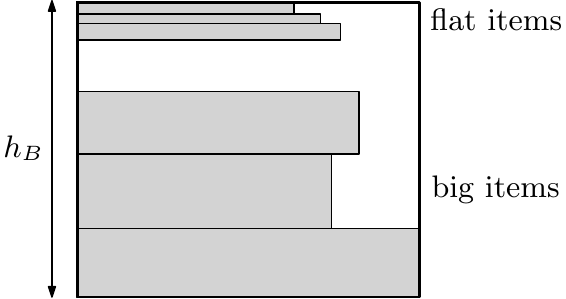}
		\caption{Packing of items in a container}
		\label{abs:fig:containerContent}
	\end{subfigure}
	\begin{subfigure}[t]{0.47\textwidth}
		\centering
		\includegraphics[height=3cm]{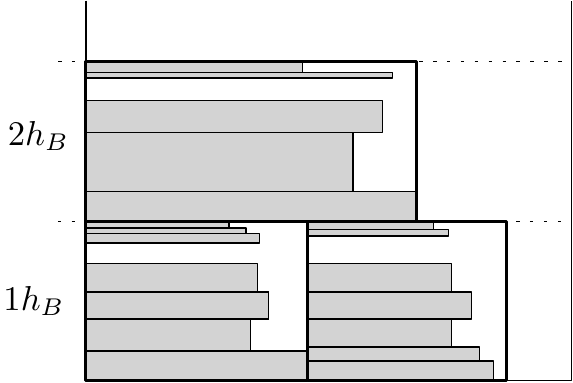} 
		\caption{Packing of containers in the strip}
		\label{abs:fig:containerAreBPItems}
	\end{subfigure}	
	\caption{
		Packing structure of our approach. Items are packed into containers of fixed height~$h_B$,
		thus the packing of containers results in a bin packing problem.}
\end{figure}

\paragraph{Containers}
In order to maintain packings with a more modular structure,
items are batched to larger rectangles of fixed height, called \textit{containers}
(see Figure \ref{abs:fig:containerContent}).
The width of a container equals the width of the widest item inside.
As each container has the same height $h_B$, the strip is divided
into \textit{levels} of equal height $h_B$ (see Figure \ref{abs:fig:containerAreBPItems})
and the goal is to fill each level with containers best possible.
Thus, finding a container packing is in fact a bin packing problem,
where levels correspond with bins and the sizes of the bin packing items are given 
by the container widths.
This approach was studied in the offline setting by Han \etal in \cite{han2007strip},
while an analysis in a migration setting is more sophisticated.

Thus, the packing of items into the strip is given by two assignments:
By the \textit{container assignment} each item is assigned to a container. 
Moreover, the \textit{level assignment} describes which container
is placed in which level (corresponds with the bin packing solution).
To guarantee solutions with good approximation ratio, both functions have to
satisfy certain properties.

\paragraph{Dynamic rounding / Invariant properties}
For the container assignment, a natural choice would be to assign the widest
items to the first container, the second widest to the second container, and so on
\cite{han2007strip}. 
However, in the online setting we can not maintain this strict order
while bounding the repacking size. Therefore, we use a relaxed ordering 
by introducing groups for containers of similar width and
requiring the sorting over the groups, rather than over
containers.
For this purpose, we adapt the dynamic rounding technique 
developed by Berndt, Jansen, and Klein in \cite{fullyDynamicBP} to
state important \textit{invariant properties}.

\paragraph{Shift}
In order to insert new items, we develop an operation called \textsc{Shift}. 
The idea is to move items between containers of different groups such that
the invariant properties stay fulfilled. 
When inserting an item $i_t$ via \textsc{Shift} into group\footnote{In the following, 
by \enquote{group of an item} we mean the group of the container
	in which the item is placed.} $g$, items are moved
from $g$ to the group $\mathit{left}(g)$, where again items are shifted to the
next group, and so on (see Figure~\ref{abs:fig:shiftTopLevelIntro}).
Thereby, the total height of the shifted items can increase in each step. 
However, it
is limited such that items that can not be shifted further (at group $g_0$
in Figure~\ref{abs:fig:shiftTopLevelIntro}) can be packed into one additional container.
This way, we get a new container
assignment for the enhanced instance which maintains the approximation guarantee
and all desired structural properties.

\paragraph{LP/ILP-techniques}
As a consequence of the shift operation, a new container may has to be inserted into the packing. 
We apply the LP/ILP-techniques developed in 
\cite{jansen2013robust} to maintain a good level assignment while
keeping the migration factor polynomial in $1 / \epsilon$.

\paragraph{Packing of small items}
Another challenging part regards the handling of items with small area.
Without maintaining an advanced structure, small items can fractionate the
packing in a difficult way. 
Such difficulties also arise in related optimization problems,
see e.g. \cite{skutellaVerschaeJournal,fullyDynamicBP}.
For the case of flat items (small height) we overcome these difficulties by the packing structure
shown in Figure~\ref{abs:fig:containerContent}: Flat items are separated from big items
in the containers and are sorted by width such that the least wide item is at the top.
Narrow items (small width) can be used to fill gaps in the packing when they are grouped to
similar height.

\begin{figure}
	\begin{center}
		\includegraphics[scale=1,page=4]{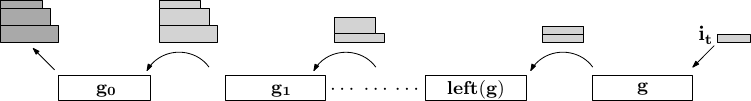}
		\caption{
			\textsc{Shift} operation moves widest items between groups to insert new item $i_t$.}
		\label{abs:fig:shiftTopLevelIntro}
	\end{center}
\end{figure}

\subsection{Lower Bound}
\label{abs:sec:lowerBound}

In this section we prove Theorem \ref{abs:theo:lowerBoundMarten}.
We use an adversary to construct an instance $I$ with arbitrary optimal
packing height, but $A(I) \geq \frac{4}{3} \OPT(I)$ for any such algorithm $A$.\\

\noindent\textit{Proof.}
Let $A$ be an algorithm for relaxed online strip packing with migration
factor $\mu$. We show that for any height $h$ there is an instance $I$
with $\OPT(I) \geq h$ and $A(I) \geq \frac{4}{3} \OPT(I)$.
The instance consists of two item types: 

A \textit{big} item has width $\frac{1}{2} - \epsilon$ and height $1$, 
while a \textit{flat} item has width $\frac{1}{2} + \epsilon$
and height $\frac{1}{2 \ceil{\mu}}$ for $\epsilon < 1/6$.
Note that $A$ can not repack a big item $b$ when a flat item $f$ arrives,
as $\SIZE(b) > \mu \SIZE(f)$.

\begin{wrapfigure}[12]{R}{0.3\textwidth}
	\centering
	\includegraphics[scale=0.7]{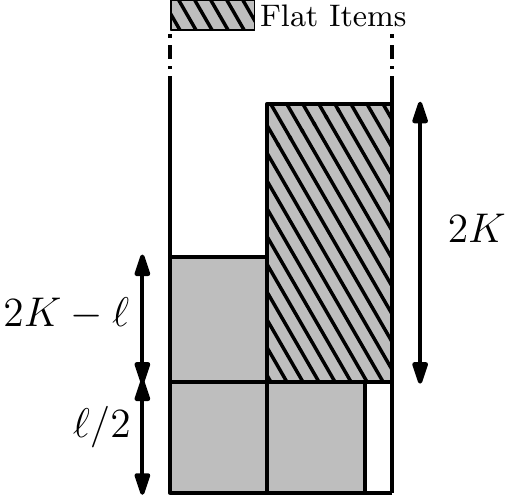}
	\caption{Optimal online packing}
	\label{abs:fig:lowerBoundOptOnline}
\end{wrapfigure}

First, the adversary sends $2K$ big items, where $K = 3 \ceil{h}$. 
Let $\ell$ be the number of big items that are packed by $A$ next to
another big item. The packing has a height of at least 
$\frac{\ell}{2} + 2K - \ell = 2K - \frac{\ell}{2}$
(see Figure \ref{abs:fig:lowerBoundOptOnline}). 
Since the optimum packing height for $2K$ big items
is $K$, algorithm $A$ has an absolute ratio
of at least $2 - \frac{\ell}{2K}$. If $\ell \leq \frac{4K}{3}$, the absolute
ratio is at least $\frac{4}{3}$ and nothing else is to show.

Now assume $\ell > \frac{4K}{3}$. In this case, the adversary sends
$k = 4 \ceil{\mu} K$ flat items of total height $2K$. In the optimal packing 
of height $2K$ big items and flat items form separate stacks placed in parallel.
Note that no two flat items can be packed in parallel.
Since $A$ can not repack any big item when a flat item arrives,
in the best possible packing achievable by $A$ flat items of total
height $2K - \ell$ are packed next to big items
(see Figure \ref{abs:fig:lowerBoundOptOnline},
flat items are packed in the dashed area).
Therefore, the packing height is at least
$2K + \frac{\ell}{2}$ and hence the absolute ratio is at least 
$1 + \frac{\ell}{4K} \geq \frac{4}{3}$.

In either case, it follows that the asymptotic ratio is at least $4/3$ by considering
$h \to \infty$.

\subsection{Remainder of the paper}
The following section introduces the approach of packing items into 
containers. We formulate central invariant properties and prove
important consequences of this approach.
Section~\ref{sec:OperationsBigFlat} presents operations that change the packing in the online
setting while maintaining the invariant properties.
Up to that section, items of small width were omitted.
In Section \ref{sec:narrowItems} we show how to integrate them into the container
packing. Finally, in Section~\ref{sec:migrationAnalysis} the amount of migration is analyzed.

Throughout the following sections, let $\epsilon \in (0,1/4]$ be
a constant such that $1 / \epsilon$ is integer.
We denote the height and width of an item $i$ by $h(i)$ and $w(i)$ (both at most 1)
and define $\SIZE(i) = w(i) h(i)$.
An item~$i$ is called \textit{big} if $h(i) \geq \epsilon$ and $w(i) \geq \epsilon$,
\textit{flat} if $w(i) \geq \epsilon$ and $h(i) < \epsilon$, and
\textit{narrow} if $w(i) < \epsilon$.
Accordingly, we partition the instance $I$ into $I_L$, the set of big and flat items
(having minimum width $\epsilon$), and $I_N$, the set of narrow items.
If $R$ is a set of items, let $\SIZE(R) = \sum_{i \in R} \SIZE(i)$
and $h(R) = \sum_{i \in R} h(i)$.

\section{Container Packing}
\label{sec:containerPacking}


Let $h_B \geq 1$ be the height of a container
and let $\mathcal{C}$ denote the set of containers.
In the following sequence of definitions we describe the container-packing approach formally.
In the first definition, a container $c \in \mathcal{C}$ is an object without geometric interpretation
to which items from $I_L$ can be assigned to, while satisfying the height capacity:

\begin{definition}[Container assignment]
	A function $\con \colon I_L \rightarrow \mathcal{C}$ is called a \textit{container assignment} 
	if $\sum_{i \colon \con(i) = c} h(i) \leq h_B$ holds for each $c \in \mathcal{C}$.

\end{definition}
The next definition shows how to build the \textit{container instance},
\ie the strip packing instance which allows to pack $I_L$ according to a container assignment.

\begin{definition}[Container instance]
	\label{def:containerInstace}
	Let $\con \colon I_L \rightarrow \mathcal{C}$ be a container assignment.
	The \textit{container instance} to $\con$,
	denoted by $\C{\con}$, is the strip packing instance
	defined as follows:
	For each $c \in \mathcal{C}$ such that there is an $i \in I_L$ with
	$\con(i) = c$, define a rectangle $\bar{c} \in \C{\con}$ of
	fixed height $h(\bar{c}) = h_B$ and width
	$w(\bar{c}) = \max_{i \in I_L} \{ w(i) \mid \con(i) = c \}$.
\end{definition}
Since by Definition~\ref{def:containerInstace} each container $c \in \mathcal{C}$ corresponds with
a unique rectangle $\bar{c} \in \C{\con}$, in the following we use both notions synonymously.

Like shown in Figure\nobreakspace \ref {abs:fig:containerContent}, big and flat items are placed
inside the container differently:
\begin{itemize}
	\item Big items form a stack starting from the bottom.
	\item Flat items form a stack starting from the top. 
	Thereby, items in that stack are sorted such that the least wide item
	is placed at the top edge of the container.
\end{itemize}
The need for the special structure for flat items will become clear in
Section\nobreakspace \ref {sec:insertFlat}.

In order to solve the container packing problem via a linear program (LP), 
the number of occurring widths has to be bounded. 
For this purpose we define a \textit{rounding function}
$R \colon \mathcal{C} \rightarrow G$ that maps 
containers to groups.
The concrete rounding function will be specified later.
Using $R$ we obtain the  \textit{rounded container instance} by rounding up
each container to the widest in its group.

\begin{definition}[Rounded container instance]
	\label{def:roundedContainerInstace}	
	Let $\C{\con}$ be a container instance and 
	$R \colon \mathcal{C} \rightarrow G$ be a rounding function.
	The \textit{rounded container instance} $\CR{\con}{R}$ is obtained 
	from the container instance $\C{\con}$
	by setting the width of each container $c$ to a new width 
	$w^R(c) \geq w(c)$ with
	$w^R(c):= \max_{c' \in \mathcal{C}} \{w(c') \mid R(c') = R(c) \}$.
\end{definition}

\subsection{LP Formulation}
\label{sec:binpackingLP}
By Definition~\ref{def:roundedContainerInstace} a rounded container instance has only containers
of $\card{G}$ different widths. Therefore we can use the following LP formulation to find
a level assignment.
It is commonly used for bin packing problems and was
first described by Eisemann \cite{eisemann1957trim}.

\begin{definition}[LP($C$)]
	\label{lp:binpacking}
	
	Let $C$ be a container instance and assume that items from $C$
	have one of $m$ different widths $w_1,\ldots,w_m \geq \epsilon$.
	Let $P$ be the set of patterns and $n=\card{P}$.
	A \textit{pattern} $P_i \in P$ is a multiset of widths
	$ \{ a(P_i,1) : n_1, a(P_i,2): n_2, \ldots a(P_i,m) : n_M \}$
	where $a(P_i,j)$ for $1 \leq i \leq n$ and $1 \leq j \leq m$
	denotes how often width $w_j$ appears in pattern $P_i$. 
	Let $b_j$ for $1 \leq j \leq m$ be the number of containers having width~$w_j$.
	LP($C$) is defined as:
	\begin{align*}
	\min & \norm{x} & \\
	\text{s.t.} & \sum_{P_i \in P} x_i a(P_i,j) \geq b_j & \forall 1 \leq j \leq m \\
	& x_i \geq 0                             & \forall 1 \leq i \leq n
	\end{align*}
	The value $\norm{x}$ of an optimal solution to the 
	above LP relaxation is denoted by $\LIN(C)$.
	If $y \in \NN^n$ is an optimal solution to the LP with
	integer constraints $x_i \in \ZZ^+$, its value is denoted with $\OPT(C)$.
	Note that $\LIN(C) \leq \OPT(C)$.
\end{definition}


Since all containers have width greater or equal $\epsilon$, 
each level has at most $1 / \epsilon$ slots.
In each slot, there are $m+1$ possibilities to fill the slot:
Either one container of the $m$ different sizes is placed there,
or the slot stays empty.
Therefore, $n \leq (m+1)^{1/ \epsilon}$.

\subsection{Grouping}

In the next subsection we adapt the dynamic rounding technique developed
in \cite{fullyDynamicBP} for containers. It is based on the one by
Karmarkar and Karp \cite{karmarkarKarp} for (offline) bin packing,
but modified for a dynamic setting.

\subsubsection{Basic Definitions}
\label{sec:groupingBasic}

Each container is assigned to a \textit{(width) category} $l \in \NN$, 
where container $c$ has width category $l$ if $w(c) \in (2^{-(l+1)},2^{-l}]$. 
Let $W$ denote the set of all non-empty categories, \ie
$W = \{ l \in \NN \mid \exists c \in \mathcal{C} ~ 
w(c) \in (2^{-(l+1)},2^{-l}] \}$.

\begin{lemma}
	\label{lemma:Omega}
	The set $W$ of categories has at most $\omega := \wGeneric$ elements.
\end{lemma}

\begin{proof}
	Containers can have any width between $\epsilon$ and 1, just like the width of big and flat items. 
	The widest containers belong to category 0, while containers of minimum width $\epsilon$
	are assigned to the category $\floor{\log \frac{1}{\epsilon}}$.
	Thus $\card{W} \leq \omega$.
\end{proof}

\begin{figure}
	\begin{center}
		\includegraphics[width=\textwidth]{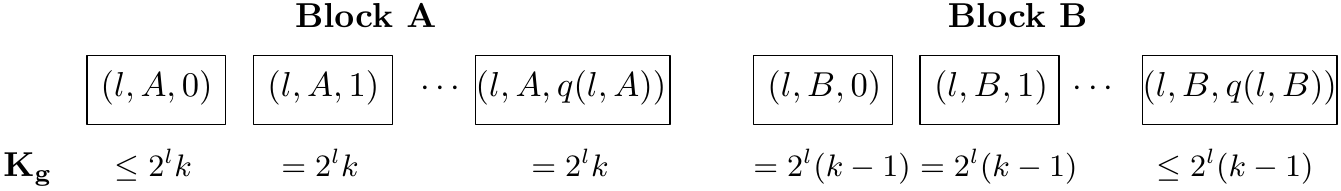}
		\caption{Groups of one category $l \in W$ and number of containers
			$K_g$ per group}
		\label{fig:groupStructure}
	\end{center}
\end{figure}

Furthermore, we build groups within the categories: A group $g \in G$ is a 
triple $(l,X,r)$, where $l \in W$ is the category, 
$X \in \{A,B\}$ is the \textit{block}, and 
$r \in \NN$ is the \textit{position} in the block. 
The maximum position of category $l$ at block $X$ that is non-empty 
is denoted by $q(l,X)$. Figure\nobreakspace \ref {fig:groupStructure} outlines the groups and 
the block structure of one category $l \in W$ (the values for $K_g$ 
will become clear in Section\nobreakspace \ref {sec:invariant}).

By the notion of blocks, groups of one category can be partitioned into two types.
This becomes helpful to maintain the invariant properties with respect to the growing set of 
items. More details on that are given in the later Sections~\ref{sec:invariant} and \ref{sec:insertBig}.
For a group $g=(l,X,r)$ we define the group to the left of it as follows:
$$
\mathit{left}(g) = 
\begin{cases} 
(l,X,r-1) 		& \mbox{if } r > 0 \vee X=A \\
(l,A,q(l,A)) 	& \mbox{if } r = 0 \wedge X=B \\ 
\end{cases}
$$
Analogously, we say $g = \mathit{right}(g')$, if $\mathit{left}(g) = g'$ holds. 
We set $\mathit{left}((l,A,0))=(l,A,-1)$ and $\mathit{right}((l,B,q(l,B))) = (l,B,q(l,B)+1)$ 
as empty groups.

From now on, let a rounding function $R$ and a container assignment $\con$ be given.
Let $K_g = \card{\{ c \in \mathcal{C} \mid R(c) = g \}}$ 
be the number of containers of group $g$.
We say that item $i$ has group $g$ if $R(\con(i)) = g$, that is, item $i$ is in a container which
belongs to group $g$.
Let $\ILg = \{ i \in I_L \mid R(\con(i)) = g \}$ be the set of items of group $g$
and define $h(g) = \sum_{i \in \ILg} h(i)$ as the total height of those items.

\subsubsection{Invariant Properties}
\label{sec:invariant}

In Section~\ref{abs:sec:technicalContribution} we argued that 
only solutions with strong structural properties can be adapted appropriately
in the online setting while maintaining a good competitive ratio.
In Definition~\ref{def:invariantProperties} we state a set of \textit{invariant}
properties, to which we refer with ($\mathcal{I}$1)-($\mathcal{I}$5) in the remainder of this paper.

\begin{definition}[Invariant properties]
	\label{def:invariantProperties}
	Let $k \in \NN$ be a parameter.
	\begin{enumerate}[label=(\ensuremath{\mathcal{I}}\arabic*)]
		\item \textbf{Width categories} \\
		\label{inv:a:categories}
		$2^{-(l+1)} < w(i) \leq 2^{-l}$  \hspace{20pt} 
		for all $i \in \ILg$ \sut $g=(l,\cdot,\cdot)$
		\item \textbf{Sorting of items over groups} \\
		\label{inv:b:sorting}
		$w(i) \geq w(i')$ \hspace{20pt}  
		for all $i \in \ILg ,i' \in {I_L}^{g'}$ \sut $g = \mathit{left}(g')$ 
		\item \textbf{Number of containers in block A} \\
		\label{inv:c:nrContainersBlockA}
		$K_{(l,A,0)} \leq 2^l k$, \\
		$K_{(l,A,r)} = 2^l k $   \hspace{20pt} 
		for all $l \in W$ and $1 \leq r \leq q(l,A)$
		\item \textbf{Number of containers in block B}  \\
		\label{inv:d:nrContainersBlockB}
		$K_{(l,B,q(l,B))} \leq 2^l(k-1)$, \\
		$K_{(l,B,r)} = 2^l(k-1)$  \hspace{20pt} 
		for all $l \in W$ and $0 \leq r < q(l,B)$
		\item \textbf{Total height of items per group} \\
		\label{inv:e:heightInterval}
		$ (h_B - 1) (K_g -1) \leq h(g) \leq (h_B - 1) K_g $ \hspace{20pt}
		for all $g \in G$
	\end{enumerate}
\end{definition}

Property \invA ensures that each item is assigned to the right category.
Note that as a consequence, each container of a group $(l,\cdot,\cdot)$
has a width in $(2^{-(l+1)},2^{-l}]$.
By property \invB, all items in a group $g$ have a width greater or equal
than items in the group $\mathit{right}(g)$. That is, instead of a strict
order over all containers, \invB ensures an order over groups of containers.
The properties \invC and \invD set the number of containers 
to a fixed value, except for special cases (see Figure~\ref{fig:groupStructure}):
Groups in block~$A$ have more containers than groups in block~$B$. 
Moreover, there are two \textit{flexible groups} 
(namely $(l,A,0)$ and $(l,B,q(l,B))$) whose number of containers is only upper
bounded.
Finally, property \invE ensures an important relation between items and containers
of a group $g$: Since $h(g) \leq K_g (h_B - 1)$, at least one of the $K_g$
containers has a filling height of at most $h_B - 1$ and thus can admit a new
item. However, the lower bound $(h_B - 1) (K_g - 1) \leq h(g)$ ensures that
each container is well filled in an average container assignment.

\subsubsection{Number of Groups}
\label{sec:numberOfGroups}

One of the important consequences of the invariant properties is the fact
that the number of non-empty groups $\card{G}$ can be bounded from above,
assuming that the instance is not too small.
Therefore, the parameter $k$ has to be set in a particular way:

\begin{lemma}
	\label{lemma:NumberOfGroups}
	Let $\omega$ be defined like in \MakeUppercase Lemma\nobreakspace \ref {lemma:Omega}. 
	For $k = \kGeneric$ the number of non-empty groups 
	in $G$ is bounded by $\bigO{\frac{\omega}{\epsilon}}$, assuming that 
	$\SIZE(I_L) \geq \frac{24 \omega h_B (h_B-1)}{\epsilon h_B - 2 \epsilon}$.
\end{lemma}	

\begin{proof}
	Let 
	$ G_1 = G \setminus \left( \bigcup_{l \in W} (l,A,0) \cup 
	\bigcup_{l \in W} (l,B,q(l,B)) \right) $
	and let $g \in G_1$. 
	Since by property \invA every container of group $g$ has width greater than 
	$2^{-(l+1)}$, it follows together with the further invariant properties
	\begin{align*}
	\SIZE(\ILg) &> 	  2^{-(l+1)} (h_B-1)(K_g-1) & \text{\invA, \invE} \\
	&\geq 2^{-(l+1)} (h_B - 1) (2^l (k-1) - 1) & \text{\invC, \invD} \\
	&=    \frac{1}{2} (h_B - 1) (k-1) - 2^{-(l+1)} (h_B-1) \\
	&\geq \frac{1}{2} (h_B - 1) (k-1) - \frac{h_B - 1}{2} \\
	&=    \frac{1}{2} (h_B - 1) (k-2) \,.
	\end{align*}	
	Now, let $I_L^{(l)}$ be the set of items in $I_L$ which belong to containers of 
	category $l$. It holds that
	$\SIZE(I_L^{(l)}) \geq \sum_{g=(l,\cdot,\cdot) \in G_1} \SIZE(\ILg) 
	\geq (q(l,A)+q(l,B)) \left( \frac{1}{2} (h_B - 1) (k-2) \right)
	$ and resolving leads to
	\begin{equation}
	\label{eq:QlaQlbCategory}
	q(l,A)+q(l,B) \leq \frac{2 \SIZE(I_L^{(l)})}{(h_B - 1)(k-2)} \,.
	\end{equation}
	We now show $(h_B - 1)(k-2) \geq \frac{\epsilon}{8 \omega h_B} \SIZE(I_L)$.
	The assumption on $\SIZE(I_L)$ is equivalent to 
	$\frac{\epsilon}{4 \omega h_B} \SIZE(I_L) - 3 
	\geq \frac{\epsilon}{8 \omega (h_B-1)} \SIZE(I_L)$. Therefore,
	$$
	k-2 
	= \floor{ \frac{\epsilon}{4 \omega h_B} \SIZE(I_L) } - 2
	\geq \frac{\epsilon}{4 \omega h_B} \SIZE(I_L) - 3
	\geq \frac{\epsilon}{8 \omega (h_B-1)} \SIZE(I_L)
	$$
	and thus
	$$
	(h_B - 1) (k-2) 
	\geq  \frac{(h_B-1) \epsilon}{8 \omega (h_B-1)} \SIZE(I_L)
	= \frac{\epsilon}{8 \omega} \SIZE(I_L) \,.
	$$
	Further, we get
	\begin{equation}
	\label{eq:boundQLAQLB}
	\begin{aligned}[b]
	\frac{2 \SIZE(I_L)}{(h_B-1) (k-2)}
	\leq \frac{2 \SIZE(I_L)}{\frac{\epsilon}{8 \omega} \SIZE(I_L)}
	= 	  \frac{16 \omega}{\epsilon}	\,.	
	\end{aligned}
	\end{equation}
	As shown in Figure\nobreakspace \ref {fig:groupStructure}, for each category $l$
	there are $q(l,A)+q(l,B)+2$ groups.
	Now, summing over all categories $l \in W$ concludes the proof:
	\begin{align*}
	&~~~   \sum_{l \in W} q(l,A) + q(l,B) + 2 \\
	&\leq  \sum_{l \in W} \left(\frac{2 \SIZE({I_L^{(l)})}}{(h_B - 1) (k-2)} + 2\right) 
	& \text{eq.\nobreakspace \textup {(\ref {eq:QlaQlbCategory})}} \\
	&=     2\card{W} + \frac{2}{(h_B - 1) (k-2)} \sum_{l \in W} \SIZE(I_L^{(l)}) \\
	&=     2\card{W} + \frac{2 \SIZE(I_L)}{(h_B - 1) (k-2)}  \\
	&\leq  2 \omega + \frac{16 \omega}{\epsilon} & \text{eq.\nobreakspace \textup {(\ref {eq:boundQLAQLB})}}
	\end{align*}	
\end{proof}

\subsection{Approximation Guarantee}
\label{sec:approxGuarantee}

If the invariant properties of Definition \ref{def:invariantProperties}
are fulfilled, the rounded container instance yields a good approximation to $I_L$.
Using a proof technique from \cite{han2007strip} we are able to prove the following theorem.

\begin{theorem}
	\label{theo:approxGuarantee}
	Let $\con \colon I_L \rightarrow \mathcal{C}$ be
	a container assignment and $R \colon \mathcal{C} \rightarrow G$
	be a rounding function such that invariant properties \invA-\invE are fulfilled.
	For the rounded container instance $\CR{\con}{R}$ it holds that
	$\OPT(\CR{\con}{R}) \leq (1+ 4 \epsilon) \OPT(I_L) + \bigO{1 / \epsilon^4}	$.
\end{theorem}

\paragraph{Preliminaries}
From now on, we suppose that the parameter $k$ of the invariant is 
$k = \kGeneric$ and the container height is $h_B = 13 / \epsilon^2$.
Let $\omega = \wGeneric$  (see Lemma~\ref{lemma:Omega}).
Furthermore, we assume from now on
\begin{equation}
\label{eq:sizeGEQhb+1}
\SIZE(I_L) \geq \frac{4 \omega h_B}{\epsilon} (h_B + 1) \,.
\end{equation}
%

The proof technique from \cite{han2007strip} uses the notion
of \textit{homogenous lists}, whose definition is given in the following.
For a list of rectangles $R$, let $h(R) = \sum_{r \in R} h(r)$.
Furthermore, let $R[j]$ be the list of rectangles from $R$ having width $w_j$.

\begin{definition}[Homogenous lists, \cite{han2007strip}] 
	\label{def:homogenousLists}
	Let $P$ and $Q$ be two lists of rectangles, where any rectangle takes
	a width of $q$ distinct numbers $w_1,w_2,\ldots,w_q$.
	$P$ is \textit{$s$-homogenous} to $Q$, where $s \geq 1$, if
	$h(Q[j]) \leq h(P[j]) \leq s \cdot h(Q[j])$ for all $j \in \{1,\ldots, q \} 
	$.
\end{definition}

The next lemma states an important property of homogenous lists.
For the proof we refer the reader to \cite{han2007strip} and 
references therein.

\begin{theorem}[\cite{han2007strip}]
	\label{theo:HomogenousLists}
	Let $P$ and $Q$ be two lists of rectangles, where each rectangle has maximum
	height $h_{max}$. If $P$ is $s$-homogenous to $Q$, then
	for any $\delta > 0$:
	$$ \OPT(P) \leq s \cdot (1+\delta) \OPT(Q) + \bigO{1 / \delta^2} h_{max}$$
\end{theorem}

\paragraph*{Proof idea of Theorem~\ref{theo:approxGuarantee}}
Instead of a strict order over containers (like in \cite{han2007strip}),
we can make use of the fact that containers are rounded to the
widest container of the group. Therefore, all items in $\mathit{right}(g)$ 
have a smaller or equal width than the items in $g$.
This observation leads to the definition of two instances, one instance $\hat{I}_L$
with rounded-down items and the other $C_1$ with all rounded containers except 
from some border groups of each category.
With the notion of homogenous lists (Definition\nobreakspace \ref {def:homogenousLists}), in the end 
Theorem\nobreakspace \ref {theo:HomogenousLists} can be applied.
Let $\con$ and $R$ be like in Theorem~\ref{theo:approxGuarantee}
and write $C$ instead of $\CR{\con}{R}$ for short.

\paragraph*{Definition of $\mathbf{C_1}$ and $\mathbf{\hat{I}_L}$}
For a category $l \in W$, the group $g_0(l)$ is defined as $(l,A,0)$ if
block $A$ is non-empty and $(l,B,0)$ otherwise. Analogously, define
$g_q(l)$ as $(l,B,q(l,B))$ if the $B$-block is non-empty and 
$(l,A,q(l,A))$ otherwise.
Let 
$G_1 = G \setminus \bigcup_{l \in W} \{g_0(l), \mathit{right}(g_0(l)),
g_q(l), \allowbreak \mathit{left}(g_q(l)) \} \,.$
That is, $G_1$ contains all groups except for the two left- and rightmost
groups (see Figure\nobreakspace \ref {fig:groupStructure}) which are non-empty.
Let $C_1 = \{ c \in C \mid R(c) \in G_1 \}$ be the set of rounded containers
of a group in $G_1$.

By rounding down every item from group $g$ to the rounded width of containers 
from the group $\mathit{right}(g)$, we obtain a rounded instance $\hat{I}_L$.
Formally, for each item $i \in \ILg$ with $g \in G_1$
define a new item $\hat{i} \in \hat{I}_L$
with $h(\hat{i}) = h(i)$ and 
$w(\hat{i}) = w^R(c)$ for any $c \in C$ with $R(c) = \mathit{right}(g)$.
Note that with invariant property \invB, $w(\hat{i}) \leq w(i)$ 
for all $i \in I_L$. Together with $h(\hat{i}) = h(i)$ we get
\begin{equation}
\label{eq:IIhat}
\OPT(\hat{I}_L) \leq \OPT(I_L) \,.
\end{equation}

\begin{lemma}
	\label{lemma:JhomogenousI}
	$C_1$ is $s$-homogenous to $\hat{I}_L$ with $s \leq 1+2\epsilon$.
\end{lemma}

\begin{proof}
	By definition $C_1$ is a set of rounded containers, where each container
	has height $h_B$ and the width of the most wide item in its group.
	Let $w_1,\ldots,w_q$ denote these widths. 
	The set $\hat{I}_L$ contains items of unchanged height and rounded-down width.
	
	Let $1 \leq j \leq q$. We consider the sets $\hat{I}_L[j]$ and $C_1[j]$ 
	containing the respective rectangles of width $w_j$.
	Let $G[j] = \{ g \in G \mid \exists c \in C_1[j] \colon R(c) = g \}$ 
	be the set of groups present in $C_1[j]$. 
	For each group $g \in G_1[j]$, items in $\hat{I}_L$
	of the group $g'=\mathit{left}(g)$ have width $w_j$ as they get rounded down
	to the width of the rounded containers to the right.
	Therefore, we can write the property from Definition\nobreakspace \ref {def:homogenousLists} as
	\begin{equation}
	\label{eq:homogenousListsGroupNotationSum}
	\sum_{g \in G[j]} h(g')
	\leq \sum_{g \in G[j]} h_B K_g 
	\leq s \sum_{g \in G[j]} h(g') \,.
	\end{equation}
	To prove Equation\nobreakspace \textup {(\ref {eq:homogenousListsGroupNotationSum})}, we show in the following
	for any group $g \in G_1$ and $g'=\mathit{left}(g)$
	\begin{align}
	\label{eq:homogenousListsGroupNotation}
	h(g') \leq h_B K_g \leq s h(g') \,.
	\end{align}
	
	Before we can show Equation~\ref{eq:homogenousListsGroupNotation} we argue that two properties hold:
	\begin{enumerate}[label=(\roman*)]
		\item $K_g \leq K_{g'} \leq K_g + 2^l$. 
		By construction of $G_1$, neither $g$ nor $g'$ can be a flexible group.
		Thus, we have fixed numbers of containers by \invC-\invD.
		If $g$ and $g'$ are in the same block, $K_g = K_{g'}$.
		If $g'$ is in block $A$ and $g$ in block $B$, we have $K_{g'}=2^l k$
		and $K_g = 2^l (k-1)$.
		
		\item $2^l (h_B-1) \leq 2^l (k-1) \leq K_g$.
		The first inequality holds since $h_B \leq k$: The minimum size on $\SIZE(I_L)$ (Equation~\ref{eq:sizeGEQhb+1})
		implies $h_B \leq \frac{\epsilon}{4 \omega h_B} \SIZE (I_L) - 1$ which is not greater than $k$.
		The second inequality follows from \invC-\invD.
			\end{enumerate}
	
	Now, the first inequality of eq.\nobreakspace \textup {(\ref {eq:homogenousListsGroupNotation})} can be
	proven:
	$$ h(g') \overset{\invE}{\leq}  (h_B - 1) K_{g'}
	\overset{\text{(i)}}{\leq} (h_B - 1) (K_g + 2^l) 
	=    K_g h_B - K_g + 2^l (h_B - 1)
	\overset{\text{(ii)}}{\leq} K_g h_B \,.
	$$
	
	It remains to prove the second inequality of Equation\nobreakspace \textup {(\ref {eq:homogenousListsGroupNotation})}.
	Let $K = \max_{g \in G_1} \frac{K_g}{K_g - 1}$ and 
	$H = \frac{h_B}{h_B - 1}$.
	By Lemma\nobreakspace \ref {lemma:KcanBeBounded}, which we prove later,
	$HK \leq s$ for $s \leq 1 + 2 \epsilon$ and this gives us
	$$h_B K_g 
	= \frac{h_B}{h_B - 1} \frac{K_g}{K_g - 1} 
	(h_B - 1) (K_g - 1) 
	\leq HK	(h_B - 1) (K_g - 1)
	\leq s (h_B - 1) (K_g - 1) 
	\leq s h(g') \,,
	$$
	where the last inequality follows from \invE and (i).
	Thus, $C_1$ is $s$-homogenous to $\hat{I}_L$.
\end{proof}

\begin{lemma}
	\label{lemma:KcanBeBounded}
	Let $K = \max_{g \in G_1} \frac{K_g}{K_g - 1}$ and
	$H = \frac{h_B}{h_B - 1} $.
	It holds that $HK \leq 1 + 2 \epsilon$.
\end{lemma}

\begin{proof}
	To prove the claim we show 
	\begin{enumerate}[label=(\roman*)]
		\item $H \leq 1+\gamma$ for $\gamma \leq \frac{\epsilon}{1+\epsilon}$ and
		\item $K \leq 1+\epsilon$.
	\end{enumerate}
	For property (i), let $\gamma = \frac{\epsilon^2}{13 - \epsilon^2} \leq
	\frac{\epsilon}{1+\epsilon}$. By definition of $h_B$ it follows $h_B = 1 + \gamma$.
	To show (ii), let $g \in G_1$.
	The minimum size assumption (Equation~\ref{eq:sizeGEQhb+1}) implies
	$\frac{\epsilon}{4 \omega h_B} \SIZE(I_L) - 2 \geq \frac{1+\epsilon}{\epsilon}$
	and therefore we get
	$$
	K_g \geq 2^l(k-1)
		   \geq k-1
		    =     \kGeneric - 1
		   \geq \frac{\epsilon}{4 \omega h_B} \SIZE(I_L) - 2 
		   \geq \frac{1+\epsilon}{\epsilon} \,,
	$$
	where the first inequality is due to \invC-\invD.
	The above statement is equivalent to $1-\frac{1}{K_g} \geq \frac{1}{1+\epsilon}$.
	Hence,
		$ K =     \max_{g \in G_1} \frac{K_g}{K_g - 1}
		=     \max_{g \in G_1} \frac{1}{1 - \frac{1}{K_g}}
		\leq  1+\epsilon
		$.

\end{proof}

With the previous lemmas we are now able to the main theorem of this section.

\begin{proof}[Proof of Theorem \ref{theo:approxGuarantee}]
	Let $s=1+2\epsilon$.
	By Lemma\nobreakspace \ref {lemma:JhomogenousI} 
	$C_1$ is $s$-homogenous to $\hat{I}_L$, so we
	can apply Theorem\nobreakspace \ref {theo:HomogenousLists} and get
	for any $\delta > 0$
	\begin{equation}
	\label{eq:JrIhat}
	\OPT(C_1) \leq s (1+\delta) \OPT(\hat{I}_L) + \bigO{1 / \delta^2} h_{max} \,.
	\end{equation}
	
	By construction of $G_1$, for each category $l \in W$ four groups were dropped
	from $G$ to obtain $G_1$. Each of the groups has by  \invC-\invD at most
	$2^l k$ containers. As each container of category $l$ has by \invA width at most $2^{-l}$ , 
	in one level of the strip $2^l$ containers can be placed.
	Hence, for each $l \in W$ we need at most 
	$4 \cdot 2^l k \cdot \frac{1}{2^l} = 4 k$
	extra levels of height $h_B$, causing additional height of $4 \card{W} h_B k$.
	Note that
	\begin{equation}
	\label{eq:LossOfG+}
	4 h_B \card{W} k
	\leq 4 h_B \omega k
	=    4 h_B \omega \kGeneric
	\leq    \epsilon \SIZE(I_L)	\,. 
	\end{equation}
	Thus any packing of $C_1$ can be turned into a packing of $C$,
	placing the missing containers into extra levels of total height
	at most $\epsilon \SIZE(I_L)$ which gives us
	\begin{equation}
	\label{eq:JrJrHat}
	\OPT(C) \leq \OPT(C_1) + \epsilon \SIZE(I_L) \,.
	\end{equation}
	Finally, we can bound $\OPT(C)$ as follows:
	\begin{align*}
	\OPT(C) &\leq \OPT(C_1) + \epsilon \SIZE(I_L) 
	&\text{eq.\nobreakspace \textup {(\ref {eq:JrJrHat})}}\\
	&\leq  s (1+\delta) \OPT(\hat{I}_L) + \epsilon \SIZE(I_L) + 
	\bigO{1 / \delta^2} h_{max}
	&\text{eq.\nobreakspace \textup {(\ref {eq:JrIhat})}} \\
	&\leq  s (1+\delta) \OPT(\hat{I}_L) + \epsilon \OPT(I_L) 
	+ 	\bigO{1 / \delta^2} h_{max}
	& \SIZE \leq \OPT \\			 
	&\leq  (s (1+\delta) + \epsilon) \OPT(I_L) 
	+ \bigO{1 / \delta^2} h_{max}
	&\text{eq.\nobreakspace \textup {(\ref {eq:IIhat})}}
	\end{align*}
	
	Setting	$\delta=\epsilon / s$ we finally get
	$       s (1+\delta) + \epsilon 
	= 1 + 4 \epsilon $
	as approximation ratio.
	As the maximum height $h_{max}$ is the container height $h_B$,
	the additive term is $\bigO{1 / \delta^2} h_B = \bigO{1 / \epsilon^4}$.
\end{proof}

\subsection{Interim Conclusion: Offline Algorithm}
\label{sec:interimConclusion}

So far we described how big and flat items get packed into containers
and analyzed the properties of this approach:
By \MakeUppercase Theorem\nobreakspace \ref {theo:approxGuarantee}, 
the rounded container instance is an approximation
to the strip packing problem of asymptotic ratio $1+4 \epsilon$.
Furthermore, solving the container packing problem can be done 
via LP \ref {lp:binpacking} since by Lemma\nobreakspace \ref {lemma:NumberOfGroups}
the number of rows is bounded by
$\bigO{\frac{\omega}{\epsilon}} 
= \bigO{\frac{1}{\epsilon} \log{\frac{1}{\epsilon}}}$.

Therefore, we could handle the offline scenario for big and flat items
with the techniques presented so far completely.
A container assignment and a rounding function
fulfilling the invariant properties can be found as follows:
Partition $I_L$ into $I_L^l$ according to categories $l \in W$.
For each set $I_L^l$, assign the widest items of each category $l$  to group $(l,A,0)$, 
the second widest items to $(l,A,1)$, and so on, using $2^l k$ containers for each group
(except for the last one). Block $B$ remains empty.
For property \invE, ensure that the total height of items per group
$h(g)$ is just above the lower bound, \ie $h(g) \leq  (h_B - 1) (K_g - 1) +1$.

Moreover, narrow items can be placed with a modified first fit algorithm
into gaps of the packing, presented in the later Section\nobreakspace \ref {sec:narrowItems}.
An outline of the offline AFPTAS is given in Algorithm\nobreakspace \ref {alg:offline}.

\begin{algorithm}
	\caption{Offline AFPTAS}
	\label{alg:offline}
	
	\SetAlgoLined
	\DontPrintSemicolon
	\SetKwInOut{Input}{Input}
	\SetKw{Or}{or}
	\SetKwFunction{Shift}{Shift}
	
	\SetAlgoNoEnd%
	
	\BlankLine
	
	Find a container assignment $\con$ and rounding function $R$ 
	fulfilling the invariant properties. \;
	
	Solve the container packing problem for $\CR{\con}{R}$ via
	LP \ref {lp:binpacking}. \;
	
	Obtain the packing by placing each item accordingly to the
	container assignment and the internal packing structure of
	a container (see Figure\nobreakspace \ref {abs:fig:containerContent}).
	
	Place narrow items greedily into gaps (see Section\nobreakspace \ref {sec:narrowItems}). \;
\end{algorithm}

\section{Online Operations for Big and Flat Items}
\label{sec:OperationsBigFlat}

In this section we present operations that integrate arriving items into the packing structure
such that all invariant properties are maintained.
The central operation for this purpose is called \textsc{Shift} and introduced in 
Section\nobreakspace \ref {sec:shiftOperation}.
While the insertion of big items (Section\nobreakspace \ref {sec:insertBig}) is basically a pure \textsc{Shift},
inserting flat items has to be done more carefully like described in Section\nobreakspace \ref {sec:insertFlat}.
From now on, the instance $I_L$ at time step $t$ is denoted by $I_L(t)$
and the arriving item by~$i_t$.
Nevertheless, we omit the parameter $t$ whenever it is clear from the context.

\subsection{Auxiliary Operations}
\label{sec:auxFunctions}

At first we define some auxiliary operations used in the \textsc{Shift} algorithm.

\begin{figure}	
	\begin{flushleft}
		\begin{subfigure}[t]{0.30\textwidth}
			\includegraphics[width=\textwidth,page=1]{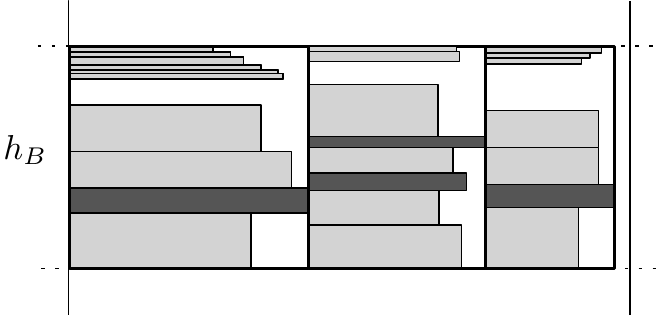}
			\caption{Given packing}
			\label{fig:sinkAlignOrigin}
		\end{subfigure}
		~
		\begin{subfigure}[t]{0.30\textwidth}
			\includegraphics[width=\textwidth,page=2]{align.pdf}
			\caption{Gaps due to absent items}
			\label{fig:sinkAlignOriginGaps}
		\end{subfigure}
		~
		\begin{subfigure}[t]{0.30\textwidth}
			\includegraphics[width=\textwidth,page=3]{align.pdf}
			\caption{After \textsc{Sink}}
			\label{fig:sink}
		\end{subfigure}	
		\vspace*{10pt}
		\begin{subfigure}[t]{0.30\textwidth}
			\includegraphics[width=\textwidth,page=4]{align.pdf}
			\caption{After \textsc{Align}}
			\label{fig:align}
		\end{subfigure}	
		~
		\begin{subfigure}[t]{0.30\textwidth}
			\includegraphics[width=\textwidth,page=5]{align.pdf}
			\caption{Overlapping item}
			\label{fig:sinkAlignOverlap}
		\end{subfigure}	
		~
		\begin{subfigure}[t]{0.30\textwidth}
			\includegraphics[width=\textwidth,page=6]{align.pdf}
			\caption{After \textsc{Stretch}}
			\label{fig:stretch}
		\end{subfigure}	
	\end{flushleft}
	
	\caption{Malformed packings due to inserted or removed items.
		Operations \textsc{Sink}, \textsc{Align}, and \textsc{Stretch}
		restore the packing structure.}	
\end{figure}

\begin{description}
	\item[WidestItems($T, s_{min}$)] 
	returns a set of items $S \subseteq T$
	such that $w(i) \geq w(i')$ for each $i \in S, i' \in T \setminus S$
	and $h(S) \in [s_{min}, s_{min}+1)$.
	That is, it picks greedily widest items from $T$ until the total height
	exceeds $s_{min}$.
	
	\item[Sink($c$)]
	places big items in container $c$ on a contiguous stack starting from the bottom 
	(gaps due to removed items are closed).
	See Figures\nobreakspace  \ref {fig:sinkAlignOrigin} to\nobreakspace  \ref {fig:sink}.
	
	\item[Align($u$)]
	places the containers of level $u$ according
	to their actual widths (widest item inside) left-aligned and without gaps.
	See Figures\nobreakspace  \ref {fig:sink} to\nobreakspace  \ref {fig:align}.
	This operation is necessary for the insertion of narrow items treated in the later 
	Section\nobreakspace \ref {sec:narrowItemsCombination}, but for convenience
	already presented here.
	
	\item[Stretch($u$)]
	resolves overlaps occurring when a new item $i$ gets placed into a container $c$ with $w(c) < w(i) \leq w^R(c)$.
	(see Figures \nobreakspace \ref {fig:stretch} to \ref{fig:sinkAlignOverlap}).
	Note that the level assignment is according to the LP \ref {lp:binpacking}
	which considers the rounded width $w^R(c)$ for a container $c$.
	The conditions for the insertion of an item (given later in 
	Section\nobreakspace \ref {sec:shiftOperation}) ensure that 
	only items $i$ with $w(i) \leq w^R(c)$ get inserted into a container $c$.
	Therefore, after \textsc{Stretch} all containers still 
	fit into their level\footnote{
		Like shown in the later Section\nobreakspace \ref {sec:narrowItems}, the gaps
		in a level will be filled with narrow items. These will be 
		suppressed by the \textsc{Stretch} operation.}.	
	
	\item[Place($g,S$)]
	packs items in $S$ into appropriate containers of group~$g$. 
	Depending on the item type (big or flat), different algorithms
	are used, which will be present in Sections\nobreakspace \ref {sec:insertBig} 
	and\nobreakspace  \ref {sec:insertFlat}.
	
	\item[InsertContainer (Algorithm\nobreakspace \ref {alg:insertContainer})]
	reflects the insertion of the new containers in the LP/ILP-solutions of LP~\ref{lp:binpacking}.
	To maintain the approximation guarantee, it makes use of the procedure 
	\textsc{Improve} from \cite{jansen2013robust}.
	Loosely speaking, calling \textsc{Improve}($\alpha$,$x$,$y$) on LP/ILP-solutions
	$x$, $y$ yields a new solution where the approximation guarantee is maintained and
	the additive term is reduced by $\alpha$.
	Further details are given in Section\nobreakspace \ref {sec:improveApprox} and Appendix\nobreakspace \ref {app:improve}.
\end{description}

\begin{algorithm}
	\caption{Insertion of a container $c$}
	\label{alg:insertContainer}
	
	\SetAlgoLined
	\DontPrintSemicolon
	\SetKwInOut{Input}{Input}
	\SetKwFunction{Shift}{Shift}
	\SetKwFunction{Improve}{Improve}	
	
	\Input{
		$x,y$ fractional and integral solutions to LP($\CR{\con}{R}$) \\
	}
	\BlankLine
	
	\Improve{1,$x$,$y$} \;
	Let $P_i$ the pattern such that $P_i = \{ w^R(c) : 1 \}$ \;
	Set $x_i := x_i + 1$ and $y_i := y_i + 1$ \;
\end{algorithm}

\subsection{Shift Operation}
\label{sec:shiftOperation}

The insertion of a set of items $S$ into containers
of a suitable group $g$ may violate \invE.
In this case, the \textsc{Shift} operation modifies the container
assignment such that \invA to \invE are fulfilled.

The easy case is when $h(g) + h(S)$ does not exceed the upper bound
$(h_B-1) K_g $ from \invE. Then, all items in $S$ can
be packed into appropriate containers of $g$.
This can be easily seen with the following indirect proof:
Assume that item $i \in S$ can not be placed. Then, each of
the $K_g$ containers is filled with items of total height 
greater than $h_B - 1$.
Thus, $h(g)+h(S) > K_g(h_B-1)$, which contradicts \invE.

Now assume that the insertion of items from $S$ exceeds
the upper bound by $\Delta>0$ and therefore violates \invE.
Basically there are two ways to deal with this situation:

\begin{itemize}
	\item Except for flexible groups, $K_g=2^l k$ or $K_g = 2^l (k-1)$.
	That is, we can not open new containers and thus
	have to remove items of group $g$ to make room for $S$. 
	When the removed items are the widest of the group $g$,
	they can be assigned to the group $\mathit{left}(g)$ while maintaining
	the sorting order \invB.
	\item The items in $S$ can be placed into new containers if $g$ is flexible
	and $K_g$ not at its upper bound. This procedure is also necessary if there
	is no group to the left of $g$.
	
\end{itemize}
The first of this two \textit{shift modes} is called \textit{left group}
and the second \textit{new container}.
Further details on both shift modes are given in the following.

\paragraph*{Mode: Left group}
First, we analyze in which cases we can proceed like that.
According to invariant \invC-\invD, there are two flexible groups, 
namely $(l,A,0)$ and $(l,B,q(l,B))$.
In those groups, the shift of items to the left group shall only be performed
if a new container can not be opened, \ie $K_g=2^l k$ for block $A$ or
$K_g=2^l (k-1)$ for block $B$ \invC-\invD.
Furthermore, this shift mode shall also not be used if $g=(l,A,-1)$,
that is, $S$ contains items that were shifted out from $(l,A,0)$ in the
previous call of shift. 
Thus the condition for the shift mode \textit{left group} is:
$$
\left( g=(l,A,0) \Rightarrow K_g = 2^l k \right) 
\wedge  \left( g=(l,B,q(l,B)) \Rightarrow K_g = 2^l (k-1) \right) 
\wedge  g \neq (l,A,-1) \,.
$$

Now, we turn to the changes in the packing.
In order to insert the items in $S$, we choose a set of items
$S_{out} \subset \ILg$ and move them to the group $\mathit{left}(g)$.
Since the sorting of items over the groups \invB must be maintained,
$S_{out}$ contains widest items of the group $g$.
To keep the amount of shifted items small, 
$S_{out}$ is chosen such that $h(S_{out})$ is minimal 
but \invE is fulfilled again.
\textsc{WidestItems}$(\ILg \cup S,\Delta)$ (see
Section\nobreakspace \ref {sec:auxFunctions}) is designed exactly for this purpose.
Now, we can remove the items $S_{out}$ from the containers
and close gaps in the stacks by \textsc{Sink}. 
Thus there is enough room to place $S$ (note that overlaps are
resolved by \textsc{Stretch}, which is part of the \textsc{Place} operation).	
The shifting process continues with 
\textsc{Shift}$(\mathit{left}(g), S_{out})$ 
in order to insert the shifted out items.

\paragraph*{Mode: New container}
This mode is performed in all remaining cases, \ie if
$$
\left( g=(l,A,0) \wedge K_g < 2^l k \right)
\vee  \left( g=(l,B,q(l,B)) \wedge K_g < 2^l (k-1) \right)
\vee  ~ g=(l,A,-1) \,.
$$

In the first two cases we are allowed to open a new container
and therefore have a simple way to maintain \invE without violating
other invariant properties.
The last case $g=(l,A,-1)$ results from a shift out of the
group $(l,A,0)$ when $K_{(l,A,0)} = 2^l k$ holds.
Then, the newly opened container builds the new leftmost group in
block $A$, temporarily called $(l,A,-1)$.
In all cases, the residual items are packed into a new container $\bar{c}$
(it will be shown later that one container is enough).
The new container gets placed via \textsc{InsertContainer}.
Finally, a renaming of groups in block $A$ restores the notation
(first group has index $r=0$)

In Algorithm\nobreakspace \ref {alg:shift}, which shows the entire shift algorithm, we 
require that the group $g$ is \textit{suitable} for the set of items $S$, according
to the following definition:

\begin{definition}[Suitable group]
	\label{def:suitableGroup}
	For a group $g$, let
	$\wmin(g)$ resp. $\wmax(g)$ denote the width of an item with 
	minimal resp. maximal width in $\ILg$.
	Set $w_{min}(\mathit{left}((l,A,0))) = \infty$
	and $w_{max}(\mathit{right}((l,B,q(l,B)))) = 0$.
	Group $g=(l,X,r)$ is \textit{suitable} for a new item $i$ if
	$w(i) \in (2^{-(l+1)}, 2^{-l}]$,
	$\wmin(\mathit{left}(g))  \geq w(i)$, and
	$\wmax(\mathit{right}(g)) < w(i)$.
\end{definition}

By the conditions from Definition~\ref{def:suitableGroup}, $i$ is inserted into the correct width category 
\invA and maintains the sorting over the groups \invB.

\begin{algorithm}
	\caption{\textsc{Shift}}
	\label{alg:shift}
	
	\SetAlgoLined
	\DontPrintSemicolon
	\SetKwInOut{Input}{Input}
	\SetKw{Or}{or}
	\SetKwFunction{Shift}{Shift}
	\SetKwFunction{Place}{Place}
	\SetKwFunction{WidestItems}{WidestItems}
	\SetKwFunction{Align}{Align}
	\SetKwFunction{Sink}{Sink}	
	\SetKwFunction{InsertContainer}{InsertContainer}	
	
	\Input{Group $g \in G$ \\ 
		Items $S \subset I_L$, where $g$ is suitable for each $i \in S$ (Def.~\ref{def:suitableGroup})}
	
	\BlankLine
	
	$\Delta = h(g) + h(S) - (h_B-1) K_g$ \;
	\uIf(\tcp*[f]{No violation of invariant \invE})
	{$\Delta \leq 0$ 
		\label{line:shiftConditionDirectPlacing} }{
		\Place{$g$,$S$} \;
	}
	\Else{
		
		\uIf{shift mode = left group}{
			\label{line:beginModeLeft}
			
			$S_{out}$ = \WidestItems{$\ILg \cup S$, $\Delta$} 
				\label{line:shiftWidest} \;
			Remove $S_{out}$ from $g$ \;
			\Sink{$c_j$} \tcp*[f]{For all affected containers $c_j$} \;
			\Place{$g$,$S$} 
				\label{line:shiftPlace} \;
			\Shift{$\mathit{left}(g), S_{out}$} \;
		}
		
		\uElseIf{shift mode = new container} {
			\label{line:beginModeCont}
			Let $\bar{c}$ be new container of $g$ \;
			\Place{$g$,$S$} \;
			\InsertContainer{$\bar{c}$} \;
			
			Rename groups in block $A$ 
			\tcp*[f]{First group gets index 0}\;
			\label{line:endModeCont}
		}
	}
\end{algorithm}

\subsubsection{Sequence of Shift Operations}
\label{sec:shiftSequence}

\begin{figure}
	\begin{center}
		\includegraphics[width=\textwidth]{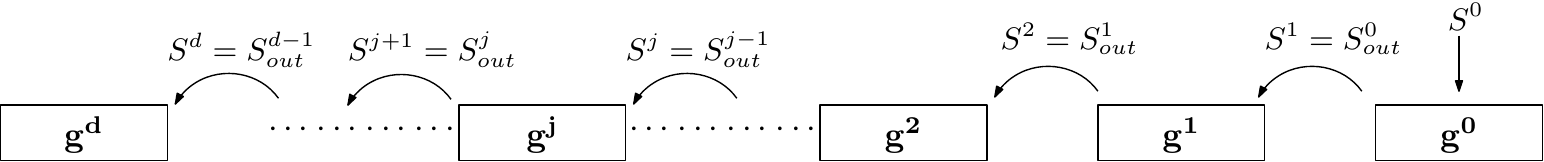}
		\caption{Shift sequence considered in Lemma\nobreakspace \ref {lemma:hSout} and\nobreakspace Corollary\nobreakspace \ref {lemma:1-2bins}}
		\label{fig:shiftTopLevel}
	\end{center}
\end{figure}

According to Algorithm\nobreakspace \ref {alg:shift} a shift operation
can end with a recursive call \textsc{Shift}($\mathit{left}(g), S_{out}$).
Note that, due to the procedure \textsc{WidestItems}, $h(S_{out}) > h(S)$ can hold.
This way a sequence of shift operations can occur, where the height of shifted
items grows in each part of the sequence.
We consider the shift sequence 
$$
\textsc{Shift}(g^0,S^0), \textsc{Shift}(g^1,S^1), \ldots, 
\textsc{Shift}(g^{d-1},S^{d-1}), \textsc{Shift}(g^d,S^d)
$$
(see Figure\nobreakspace \ref {fig:shiftTopLevel}) and
denote the values of $S$, $S_{out}$, and $\Delta$ in 
the call $\textsc{Shift}(g^j,S^j)$ by $S^j$, $S_{out}^j$, and $\Delta_j$.
The next lemma states that the total height of the shifted out items $S_{out}^j$ 
grows linearly in $j$, the position in the shift sequence.

\begin{lemma}
	\label{lemma:hSout}
	For any $j$ with $0 \leq j \leq d$ in the above defined shift sequence,
	$h(S_{out}^j) \leq h(S^0) + j + 1 \,. $
\end{lemma}
\begin{proof}
	First note that $\Delta_j \leq h(S^j)$ holds for each $j$ by \invE .
	Further, the function \textsc{WidestItems}($\cdot,\Delta_j)$ returns a 
	set $S_{out}^j$ with $h(S_{out}^j) < \Delta_j + 1$.
	For $j=0$ it holds that 
	$ h(S_{out}^0) < \Delta_0 + 1 \leq h(S^0) + 1 $.
	Now suppose $h(S_{out}^j) \leq h(S^0) + j + 1$ for some $j \geq 0$.	
	Note that $S^{j+1}=S_{out}^j$, thus
	$
	h(S_{out}^{j+1}) < \Delta_{j+1} + 1 \leq h(S^{j+1}) + 1
	= h(S_{out}^j) + 1$.
	By assumption, $h(S_{out}^j) \leq h(S^0) + j + 1$ and
	thus $h(S_{out}^{j+1}) \leq h(S^0) + (j + 1) + 1$.
\end{proof}

\begin{corollary}
	\label{lemma:1-2bins}	
	In the shift sequence defined above, it holds that 
	$h(S^d) \leq h(S^0) + h_B - 2$.
\end{corollary}

\begin{proof}
	The maximal total height of $S^d$ occurs in the longest possible shift
	sequence, that is, when $g^0 = (l,B,q(l,B))$, $g^{d-1}=(l,A,0)$, and 
	$g^d = (l,A,-1)$ hold.
	Note that $S^d=S_{out}^{d-1}$.
	
	We have $d=q(l,A)+q(l,B)+2$ for some $l \in W$ and with Lemma\nobreakspace \ref {lemma:NumberOfGroups} 
	(Equations\nobreakspace \textup {(\ref {eq:QlaQlbCategory})} and\nobreakspace  \textup {(\ref {eq:boundQLAQLB})}) 
	$d \leq 16 \omega / \epsilon + 2\omega$.
	Lemma\nobreakspace \ref {lemma:hSout} gives	
	$
	h(S_{out}^{d-1})
	\leq h(S^0) + (d-1) + 1
	\leq h(S^0) + 16 \omega / \epsilon + 2 \omega
	$.
	It remains to show $16 \omega / \epsilon + 2 \omega \leq h_B$, which follows 
	by definition of $\omega$ and $h_B$.
\end{proof}

Note that \MakeUppercase Corollary\nobreakspace \ref {lemma:1-2bins} implies that shifting a single item 
$i_t$ into a group, we have 
$h(S^d) \leq h(i_t) + h_B - 2 \leq h_B - 1$.
Therefore, one container is enough to pack all items in $S^d$ arriving in $(l,A,-1)$.

\subsection{ShiftA}
\label{sec:shiftA}
In Section~\ref{sec:dynBlockBalancing} we need an operation which moves one group from block 
$B$ to block $A$ within one category $l \in W$. This \textsc{ShiftA} operation, 
closely related to the dynamic rounding technique adapted from \cite{fullyDynamicBP}, 
is considered in the following.

The characteristic property of groups in block $B$ is the number of containers
$2^l (k-1)$, while groups in block $A$ have $2^l k$ containers (invariant properties \invC - \invD,
except for flexible groups).
\textsc{ShiftA} enlarges the group $(l,B,0)$ by $2^l$ additional
containers such that it can act as the new $(l,A,q(l,A))$ group.
To fulfill \invE, widest items (similar to \textsc{Shift}) are moved
from group $(l,B,1)$ to $(l,B,0)$ such that \invE is fulfilled for the 
group $(l,B,0)$ with $2^l$ additional containers. As items are taken
from $(l,B,1)$, widest items from $(l,B,2)$ have to be shifted to
$(l,B,1)$ and so on.

It might be the case that the removal of items from the last group 
$(l,B,q(l,B))$ leads to a violation of \invE since there are too many
containers for the total height of the residual items. 
Then, some containers have to be emptied and removed. 
At the extreme, all containers of the group get removed
and thus $(l,B,q(l,B)-1)$ becomes the new group $(l,B,q(l,B))$ via
renaming.

Note that each of the $2^l$ new containers of the new group $(l,A,q(l,A))$ has
width at most $2^{-l}$ and thus they can be placed in one level of height $h_B$ 
in the strip.
Algorithm\nobreakspace \ref {alg:shiftA} shows the steps explained above.
Thereby, two auxiliary algorithms for the insertion and removal of containers
(Algorithms\nobreakspace  \ref {alg:insertContainerShiftA} to\nobreakspace  \ref {alg:deleteContainerShiftA} ) occur, 
which are given in the following subsection.

\begin{algorithm}
	\caption{ShiftA}
	\label{alg:shiftA}
	
	\SetAlgoLined
	\DontPrintSemicolon
	\SetKwInOut{Input}{Input}
	\SetKw{Or}{or}
	\SetKwFunction{Shift}{Shift}
	\SetKwFunction{Insert}{Insert}	
	\SetKwFunction{shiftA}{shiftA}
	\SetKwFunction{Place}{Place}
	\SetKwFunction{WidestItems}{WidestItems}	
	\SetKwFunction{Align}{Align}
	\SetKwFunction{Sink}{Sink}	
	
	\Input{Width category $l \in W$}
	\BlankLine
	
	Let $g_i = (l,B,i)$ for $0 \leq i \leq q(l,B)$ \;
	Let $c_1, \ldots c_{2^l}$ be new containers of group $g_0$ \;
	
	\For{i=0 \ldots q(l,B)-1}{	
		$u = 
		\begin{cases} 
		2^l k (h_B - 1)	    & \mbox{if } i = 0 \\
		2^l (k-1) (h_B - 1)	& \mbox{if } i > 0
		\end{cases}$ 
		\label{line:shiftAdefU} \;
		$s_{min} = u - h(g_i) - 1$ \;
		
		$S$ = \WidestItems{${I_L}^{g_{i+1}}$, $s_{min}$} 
		\label{line:shiftAdefS} \;
		Remove $S$ from $g_{i+1}$ \;
		\Sink{$c_j$} \tcp*[f]{For all affected containers $c_j$} \;
		\Place{$S$,$g_i$} 
		\label{line:shiftAPlace} \;
	}
	Insert containers $c_1, \ldots c_{2^l}$ via Algorithm\nobreakspace \ref {alg:insertContainerShiftA}\;
	Rename $g_0$ to $(l,A,q(l,A)+1)$ \;
	
	\tcp*[l]{If necessary, remove containers from $(l,B,q(l,B))$}
	$\Delta = (K_{g_{q(l,B)}} - 1) (h_B - 1) - h(g_{q(l,B)})$ \label{line:removeQLBContainer}\;
	\uIf{$\Delta > 0$}{
		Empty $\ceil{\frac{\Delta}{h_B - 1}}$ containers 
		\label{line:emptyContainers}\;
		Remove each of them via Algorithm\nobreakspace \ref {alg:deleteContainerShiftA}\;
	}	
	
	Renaming\;
\end{algorithm}

The next lemma states how many containers have to be removed from a 
flexible group when the total height of items falls below the lower bound of \invE.
Note that 
Algorithm\nobreakspace \ref {alg:shiftA} behaves exactly accordingly to the lemma
in Line\nobreakspace \ref {line:emptyContainers}.

\begin{lemma}
	\label{lemma:tooManyContainers}
	Assume that $h(g) = (K_g - 1) (h_B - 1) - \Delta$ for a group $g$ and $\Delta>0$.
	Removing $\ceilS{\frac{\Delta}{h_B - 1}}$ containers restores \invE.
\end{lemma}

\begin{proof}
	Let $\delta = \ceil{\frac{\Delta}{h_B - 1}}$.
	It is easy to show that $h(g) \geq (K_g - \delta - 1) (h_B - 1)$.
\end{proof}

\subsubsection{Insertion and Deletion of Containers}

As a result of the \textsc{ShiftA} procedure, for a category $l$ 
there are $2^l$ containers that have to be inserted into the packing.
Instead of using Algorithm\nobreakspace \ref {alg:insertContainer} for each single container,
we rather use a slightly modified algorithm presented below.
Since all containers have rounded width $w^R(c) \leq 2^{-l}$, they fit into one level
of the strip. Hence, one single call of \textsc{Improve}, 
followed by a change of the LP/ILP-solution, is enough.
Algorithm\nobreakspace \ref {alg:insertContainerShiftA} shows the steps for the insertion
of $2^l$ new containers.

\begin{algorithm}
	\caption{Insertion of containers $c_1, \ldots, c_{2^l}$ for \textsc{ShiftA}}
	\label{alg:insertContainerShiftA}
	
	\SetAlgoLined
	\DontPrintSemicolon
	\SetKwInOut{Input}{Input}
	\SetKwFunction{Shift}{Shift}
	\SetKwFunction{Improve}{Improve}	
	
	\Input{$x,y$ fractional and integral solutions to LP($\CR{\con}{R}$)
	}
	\BlankLine
	
	\Improve{1,$x$,$y$} \;
	Let $P_i$ the pattern s.t.
	$P_i = \{ w^R(c_1) : 1, w^R(c_2) : 1 , \ldots, w^R(c_{2^l}) : 
	1\}$ \;
	Set $x_i := x_i + 1$ and $y_i := y_i + 1$ \;
\end{algorithm}

Furthermore, containers of group $(l,B,q(l,B))$ may get deleted
at the end of the \textsc{ShiftA} algorithm.
Analogously to the insertion of containers, 
modifications in the LP/ILP-solutions are necessary to reflect
the change of the packing. See Algorithm\nobreakspace \ref {alg:deleteContainerShiftA}.

\begin{algorithm}
	\caption{Deletion of a container $c$ for \textsc{ShiftA}}
	\label{alg:deleteContainerShiftA}
	
	\SetAlgoLined
	\DontPrintSemicolon
	\SetKwInOut{Input}{Input}
	\SetKwFunction{Shift}{Shift}
	\SetKwFunction{Improve}{Improve}	
	
	\Input{$x,y$ fractional and integral solutions to LP($\CR{\con}{R}$)
	}
	\BlankLine
	
	Let $P_i$ a pattern containing $w^R(c)$\;
	Let $P_j = P_i \setminus w^R(c)$ be the pattern without $w^R(c)$. \;
	Set $x_i := x_i - 1$ and $y_i := y_i - 1$ \;
	Set $x_j := x_j + 1$ and $y_j := y_j + 1$ \;
\end{algorithm}

\subsection{Operations Maintain Properties}
\label{sec:opsMaintainProps}
In this section we show that the operations
\textsc{Shift} and \textsc{ShiftA} maintain all invariant properties
(Lemma\nobreakspace \ref {lemma:opsMaintainInvariant}).
Further, in Lemma~\ref{lemma:opsMaintainFeasibleXY} it is shown that
also the LP/ILP-solutions (modified by
Algorithms\nobreakspace \ref {alg:insertContainer},  \ref {alg:insertContainerShiftA}, and\nobreakspace  \ref {alg:deleteContainerShiftA}) stay feasible.

\begin{lemma}
	\label{lemma:opsMaintainInvariant}
	Assume that the \invA - \invE are fulfilled by a container assignment 
	$\con$ and a rounding function $R$.
	Applying one of the operations \textsc{Shift}($g$,$S$) for any $g \in G$
	and $S$ with $h(S) \leq h_B - 1$ and $g$ is suitable for all $i \in S$, and
	\textsc{ShiftA}($l$) for any $l \in W$ 
	maintains \invA - \invE.
\end{lemma}

\begin{proof} 
	Implicitly, the operations modify both functions:
	The container assignment is changed in the \textsc{Shift} and 
	\textsc{ShiftA} algorithm due to the removal or placing of item sets.
	The rounding function changes when new containers get assigned.
	
	\paragraph*{Shift} 
	Let $g=(l,X,r)$.
	Each item assigned to group $g$ by \textsc{Shift} is suitable for $g$ or comes from
	the group $\mathit{right}(g)$. In the first case, \invA - \invB hold by Definition~\ref{def:suitableGroup},
	in the latter case by the fact that only widest items are moved to the left within the same category.
	
	The number of containers only changes in the
	second shift mode (Lines\nobreakspace  \ref {line:beginModeCont} to\nobreakspace  \ref {line:endModeCont} ).
	The preconditions of the respective shift mode
	ensure that either an additional
	container does not violate \invC - \invD since $g$ is a flexible group,
	or the new container belongs to group $(l,A,-1)$. After the renaming
	this group becomes $(l,A,0)$, has only one container, and thus fulfills \invC.
	In the latter case groups in block $B$ remain unchanged, 
	thus in all cases \invC - \invD are maintained.
	
	Finally, the algorithm is designed to maintain property \invE by
	shifting items of appropriate height.
	If $\Delta \leq 0$,
	then the new total height of items in group $g$ is
	$h(g) + h(S) = \Delta + (h_B - 1) K_g \leq (h_B - 1) K_g$.
	On the other hand, 
	$h(g) +h(S) \geq (K_g - 1) (h_B - 1)$ holds by the assumption that
	\invE is fulfilled beforehand and $h(S) \geq 0$.
		Now assume $\Delta > 0$. We analyze both shift modes separately.
	
	\subparagraph*{Mode: Left group}
	We have to show that the total height of items after the removal of $S_{out}$
	and insertion of $S$ lies in the interval $ [ (h_B - 1) (K_g - 1), (h_B - 1) K_g]$.
	Recall $\Delta = h(g) + h(S) - (h_B - 1) K_g$. It holds $h(S_{out}) \geq \Delta$ and thus
	$$h(g) - h(S_{out}) + h(S) 
	\leq h(g) - \Delta + h(S) 
	= (h_B - 1) K_g
	\,.
	$$
	On the other side, $h(S_{out}) < \Delta + 1$ and thus
	$$
	h(g) - h(S_{out}) + h(S) 
	>    h(g) - \Delta - 1 + h(S) 
	=    (h_B - 1) K_g - 1
	\geq (h_B - 1) (K_g - 1) \,,
	$$
	where the last inequality follows by $h_B \geq 2$.
	Hence, property \invE is fulfilled.
	
	\subparagraph*{Mode: New container}
	In the case that a new container gets inserted into one of the flexible groups
	$(l,A,0)$ or $(l,B,q(l,B))$, the number of containers $K_g$ is changed
	to $K_g' = K_g + 1$. We have to show that the new total height
	of items per group $h(g) + h(S)$ lies in the interval
	$[(K_g' - 1) (h_B - 1), K_g' (h_B - 1)]$.
	Since $\Delta>0$, it holds that 
	$h(g)+h(S) \geq K_g (h_B - 1) = (K_g' - 1) (h_B - 1)$.
	By assumption, $h(g) \leq K_g (h_B - 1)$ and $h(S) \leq h_B - 1$ hold,
	and thus $h(g)+h(S) \leq K_g (h_B - 1) + (h_B - 1) = (K_g + 1) (h_B - 1)
	= K_g' (h_B - 1)$.
	
	The remaining case is $g=(l,A,-1)$. Then, we insert a single container
	into the empty group $g$ and have to show
	$h(S) \in [(K_g' - 1) (h_B - 1), K_g' (h_B - 1)]$.
	Since $K_g'=1$, the lower bound holds obviously for any $h(S)\geq 0$. 
	Again by assumption, $h(S) \leq h_B - 1 = K_g' (h_B - 1)$.
	
	\paragraph*{ShiftA}
	Since items are moved between groups of the same category, \invA can not be
	violated. Let $g_i = (l,B,i)$ for $0 \leq i \leq q(l,B)$.
	Again, widest items are removed from one group $g_{i+1}$ and
	then assigned to group $g_i$, where $g_i = \mathit{left}(g_{i+1})$ and therefore
	\invB is maintained.
	The number of containers of group $g_0$ is enlarged by $2^l$ to 
	$2^l (k-1) + 2^l = 2^l k$. Hence, $g_0$ can be moved
	to block $A$ afterwards while maintaining \invC.
	In block $B$, the number of containers per group is not changed, except
	for $g_{q(l,B)}$, whose number of containers might be decreased. 
	But since $(l,B,q(l,B))$ is a flexible group, \invD holds anyway.
	
	The crucial part is again to show that \invE is fulfilled after the 
	\textsc{ShiftA} operation.
	The new total height of items in group $g_i$ 
	is $h'(g_i) = h(g_i) + h(S)$ and due to the procedure \textsc{WidestItems}
	we have
	$ h'(g_i) =    h(g_i) + h(S) 
	\leq h(g_i) + s_{min} + 1
	=    u $
	and		   
	$ h'(g_i) =    h(g_i) + h(S) 
	\geq h(g_i) + s_{min}
	=    u - 1 
	$.
	
	In the case $i=0$, we have to show 
	$(2^l k - 1) (h_B - 1) \leq h'(g_0) \leq 2^l k (h_B - 1)$
	as $g_0$ shall act as new $(l,A,q(l,A))$ group.
	By definition of $u$,
	it holds that $h'(g_0) \leq u = 2^l k (h_B - 1)$
	and $h'(g_0) \geq u - 1 = 2^l k (h_B  - 1) - 1 \geq (2^l k - 1) (h_B  - 1)$,
	where the last inequality follows by $h_B - 1 \geq 1$.
	
	For $1 \leq i < q(l,B)$, the value of $h'(g_i)$ has to be in the interval
	$[(2^l (k-1) - 1) (h_B - 1), 2^l (k-1) (h_B - 1)]$. We get similarly
	$h'(g_i) \leq u = 2^l (k-1) (h_B - 1)$ and again with $h_B - 1 \geq 1$
	it follows that
	$h'(g_i) \geq u - 1 = 2^l (k-1) (h_B - 1) - 1 \geq (2^l (k-1) - 1) (h_B - 1)$.
	
	It remains to analyze the case $i=q(l,B)$.
	Items from group $g_{q(l,B)}$ are shifted to group $g_{q(l,B)-1}$ in order
	to fulfill \invE for $g_{q(l,B)-1}$. Therefore, the loss of total height in group 
	$g_{q(l,B)}$ is at most
	\begin{align}
	\label{eq:shiftALossQLB}
	s_{min} + 1 &=    u - h(g_{q(l,B)-1}) \nonumber \\
	&\leq 2^l k (h_B - 1) - h(g_{q(l,B)-1})  \nonumber \\
	&\leq 2^l k (h_B - 1) - (2^l (k-1) - 1) (h_B - 1) \nonumber \\
	&=    (2^l + 1) (h_B - 1)
	\end{align}
	which can not be compensated by another shift, since $(l,B,q(l,B))$ has no right
	neighbor. Instead, decreasing the number of containers will repair property \invE.
	Maybe, the loss of items can be compensated (partially) because
	the previous total height $h(g)$ was greater than 
	$(K_{g_{q(l,B)}} - 1)(h_B - 1)$, the lower bound of \invE.
	For this purpose, define the actual amount of height $\Delta$ that
	is required to fulfill \invE:
	Let $\Delta = (K_{g_{q(l,B)}} - 1) (h_B - 1) - h'(g_{q(l,B)})$. 
	We have
	\begin{align*} 
	h'(g_{q(l,B)})
	&\geq h(g_{q(l,B)}) - (s_{min} + 1)  \\
	&\geq (K_{g_{q(l,B)}} - 1) (h_B - 1) - (2^l + 1)(h_B - 1) & \text{eq.\nobreakspace \textup {(\ref {eq:shiftALossQLB})}}  \\
	&=    (K_{g_{q(l,B)}} - 2^l - 2) (h_B - 1)
	\end{align*}
	and therefore
	\begin{align*}
	\Delta &=    (K_{g_{q(l,B)}} - 1) (h_B - 1)  - h'(g_{q(l,B)}) \nonumber \\
	&\leq (K_{g_{q(l,B)}} - 1) (h_B - 1) - 
	(K_{g_{q(l,B)}} - 2^l - 2) (h_B - 1) \nonumber \\
	&=    (K_{g_{q(l,B)}} - K_{g_{q(l,B)}} + 2^l + 1) (h_B - 1) \nonumber \\
	&=    (2^l + 1) (h_B - 1)
	\,.
	\end{align*}
	Clearly, if $\Delta \leq 0$, nothing has to be done since \invE is already fulfilled.
	Assuming $\Delta > 0$, by Lemma\nobreakspace \ref {lemma:tooManyContainers} the removal of
	
	\begin{equation}
	\label{eq:delContainersToRepairE}
	\left \lceil \frac{\Delta}{h_B - 1} \right \rceil
	\leq \left \lceil \frac{(2^l + 1) (h_B - 1)}{h_B - 1} \right \rceil
	=    \left \lceil 2^l + 1 \right \rceil
	\end{equation}
	containers is enough.
	Since Algorithm\nobreakspace \ref {alg:shiftA} behaves exactly like this
	starting from Line\nobreakspace \ref {line:removeQLBContainer}, we can conclude that
	\invE is fulfilled for group $g_{q(l,B)}$ by adjusting the number of containers
	appropriately.
\end{proof}

\begin{lemma}
	\label{lemma:opsMaintainFeasibleXY}
	Assume that the \invA -  \invE are fulfilled by a container assignment 
	$\con$ and a rounding function $R$.
	Let $x$ and $y$ be corresponding LP/ILP-solutions.
	Applying one of the operations \textsc{Shift}, \textsc{ShiftA} 
	defines new functions $\con'$, $R'$ and new solutions $x', y'$.
	The solutions $x',y'$ are feasible for the LP defined on the new
	rounded container instance $C_{\con'}^{R'}$.
\end{lemma}

\begin{proof}
	Again, we look at each operation separately.
	\paragraph*{Shift (Mode: Left group)}		
	The conditions from Definition~\ref{def:suitableGroup} ensure that
	no item that is inserted into a container increases the rounded width of that container.
	Therefore, each container in $C_{\con'}^{R'}$ has
	a smaller or equal width than in $C_{\con}^{R}$,
	\ie $w^{R'}(c) \leq w^{R}(c)$.
	
	Since in this case the cardinalities of the rounding groups do not
	change (no container gets inserted or removed), the right hand side
	in LP \ref {lp:binpacking} does not change.
	All configurations of $C_{\con}^{R}$ can be transformed into
	feasible configurations of $C_{\con'}^{R'}$.
	
	\paragraph*{Shift (Mode: New container)}
	In this case, a new container $\bar{c}$ is placed in a new level.
	The algorithm \textsc{InsertContainer}
	reflects this action in the LP/ILP-solution:
	For the pattern $P_i$ that contains the rounded width of $\bar{c}$
	once the corresponding values $x_i$ and $y_i$ are increased by one.
	
	\paragraph*{ShiftA} The \textsc{ShiftA} algorithm can be seen as a series of
	shift operations where each operation affects two neighboring
	groups. The new containers get inserted via Algorithm\nobreakspace \ref {alg:insertContainerShiftA},
	the removal of containers is done by Algorithm\nobreakspace \ref {alg:deleteContainerShiftA}.
	Both algorithms are similar to \textsc{InsertContainer}.
	Therefore, the claim can be shown analogously to \textsc{Shift}.
\end{proof}

\subsection{Insertion Algorithms}

Before we can give the entire algorithms for the insertion of a big or flat
item, we have to deal with another problem:
Remember that the parameter $k$, which controls the group sizes by
\invC - \invD, depends on $\SIZE(I_L(t))$ and thus changes over time.
Hence, for this section we use the more precise notation 
$k(t) = \floor{ \frac{\epsilon}{4 \omega h_B} \SIZE(I_L(t)) }$. 
Let $\kappa(t) = \frac{\epsilon}{4 \omega h_B} \SIZE(I_L(t))$.

At some point $t+1$, the value of $k(t+1)$ will increase such that
$k(t+1) = k(t) + 1$. Obviously, we can not rebuild the whole container assignment 
to fulfill the new group sizes required by \invC - \invD according to $k(t+1)$.
Instead, the block structure is exactly designed to deal with this situation.
By adjusting the ratio of the block sizes of $A$ and $B$, the problem mentioned
above can be overcome.
This dynamic block balancing technique was developed in
\cite[Sec. 3.3]{fullyDynamicBP} and is described in the following section.

\subsubsection{Dynamic Block Balancing}
\label{sec:dynBlockBalancing}
All groups of block $A$ that fulfill invariant \invC - \invD with parameter $k(t)$ 
can act as groups of block $B$ with parameter $k(t)+1 = k(t+1)$. 
So assuming that block $B$ is empty, renaming block
$A$ into $B$ fulfills the invariant properties. 
Note that with \textsc{ShiftA} we have an operation 
that transfers a single group from block $B$ to $A$.

Therefore, we have to ensure that whenever the (integer) value of $k(t)$ is willing to 
increase (\ie $\kappa(t)$ has a fractional value close to one), 
the block $B$ is almost empty.
Let the \textit{block balance} be defined as $bb(t) = \frac{A(t)}{A(t)+B(t)}$,
where $A(t)$ and $B(t)$ denote the number of groups in the respective block 
summing over all categories. 
Note that $bb(t)$ equals one if and only if the $B$-block is empty.

Let $\fract{\kappa(t)} \in [0,1)$ be the fractional part of $k(t)$.
To adjust the block balance to the value of $\fract{\kappa(t)}$, 
we partition the interval $[0,1)$ into smaller intervals $J_i$:

$$ J_i = \left[ \frac{i}{A(t)+B(t)}, \frac{i+1}{A(t)+B(t)} \right) 
\hspace*{20pt} \text{for } i=0, \ldots, A(t)+B(t)-1 $$
Algorithm\nobreakspace \ref {alg:blockBalancing} shows the block balancing algorithm: A number $d$ of groups is moved 
from block $B$ to $A$ such that afterwards $bb(t) \in J_i \iff \fract{\kappa(t)} \in J_i$.
Hence, block $B$ is almost empty when $k(t)$ is willing to increase.

\begin{algorithm}
	\SetAlgoLined
	\DontPrintSemicolon
	\SetKwInOut{Input}{Input}
	\SetKw{Or}{or}
	\SetKwFunction{Shift}{Shift}
	\SetKwFunction{Insert}{Insert}	
	\SetKwFunction{ShiftA}{ShiftA}		
	
	\BlankLine
	Let $i$ s.t. $\fract{\kappa(t)} \in J_i$ \;
	Let $j$ s.t. $bb(t) \in J_j$ \;
	Set $d = i-j \bmod A(t)+B(t)$ \;
	\For{p=0 \ldots d-1}{
		\uIf(\tcp*[f]{Block $B$ is empty}){$j+p = A(t) + B(t)$ }{
			Rename Block $A$ to Block $B$ \;
		}
		\shiftA{$l$} \tcp*[f]{for any suitable category $l$} \;
	} 
	
	\caption{Block Balancing Algorithm}
	\label{alg:blockBalancing}
\end{algorithm}

\begin{lemma}
	\label{lemma:blockBalancingIsCorrect}
	Assume that $\fract{\kappa(t)} \in J_i$. 
	At the end of Algorithm\nobreakspace \ref {alg:blockBalancing} it holds that
	$bb^*(t) \in J_i$, where $bb^*(t)$ is the block balance after shifting $d$
	groups from $B$ to $A$.
\end{lemma}

\begin{proof}
	Let $J_j$ be the interval containing $bb(t)$, the block balance before
	the \textsc{ShiftA} operations. We have to show that after performing
	$d = i-j \bmod A(t)+B(t)$ \textsc{ShiftA} operations, it holds that
	$bb^*(t) \in J_i$.
	Each call of \textsc{ShiftA} moves one group from $B$ to $A$.
	Let $bb'(t)$ be the block balance after one single call of \textsc{ShiftA}.
	We have
	$$ bb'(t) = bb(t) + \frac{1}{A(t)+B(t)} \geq \frac{j}{A(t)+B(t)} + \frac{1}{A(t)+B(t)}
	= \frac{j+1}{A(t)+B(t)} $$
	and
	$$ bb'(t) = bb(t) + \frac{1}{A(t)+B(t)} < \frac{j+1}{A(t)+B(t)} + \frac{1}{A(t)+B(t)}
	= \frac{j+2}{A(t)+B(t)} \,.$$
	Hence, $bb'(t) \in J_{j+1}$, where $j+1$ can be seen as performed 
	modulo $A(t)+B(t)$: When $j+p = A(t) + B(t)$ for $0 \leq p \leq d-1$, the algorithm
	renames block $A$ to $B$. Note that block $B$ is empty at this time since
	$p$ calls of \textsc{ShiftA} were performed previously.
	After renaming, the $A$ block is empty and the block balance lies in the
	interval $J_0$.
	As $bb^*(t)$ denotes the block balance after $d$ calls of \textsc{ShiftA}, 
	the interval index of $bb^*(t)$ equals
	$j + d \bmod A(t)+B(t) = j + (i-j \bmod A(t)+B(t)) \bmod A(t)+B(t)
	= i \,.$
\end{proof}

For the migration analysis it is important
how many groups are shifted between blocks. The next lemma shows
that the total number of groups shifted for the insertion
of a set $M$ grows proportional with $\SIZE(M)$.

\begin{lemma}
	\label{lemma:movedGroupsBlockBalancing}
	
	Let $t, t'$ be two time steps and $M$ the set of items inserted in between,
	\ie $I_L(t') = I_L(t) \cup M$.
	Assume that each item $i \in M$ causes one call of Algorithm\nobreakspace \ref {alg:blockBalancing}
	and let $d_i$ be the parameter $d$ for item $i$.
	It holds that 
	$\sum_{i \in M} d_i \leq \frac{8 + \epsilon}{2 h_B} \SIZE(M) + 1$.
\end{lemma}

\begin{proof}
	By definition, $\kappa(t)$ grows linearly in $\SIZE(I_L(t))$, thus
	$\kappa(t') - \kappa(t) = \allowbreak \frac{\epsilon}{4 \omega h_B} \SIZE(M)$.
	We obtain an upper bound for the difference of the fractional parts:
	$
	\fract{\kappa(t')} - \fract{\kappa(t)} 
	\leq \frac{\epsilon}{4 \omega h_B}  \SIZE(M)
	$.
	
	By Lemma\nobreakspace \ref {lemma:NumberOfGroups}, at each time $t$ the total number
	of groups $A(t)+B(t)$ is at most 
	$\frac{16 \omega + 2 \omega \epsilon}{\epsilon}$.
	Hence, all intervals $J_i$ have the length 
	$\frac{1}{A(t)+B(t)} \geq \frac{\epsilon}{16 \omega + 2 \omega \epsilon}$.
	From Lemma\nobreakspace \ref {lemma:blockBalancingIsCorrect} we know that
	$bb(t) \in J_i \iff \fract{\kappa(t)} \in J_i$ after each run of Algorithm\nobreakspace \ref {alg:blockBalancing}.
	The total number of intervals that can lie between $\kappa(t')$ 
	and $\kappa(t)$ can thus be bounded as follows:
	$$ \sum_{i \in M} d_i 
	\leq \ceil{ \frac{\fract{\kappa(t')} - \fract{\kappa(t)}}
		{\frac{\epsilon}{16 \omega + 2 \omega \epsilon}} }
	\leq \ceil{ \frac{ \frac{\epsilon}{4 \omega h_B} \SIZE(M) }
		{\frac{\epsilon}{16 \omega + 2 \omega \epsilon}} }
	\leq    \frac{8 + \epsilon}{2 h_B} \SIZE(M) + 1
	$$ 
\end{proof}

\subsubsection{Insertion of Big Items}
\label{sec:insertBig}

The insertion algorithm for a big item $i_t$ given in
Algorithm\nobreakspace \ref {alg:insertBigItem} is very simple:
Basically, the insertion of item $i_t$ is done by the
\textsc{Shift} algorithm called with a suitable group.
Afterwards, the block balancing presented in Section\nobreakspace \ref {sec:dynBlockBalancing} is 
performed.

\begin{algorithm}
	\SetAlgoLined
	\DontPrintSemicolon
	\SetKwInOut{Input}{Input}
	\SetKw{Or}{or}
	\SetKwFunction{Shift}{Shift}
	\SetKwFunction{BlockBalancing}{BlockBalancing}
	\SetKwFunction{Insert}{Insert}	
	\SetKwFunction{shiftA}{shiftA}		
	
	\Input{Item $i_t \in I_L$}
	\BlankLine
	
	Find suitable group $g=(l,X,r)$ according to Definition~\ref{def:suitableGroup} \;
	\Shift{$g,\{i_t\}$} \; 
	\BlockBalancing \;
	
	\caption{Insertion of a big item}
	\label{alg:insertBigItem}
\end{algorithm}

It remains to give the \textsc{Place} algorithm which is used as a
subroutine in \textsc{Shift}. 
For each item $i$ to be inserted, Algorithm\nobreakspace \ref {alg:placeBig} looks for a container
where $i$ can be added without overfilling the container, places
the item on top of the stack and calls \textsc{Stretch} to resolve overlaps.

\begin{algorithm}
	\caption{Placing of big items}
	\label{alg:placeBig}
	
	\SetAlgoLined
	\DontPrintSemicolon
	\SetKwInOut{Input}{Input}
	\SetKw{Or}{or}
	\SetKwFunction{Shift}{Shift}
	\SetKwFunction{Stretch}{Stretch}
	
	\SetAlgoNoEnd%
	
	\Input{Group $g$, set of big items $S$}
	\BlankLine
	
	Let $H(c) = \sum_{i \in I_L : \con(i)=c} h(i)$ \\
	
	\For{$i \in S$}{
		Find container $c$ with $R(c)=g$ and $H(c) \leq h_B - h(i)$ \;
		Place $i$ on top of the stack for big items in $c$ \;
		\Stretch{u}, where $u$ is the level of $c$ \;
	}
\end{algorithm}

\subsubsection{Insertion of Flat Items}
\label{sec:insertFlat}

The main difficulty of flat items becomes clear in the following scenario:
Imagine that flat items of a group $g$ are elements of 
$S_{out} = \textsc{WidestItems}(g,\Delta)$ in a shifting process.
Remember that generally each container from which items are removed has to be
sinked, \ie at most $\card{S_{out}}$ containers. 
In case of big items, due to their minimum height $\epsilon$ we get 
$\card{S_{out}} \leq \floorS{\Delta / \epsilon}$.
In contrast, flat items can have an arbitrary small height and thus no such
bound is possible.
But \textsc{Sink} on all $K_g$ containers would lead to unbounded
migration (since $K_g$ depends on $\SIZE(I_L)$). 
Therefore, we aim for a special packing structure that avoids the above
problem of sinking too many containers.

Like shown in Figure\nobreakspace \ref {abs:fig:containerContent},
flat items build a sorted stack at the top of the container such that the least wide item is placed
at the top edge.
Thereby, widest items can be removed from the container without leaving a gap.
To maintain the sorting, we introduce a buffer for flat items called 
\textit{$F$-buffer}.
It is located in a rectangular segment of width 1 and height $\omega y$,
somewhere in the packing,
where $y = \card{G} \epsilon^2 + \epsilon$.
Note that the additional height for the F-buffer is
bounded by $\omega y = \bigO{\epsilon \left(\log 1 / \epsilon \right)^2}$.
The internal structure of the F-buffer is shown in Figure\nobreakspace \ref {fig:flatBuffer}:
For each category $l$ there are $2^l$ slots in one level of height
$y$. Items can be placed in any slot of their category.

\begin{figure}
	\centering
	\includegraphics[width=0.5\textwidth]{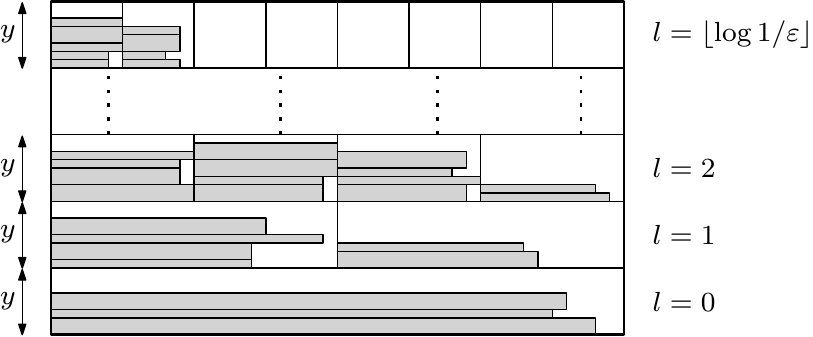}
	\caption{F-buffer contains $2^l$ slots of height $y$ for each category $l$.}
	\label{fig:flatBuffer}
\end{figure}

An incoming flat item may overflow the F-buffer, more precisely, the level of one category in the F-buffer.
For this purpose, Algorithm\nobreakspace \ref {alg:insertFlat} iterates over all groups 
$g_q, g_{q-1}, \ldots, g_0$ of this category, where $g_q$ is the rightmost
and $g_0$ the leftmost group\footnote{Note that the direction of the iterative shifting is crucial: 
	Calling \textsc{Shift} for a group $g$ may reassign items in all groups left
	to $g$. Therefore, iterating from \enquote{right to left} is necessary to guarantee
	that after shifting into group $g$, no group to the right of $g$ is suitable for a remaining
	item in $S$.
	In other words, with this direction one shift call for each group is enough,
	which is in general not true for the direction \enquote{left to right}.}. 
For each group, the set $S$ contains 
those items in the F-buffer for which $g$ is a suitable group.
The set $S$ is split into smaller subsets of total height at most 1, then each 
subset gets inserted via a single call of \textsc{Shift}.

\begin{algorithm}
	\caption{Insertion of a flat item}
	\label{alg:insertFlat}
	
	\SetAlgoLined
	\DontPrintSemicolon
	\SetKwInOut{Input}{Input}
	\SetKw{Or}{or}
	\SetKw{And}{and}
	\SetKwFunction{Shift}{Shift}
	\SetKwFunction{Stretch}{Stretch}	
	\SetAlgoNoEnd
	
	\Input{Flat item $i_t \in I_L$ of category $l$}
	\BlankLine
	
	\uIf{$i_t$ can be placed in the F-buffer}{
		Place $i_t$ in the F-buffer \;
	}
	\uElse{
		Let $B(l)$ be the set of items in the buffer slots of category $l$ \;
		\For{$j=q(l,A)+q(l,B)+1, \ldots,0$}{
			\label{line:shiftFlatGroupLoop}
			Let $g_j = \begin{cases} 
			(l,A,j)			  & j \leq q(l,A) \\
			(l,B,j-q(l,A)-1)  & j >    q(l,A)
			\end{cases}$ \;
			Let $S = \{ i \in B(l) \mid g_j \text{ is suitable for } i \}$ \;
			Let $S_1, \ldots, S_n$ be partition of $S$ with $h(S_r) \in (1-\epsilon,1]$
			for all $1 \leq r \leq n$.
			\label{line:partitionShiftFlat} \;
			\For{$r=1, \ldots, n$}{
				\Shift{$g_j$,$S_r$} \;
			}
			Remove $S$ from $B(l)$ \;
		}
		\BlockBalancing \;
	}
\end{algorithm}

Algorithm\nobreakspace \ref {alg:placeFlat} shows \textsc{Place}
for flat items, required for \textsc{Shift}: 
The set of flat items $S$ gets partitioned into sets $S_j$ each of total height at most 1.
Then, each set $S_j$ is
placed into a container analogously to Algorithm\nobreakspace \ref {alg:placeBig}.
Thereby, the stack of flat items needs to be resorted.

\begin{algorithm}
	\caption{Placing of flat items}
	\label{alg:placeFlat}
	
	\SetAlgoLined
	\DontPrintSemicolon
	\SetKwInOut{Input}{Input}
	\SetKw{Or}{or}
	\SetKw{And}{and}
	\SetKwFunction{Shift}{Shift}
	\SetKwFunction{Stretch}{Stretch}	
	\SetAlgoNoEnd
	
	\Input{Group $g$, set of flat items $S$}
	\BlankLine
	
	Let $H(c) = \sum_{i \in I_L: \con(i)=c} h(i)$ \;
	Let $S_1,\ldots,S_n$ be partition of $S$ with $h(S_r) \in (1-\epsilon,1]$
	for all $1 \leq r \leq n$. \;
	
	\For{$j=1,\ldots,n$}{
		Find container $c$ with $R(c)=g$ and $H(c) \leq h_B - h(S_j)$	\;
		Add $S_j$ to the stack of flat items and resort the stack \;
		\Stretch{u}, where $u$ is the level of $c$ \;
	}
\end{algorithm}

By choice of $y$, the slot height of the F-buffer, we get the following observation
required for a later proof:

\begin{lemma}
	\label{lemma:sizeShiftPartitionFlat}
	The for-loop in Line\nobreakspace \ref {line:partitionShiftFlat} of Algorithm\nobreakspace \ref {alg:insertFlat} iterates at most $\bigO{1 / \epsilon}$ times.
\end{lemma}
\begin{proof}
	The number of iterations is $n \leq \ceil{\frac{h(S)}{1-\epsilon}}$.
	Thus is is enough to show $h(S)  \leq \bigO{1  / \epsilon}$.
	Each category $l$ has $2^l$ slots of height $y = \card{G} \epsilon^2 + \epsilon$.
	By Lemma\nobreakspace \ref {lemma:NumberOfGroups} it follows $\card{G} \epsilon^2 
	\leq 16 + 2 \epsilon$ and thus $y \leq 17$.
	Further, the maximum slot number is $2^{\floor{\log 1 / \epsilon}} \leq 1 / \epsilon$, 
	therefore $h(S) \leq 17 / \epsilon$.
	%
\end{proof}

\subsubsection{Online-Approximation Guarantee}
\label{sec:improveApprox}

In Section~\ref{sec:opsMaintainProps} we shows that single \textsc{Shift}- and \textsc{ShiftA}-operations
maintain the invariant properties. This holds for Algorithms\nobreakspace \ref {alg:insertBigItem} 
and\nobreakspace  \ref {alg:insertFlat} as well. 

\begin{lemma}
	\label{lemma:insertionAlgsMaintainInv}
	Algorithms \nobreakspace \ref {alg:insertBigItem} and \ref{alg:insertFlat} maintain all invariant
	properties and the LP/ILP-solutions stay feasible.
\end{lemma}
\begin{proof}
	Both algorithms use the \textsc{Shift} operation to insert new items. 
	In Lemma\nobreakspace \ref {lemma:opsMaintainInvariant} we showed that this operation
	maintains the invariant properties if applied with a set $S$ of items with
	$h(S) \leq h_B - 1$. Thus, the claim follows by 
	Lemma\nobreakspace \ref {lemma:opsMaintainInvariant} if we can show the latter condition.

	We show that for the first call $\textsc{Shift}(g^0,S^0)$ it holds that
	$h(S^0) \leq 1$. Then, $h(S) \leq h_B - 1$ follows by Lemma\nobreakspace \ref {lemma:hSout}
	and \MakeUppercase Corollary\nobreakspace \ref {lemma:1-2bins}.
	In case of big items (Algorithm\nobreakspace \ref {alg:insertBigItem}), $S^0$ contains a single item,
	which has maximum height 1.
	For flat items (Algorithm\nobreakspace \ref {alg:insertFlat}), the partition of the set $S$ ensures
	that $\textsc{Shift}$ is only called with items of total height 1.
	
	Afterwards both algorithms call the block balancing Algorithm\nobreakspace \ref {alg:blockBalancing}.
	Here, the crucial operation is \textsc{ShiftA}, but 
	Lemma\nobreakspace \ref {lemma:opsMaintainInvariant} already gives the claim.
	Also by Lemma\nobreakspace \ref {lemma:opsMaintainInvariant}, all changes to the LP/ILP-solutions
	result in new feasible solutions.
\end{proof}

However, since the packing height can increase due new containers,
also the level assignment gets changed in order to fulfill the approximation guarantee. 
Here, the crucial operation is \textsc{Improve} called in Algorithms~\ref{alg:insertContainer},
\ref{alg:insertContainerShiftA}, and \ref{alg:deleteContainerShiftA}
which we analyze in the following.
The next theorem is a modified version of Theorem 3 in \cite{fullyDynamicBP}
and justifies that we can apply \textsc{Improve}
on suitable LP/ILP-solutions $x$, $y$ to obtain
a solution with reduced additive term.
The proof of Theorem\nobreakspace \ref {theo:improveApplication} is moved to Appendix\nobreakspace \ref {app:improve}.

For a vector $x$, let $\nnz{x}$ denote the number of non-zero components.
Let $\delta$ be a parameter specified later.
By \MakeUppercase Theorem\nobreakspace \ref {theo:approxGuarantee}, we have
$\OPT(\CR{\con}{R}) \leq (1+\epsilon') \OPT(I_L) + z$
for $\epsilon' = 4 \epsilon$ and $z = \bigO{1 / \epsilon^4}$.
Let $\Delta = (1 + \epsilon') (1+\delta) - 1$
and $m$ be the number of constraints in LP~\ref {lp:binpacking}. 
Since there is exactly one constraint for each occurring width (group), 
by Lemma\nobreakspace \ref {lemma:NumberOfGroups}
$m \leq \frac{16 \omega}{\epsilon} + 2 \omega$.
Finally, let $D = \Delta \OPT(I_L) + m + (1+\delta)z$.
With the above definitions, we have
\begin{align}
\label{eq:improveOptCrOptIl}
(1+\epsilon) \OPT(\CR{\con}{R}) \leq (1+\Delta) \OPT(I_L) + (1+\epsilon)z
\,.
\end{align}

\begin{theorem}
	\label{theo:improveApplication}
	Given a rounded container instance $\CR{\con}{R}$ and an LP defined 
	for $\CR{\con}{R}$, let $\delta > 0$, $\alpha \in \NN$, and $x$ be a fractional 
	solution of the LP with
	\begin{align}
	\norm{x} &\leq (1+\Delta) \OPT(I_L) + (1+\delta)z \label{eq:improve:theo3a} \,,\\
	\norm{x} &\geq 2 \alpha \left( \tfrac{1}{\delta}+1 \right) \,,
	\label{eq:improve:theo3b}\\
	\norm{x} &=    (1+\delta') \LIN(\CR{\con}{R}) 
	\hspace*{12pt} \text{for } \delta'>0 \,.
	\label{eq:improve:theo3c}
	\intertext{Let $y$ be an integral solution of the LP with}
	\norm{y} &\geq (m+2) \left( \tfrac{1}{\delta}+2 \right )  \,,
	\label{eq:improve:theo3d}\\
	\norm{y} &\leq (1+2\Delta) \OPT(I_L) + 2 (1+\delta)z + m \,.
	\label{eq:improve:theo3e}
	\end{align}
	Further, let $\nnz{x} = \nnz{y} \leq D$ and
	$x_i \leq y_i$ for all $1 \leq i \leq n$.

	Then, algorithm \textsc{Improve}($\alpha$) on $x$ and $y$ returns a new
	fractional solution $x'$ with
	$\norm{x'} \leq (1+\Delta) \OPT(I_L) + (1+\delta)z - \alpha$
	and also a new integral solution $y'$ with
	$\norm{y'} \leq (1+2\Delta) \OPT(I_L) + 2 (1+\delta) z + m - \alpha$.
	Further, $\nnz{x'}=\nnz{y'} \leq D$, and
	for each component we have $x'_i \leq y'_i$.
	The number of levels to change in order to obtain the new packing 
	corresponding with $y'$ is bounded by $\bigO{m / \delta}$.
\end{theorem}

\begin{theorem}
	\label{theo:heightContainerPackingAfterImprove}
	Assuming $\SIZE(I_L(t)) \geq \frac{4 \omega h_B}{\epsilon} (h_B + 1)$,
	Algorithms\nobreakspace \ref {alg:insertBigItem} and\nobreakspace  \ref {alg:insertFlat} are AFPTASs for the insertion of big 	
	and flat items with asymptotic ratio 
	$1+2\Delta$ for $\Delta \in \bigO{\epsilon}$.
\end{theorem}
\begin{proof}~
	\paragraph*{Approximation Guarantee}
	Let $\Delta, \epsilon', z, m$ be defined as above and set $\delta=\epsilon$. 
	Then, $\Delta = 4\epsilon + \epsilon + \epsilon^2 = \bigO{\epsilon}$.
	We show by induction that the algorithms maintain a packing of height $h$ with
	\begin{equation}
	\label{eq:heightContainerPackingAfterImprove}
	h \leq (1+2\Delta) \OPT(I_L) + 2(1+\epsilon)z + m \,.
	\end{equation}
	
	Suppose that the packing corresponds with $x$ and $y$, 
	which are fractional/integer solutions to $LP(\CR{\con}{R})$ such that
	$\norm{x} =    (1+\epsilon) \LIN(\CR{\con}{R})$ and
	$\norm{y} \leq (1+\epsilon) \OPT(\CR{\con}{R}) + m$.
	Further, let $\nnz{x} = \nnz{y} \leq m$ and $x_i \leq y_i$ for all $i \leq n$.
	Note that if $x$ is a basic solution for the LP relaxation 
	with accuracy $1+\epsilon$, a solution $y$ with the above properties can be derived 
	by rounding up each non-zero entry of $x$. 
	Since we obtain the container packing by an integral solution $y$ of the LP,
	the conclusion of Theorem\nobreakspace \ref {theo:improveApplication} implies 
	Equation\nobreakspace \textup {(\ref {eq:heightContainerPackingAfterImprove})}.
	Therefore, the remainder of the proof consists of two parts:
	\begin{enumerate}[label=(\roman*)]
		\item We show that Theorem\nobreakspace \ref {theo:improveApplication} is applicable for  
		$x, y$, and $\alpha=1$ (since \textsc{Improve} is applied only with $\alpha=1$).
		\label{item:improveProofStep1}
		\item We show that for all algorithms that modify the LP/ILP-solutions
		(Algorithms\nobreakspace \ref {alg:insertContainer},  \ref {alg:insertContainerShiftA}, and\nobreakspace  \ref {alg:deleteContainerShiftA})
		the returned solutions fulfill 
		the prerequisites of Theorem\nobreakspace \ref {theo:improveApplication}.
		\label{item:improveProofStep2}	
	\end{enumerate}
	
	\begin{description}
		\item[\ref {item:improveProofStep1}]
		By condition,
		$\SIZE(I_L(t)) \geq \frac{4 \omega h_B}{\epsilon} (h_B + 1)$.
		We show that this implies $\SIZE(I_L(t)) \geq h_B (m+2) (\frac{1}{\delta} + 2)$.
		The following equivalence holds using $\delta=\epsilon$:
		\begin{align*}
		\frac{4 \omega h_B}{\epsilon} (h_B + 1) 
		\geq h_B (m+2) \left(\tfrac{1}{\delta} + 2 \right)
		\Longleftrightarrow 4 \omega (h_B+1) &\geq (m+2) (1+2\epsilon)
		\end{align*}
		The second inequality is true since $4 (h_B + 1) \geq m+2$ and 
		$\omega \geq 1 + 2\epsilon$ for $\epsilon \leq 0.25$.
		Further, the following relation between $\SIZE(I_L)$, $\norm{x}$, and
		$\norm{y}$ holds:
		\begin{equation}
		\label{eq:relationSizes}
		\SIZE(I_L) \leq \SIZE(\C{\con}) \leq \SIZE(\CR{\con}{R}) \leq h_B \norm{x} \leq h_B \norm{y}
		\end{equation}
		
		With $\SIZE(I_L(t)) \geq h_B (m+2) \left( \frac{1}{\delta} + 2 \right)$ 
		and Equation\nobreakspace \textup {(\ref {eq:relationSizes})},
		$\norm{x} \geq (m+2) \left( \frac{1}{\delta} + 2 \right) 
		\geq 2  \left( \frac{1}{\delta} + 2 \right)$ and
		$\norm{y} \geq \norm{x} \geq (m+2) \left( \frac{1}{\delta} + 2 \right) $ follow.
		Hence, conditions \textup {(\ref {eq:improve:theo3b})} and \textup {(\ref {eq:improve:theo3d})}
		are fulfilled.
		
		Further, we have properties \textup {(\ref {eq:improve:theo3c})} 
		(with $\delta' = \epsilon$), \textup {(\ref {eq:improve:theo3e})}, and \textup {(\ref {eq:improve:theo3a})}:
		\begin{align*}
		\norm{x} &= (1+\epsilon) \LIN(\CR{\con}{R}) \leq (1+\epsilon) \OPT(\CR{\con}{R})
		\overset{\textup {(\ref {eq:improveOptCrOptIl})}}{\leq} (1+\Delta) \OPT(I_L) + (1+\epsilon) z \\
		\norm{y} &\leq (1+\epsilon) \OPT(\CR{\con}{R}) + m 
		\overset{\textup {(\ref {eq:improveOptCrOptIl})}}{\leq} (1+\Delta) \OPT(I_L) + (1+\epsilon) z + m
		\end{align*}
		
		The last two conditions of Theorem\nobreakspace \ref {theo:improveApplication} hold by assumption,
		thus it is applicable.		
		
		\item[\ref {item:improveProofStep2}]
		Since Algorithms\nobreakspace \ref {alg:insertContainer},  \ref {alg:insertContainerShiftA}, and\nobreakspace  \ref {alg:deleteContainerShiftA}
		all perform analog modifications on the LP/ILP-solutions, we can
		consider them at once. Let $x'$ and $y'$ be the solutions returned by 
		\textsc{Improve}(1,$x$,$y$) before the update of single components. 
		According to Theorem\nobreakspace \ref {theo:improveApplication},
		$\norm{x'} \leq (1+\Delta) \OPT(I_L) + (1+\delta)z - 1$ and 
		$\norm{y'} \leq (1+2\Delta) \OPT(I_L) + 2 (1+\delta) z + m - 1$.
		Further, $x'$ and $y'$ have the same number of non-zero entries and
		$x_i' \leq y_i'$ for each component.
		
		Now, let $x''$ and $y''$ be the solutions finally returned by the algorithms.
		We have: 
		\begin{align*}
		\norm{x''} &\leq \norm{x'} + 1 \leq (1+\Delta) \OPT(I_L) + (1+\delta)z  \,,\\
		\norm{y''} &\leq \norm{y'} + 1 \leq (1+2\Delta) \OPT(I_L) + 2(1+\delta)z + m 
		\,.
		\end{align*}
		Since $x_i'$ and $y_i'$ get simultaneously increased or decreased, 
		the properties $x_i'' \leq y_i''$ and $\nnz{x''} = \nnz{y''} \leq D$ hold.
		Also, since $\norm{x''} \geq \norm{x}$ and $\norm{y''} \geq \norm{y}$,
		eqs.\nobreakspace \textup {(\ref {eq:improve:theo3b})} and\nobreakspace  \textup {(\ref {eq:improve:theo3d})} are maintained.
		Finally, \textup {(\ref {eq:improve:theo3c})} holds, even if we have no real control 
		about $\delta'$.
		
	\end{description}
	
	\paragraph*{Running Time}
	The running time of the overall algorithm is clearly dominated by the operation
	$\textsc{Improve}$. Like in \cite{jansen2013robust} and \cite{fullyDynamicBP}
	we apply \textsc{Improve} on a LP with 
	$m=\bigO{\frac{\omega}{\epsilon}} = \bigO{\frac{1}{\epsilon} \log \frac{1}{\epsilon}}$
	many rows, where the number of non-zero-entries is bounded from above by $D$.
	Therefore, we obtain a running time polynomial in $1/\epsilon$ and
	$\card{I(t)}$, see \cite{jansen2013robust} for further details.	
\end{proof}

\section{Narrow Items}
\label{sec:narrowItems}

For narrow items we use the concept of \textit{shelf algorithms}
introduced by Baker and Schwarz~\cite{baker1983shelf}.
The main idea is to place items of similar height in a row.
For a parameter $\alpha \in (0,1)$ item $i$ belongs to \textit{group} 
$r \in \NN \setminus \{0\}$ 
if $h(i) \in [(1-\alpha)^r, (1-\alpha)^{r-1})$.
Narrow items of group $r$ are placed into a \textit{shelf of group} $r$,
which is a rectangle of height $(1-\alpha)^{r-1}$,
see Figure\nobreakspace \ref {fig:narrowItems}.
Analogously to \cite{baker1983shelf}, we say a shelf of width $w$ is \textit{dense} when it 
contains items of total width greater than $w-\epsilon$ and \textit{sparse} otherwise.

\begin{figure}
	\centering
	\includegraphics[width=0.45\textwidth]{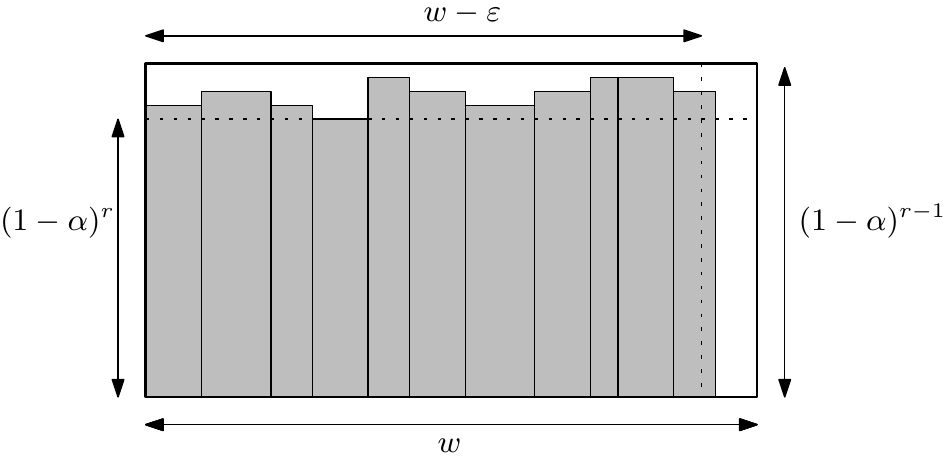}
	\caption{Shelf for narrow items of group $r$ (dense)}
	\label{fig:narrowItems}
\end{figure}

When the instance consists of narrow items only, the concept of shelf algorithms yields
an online AFPTAS immediately.
This is shown next in Section\nobreakspace \ref {sec:onlyNarrowItems}.
However, the goal is to integrate narrow items into the container packing introduced
in Section~\ref{sec:containerPacking}.
We show in Section~\ref{sec:narrowItemsCombination} how to fill gaps in the container packing with shelfs
of narrow items. This first-fit-algorithm for narrow items
maintains an asymptotic approximation ratio of $1+ \bigO{\epsilon}$, as
finally shown in Lemma\nobreakspace \ref {lemma:finalPackingHeight}.

\subsection{Narrow Items Only}
\label{sec:onlyNarrowItems}

Consider the following first-fit shelf algorithm:
Place an item of group $r$ into the first shelf of group $r$ where it fits.
If there is no such shelf, open a new shelf of group $r$ on top of the packing.

\begin{lemma}
	\label{lemma:onlyNarrowItems}
	The shelf algorithm with
	parameter $\alpha = \frac{\epsilon^2}{1-\epsilon^2}$ yields a 
	packing of height at most
	$(1+\epsilon) \OPT(I_N) + \bigO{1 / \epsilon^4}$.
\end{lemma}

\begin{proof}
	Let	$\beta_r$ the number of shelfs of group $r$ in the packing obtained 
	by the shelf algorithm.
	Further, for a group $r$ let 
	$I_r = \{i \in I \mid h(i) \in [(1-\alpha)^r, (1-\alpha)^{r-1}) \}$ .
	Each dense shelf for group $r$ contains items of size at least
	$(1-\alpha)^r (1-\epsilon)$, see Figure\nobreakspace \ref {fig:narrowItems}.
	Note that by the first-fit-principle, for each group at most one shelf is sparse.
	Thus there are at least $\beta_r - 1$ dense shelfs for each group $r$, hence
	$\SIZE(I_r) \geq (\beta_r - 1)(1-\alpha)^r(1-\epsilon)$,
	or equivalently
	\begin{align}
	\label{eq:shelfLoadOfGroup}
	\beta_r \leq \SIZE(I_r) (1-\alpha)^{-r} (1-\epsilon)^{-1} + 1 \,.
	\end{align}
	The packing consists of $\beta_r$ shelfs of height $(1-\alpha)^{r-1}$ for each
	group $r$ (set $\beta_r=0$ if the group does not exist).
	Therefore, the packing height is by Equation~\ref {eq:shelfLoadOfGroup}
	%

$$	\sum_{r=0}^\infty \beta_r (1-\alpha)^{r-1} 
	\leq \sum_{r=0}^\infty \left( \SIZE(I_r) (1-\alpha)^{-r} (1-\epsilon)^{-1} 
	+ 1 \right) (1-\alpha)^{r-1}  \,.$$
	By splitting the sum and moving constant factors in front we get that the last term
	equals
	$$\frac{1}{(1-\epsilon)(1-\alpha)} \sum_{r=0}^\infty \SIZE(I_r)
		+ \sum_{r=0}^\infty (1-\alpha)^{r-1} 
		\leq (1+\epsilon) \SIZE(I_N)
		+ \sum_{r=0}^\infty (1-\alpha)^{r-1} \,,
		$$
		where the inequality follows by definition of $\alpha$ and $I_r$.
		Since $\SIZE(I_N) \leq \OPT(I_N)$, the claim follows if we can show
		that the additive term is $\bigO{1 / \epsilon^4}$.
		This follows by the geometric series:
	$ \sum_{r=0}^\infty (1-\alpha)^{r-1}
	= \frac{1}{1-\alpha} \sum_{r=0}^\infty (1-\alpha)^r
	= \frac{1}{1-\alpha} \frac{1}{\alpha}
	= \bigO{1 / \epsilon^4}
	$.
\end{proof}

\subsection{Combination with Container Packing}
\label{sec:narrowItemsCombination}

As shown in Section\nobreakspace \ref {sec:onlyNarrowItems} shelfs are a good way to
pack narrow items efficiently. But before opening a new shelf that
increases the packing height, we have to ensure that the existing
packing is well-filled.
Therefore, the idea is to fill gaps in the container packing with shelfs of
narrow items. Thereby, we define a \textit{gap} is the rectangle of height $h_B$
that fills the remaining width of an aligned level.

To simplify the following proofs, we introduce artificial \textit{D-containers} 
filling the remaining width of a container level completely
and think of placing shelfs inside the D-containers.
We say that a D-container is \textit{full} if shelfs of total height
greater than $h_B-1$ are placed inside.
For this section, we call the containers for big and flat items 
(introduced in Section\nobreakspace \ref {sec:containerPacking}) \textit{C-containers} 
to distinguish them from D-containers that contain shelfs of narrow items.

A level is called \textit{well-filled} if the total width of containers
(including the D-container, if existing) is at least $1-2\epsilon$, 
and \textit{badly-filled} otherwise.
Figure\nobreakspace \ref {fig:fillingGapsTopLevel} shows a container packing filled with
D-containers: All levels are well-filled, assuming that the gray dashed
area is of total width less than $2 \epsilon$.
Note that a badly-filled level can be made well-filled
by aligning the C-containers with the \textsc{Align} operation
and then define a D-container of the remaining width.

\begin{figure}
	\centering
	\begin{subfigure}[t]{0.45\textwidth}
		\centering
		\includegraphics[height=6cm]{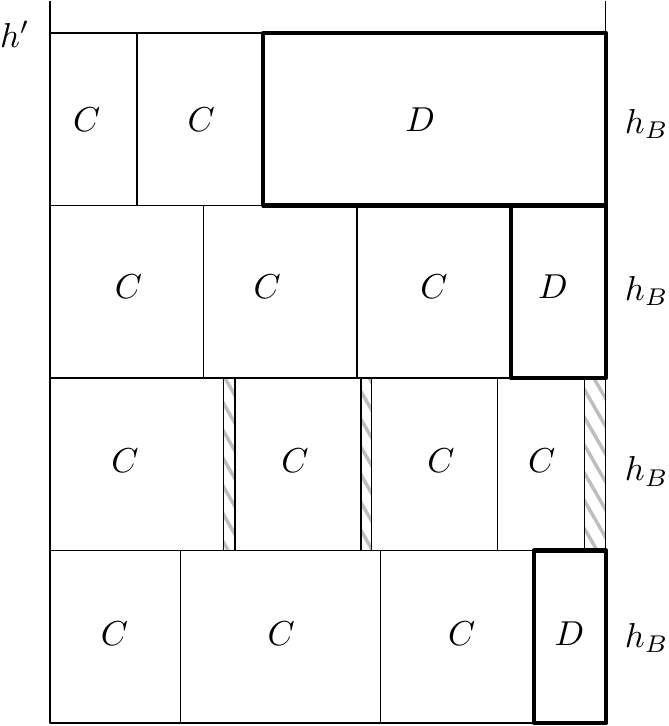}
		\caption{Well-filled container packing}
		\label{fig:fillingGapsTopLevel}
	\end{subfigure}
	~
	\begin{subfigure}[t]{0.35\textwidth}
		\centering
		\includegraphics[height=5.8cm]{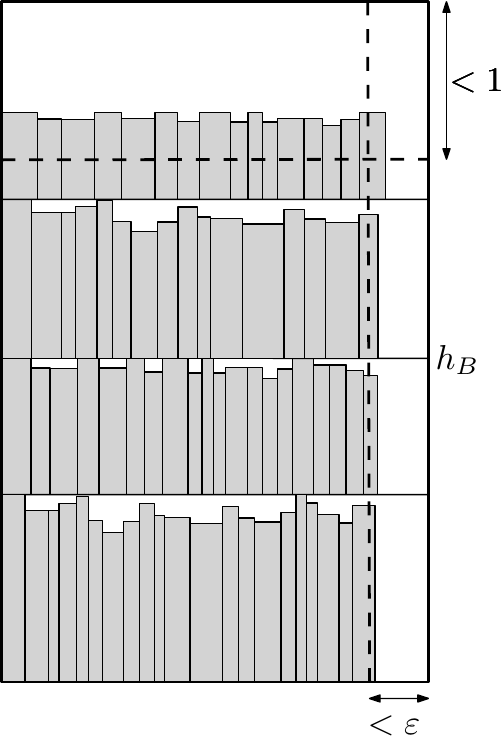}
		\caption{Full D-container}
		\label{fig:internalDContainer}
	\end{subfigure}	
	\caption{D-containers are introduced to fill gaps in the container
		packing with shelfs.}
	\label{fig:dContainerPacking}
\end{figure}
Further, we introduce the term \textit{load}:
For a container-rectangle $r$, let $\LOAD(r)$ denote the total size
of items that are packed inside the rectangle $r$.
Clearly, $\LOAD(r) \leq \SIZE(r)$.
For example, see the dense shelf $r$ for group $s$
in Figure\nobreakspace \ref {fig:narrowItems}: Here,
$\LOAD(r) \geq (1-\alpha)^s (w-\epsilon)$.
If $R$ is a set of rectangles, we define $\LOAD(R) = \sum_{r \in R} \LOAD(r)$.

In the following, we use $\alpha = \epsilon$ as shelf parameter and set 
$\lambda = 1 / (\epsilon-\epsilon^2)$.
Remember that the total height of sparse shelfs is at most
$\sum_{r=0}^\infty (1-\epsilon)^{r-1} = 1 / (\epsilon-\epsilon^2) = \lambda$
(see Section~\ref {sec:onlyNarrowItems}).

\subsubsection{Insertion Algorithm}
\label{sec:insNarrowItem}

Similar to flat items, we need a buffer for narrow items.
The \textit{N-Buffer} is a rectangular segment of height $h_B$ and width 1
placed somewhere in the strip. Items inside the N-buffer are organized in shelfs.

First we define an auxiliary algorithm called \textsc{Shelf-First-Fit}.
This algorithm tries to find a position for a narrow item without increasing the packing height.
It may not use the N-buffer. If no such position exists, the algorithm returns $\mathit{false}$.

\begin{algorithm}
	\caption{\textsc{Shelf-First-Fit}}
	\label{alg:ShelfFirstFit}
	
	\SetAlgoLined
	\DontPrintSemicolon
	\SetKwInOut{Input}{Input}
	\SetKwFunction{Shift}{Shift}
	\SetKwFunction{WidestItems}{WidestItems}	
	\SetKwFunction{Align}{Align}		
	\SetKw{If}{if}
	\SetKw{Else}{else}
	\SetKw{exit}{exit}
	
	\Input{Narrow item $i$ of group $r$}
	\BlankLine
	
	\If there is a shelf of group $r$ that can take $i$, place $i$ there \;
	\Else \If a new shelf of group $r$ can be
	opened without increasing the packing height, open it 
	and place $i$ there \;
	\Else return false \;
\end{algorithm}

\begin{figure}
	\centering
	\includegraphics[width=0.7\textwidth]{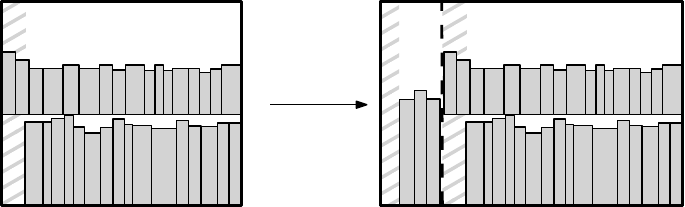}
	\caption{Placing a new shelf into the empty area of an enlarged 
		D-container does not maintain the shelf structure in general.}
	\label{fig:DcontainerExtended}
\end{figure}

With Algorithm\nobreakspace \ref {alg:ShelfFirstFit} as a subroutine, we can give the entire
algorithm for the insertion of narrow items in Algorithm\nobreakspace \ref {alg:insertNarrowItem}.
When an item can not be placed without increasing the packing height, 
the (full) N-buffer gets flushed:
For this purpose, the algorithm tries to find $q$ badly-filled levels,
where $q$ will be specified later. Aligning levels increases the area
for shelfs, but the following observation is crucial:

Assume that a D-container is filled with shelfs and its width gets increased
due to an \textsc{Align} operation. Placing a new shelf inside
the new area does not maintain the shelf structure and thus
not the lower bound on its load (see Figure\nobreakspace \ref {fig:DcontainerExtended} for an example).
Therefore, when aligning a level, items in the respective D-container get
temporarily removed. After aligning $q$ levels,
removed items from former D-containers, the buffer items,
and the current item get inserted via \textsc{Shelf-First-Fit}. 
This way, the shelf structure in the enlarged D-containers is maintained.

Finally, if there still remain items (for example, because there were
too few badly-filled levels), new shelfs are opened on top of the packing.
The next lemma states an important invariant property of Algorithm~\ref{alg:insertNarrowItem}.
\begin{algorithm}
	\caption{Insertion of a narrow item}
	\label{alg:insertNarrowItem}
	
	\SetAlgoLined
	\DontPrintSemicolon
	\SetKwInOut{Input}{Input}
	\SetKwFunction{Shift}{Shift}
	\SetKwFunction{WidestItems}{WidestItems}	
	\SetKwFunction{Align}{Align}		
	\SetKw{If}{if}
	
	\Input{Narrow item $i$}
	\BlankLine

	\If $i$ can be placed via \textsc{Shelf-First-Fit} or in the N-buffer,
	place $i$ \label{line:narrowItemDirectPlace}\;
	\uElse{
		\tcp*[l]{N-Buffer is full. Align badly-filled levels}
		Let $L$ be the set of badly-filled levels \;
		\For{$u \in L$, but at most $q$ times} {		
			\textsc{Align}$(u)$\;
			Remove items in D-container, if existing \label{line:removeNarrowItems}\;
			Define empty D-container filling the remaining width\;
		}
		Empty N-buffer \label{line:emptyNbuffer} \;
		Let $S$ be the set of all removed items from Lines~\ref{line:removeNarrowItems}
		and \ref{line:emptyNbuffer}, and item $i$ \;
		Try to reinsert each $j \in S$ via \textsc{Shelf-First-Fit} 
		\label{line:reinsertNarrowItems} \;
		For remaining items:
		Open new shelfs of width 1 on top of the packing
		\label{line:packNarrowItemsOnTop}\;
	}
\end{algorithm}
\begin{lemma}
	\label{lemma:characteristicsOfInsNarrowItem}
	 Algorithm~\ref{alg:insertNarrowItem} aligns at most $q = \bigO{1 / \epsilon}$ levels and
	 after repacking, the N-buffer is empty.
	 Further, if it increases the packing height, all levels are well-filled and D-containers are full.
\end{lemma}
\begin{proof}
	For the first claim, note that the total size of $i$ and items from the N-buffer is at most $h_B + \epsilon$.
	We prove afterwards in Lemma~\ref{lemma:narrowItemsBoundRepacking} that $q = \bigO{1/\epsilon}$
	is an appropriate choice. Afterwards in Line~\ref{line:emptyNbuffer},
	the N-buffer is emptied.
	The last claim follows immediately from the description of the algorithm.
\end{proof}
\begin{lemma}
	\label{lemma:narrowItemsBoundRepacking}
	In Algorithm\nobreakspace \ref {alg:insertNarrowItem}, aligning $q= \bigO{1 / \epsilon^2}$ levels
	is enough to (re-)insert items of size at most $h_B + \epsilon$.
\end{lemma}
\begin{proof}
	The crucial point is that narrow items of all aligned levels may also have to be reinserted.
	Let $W(D)$ be the total width of D-containers in $q$ aligned levels before aligning.
	The load of this D-containers is thus at most $W(D) h_B$. 
	Therefore, narrow items of total size at most
	$ h_B + \epsilon + W(D) h_B $
	have to be (re-)inserted.
	
	Next, we analyze the minimum load of a shelf packing in the enlarged
	D-containers.
	Let $D'$ be the set of D-containers after aligning.
	We have $\card{D'}=q$.
	Let $d' \in D'$ be the enlarged container to $d \in D$.
	Each D-container gets enlarged in width by at least $2\epsilon$, 
	\ie $w(d') \geq w(d) + 2 \epsilon$ for all $d' \in D$.
	For the load of a single D-container, we subtract $\epsilon$ from the width
	and 1 from the height, thus 
	$\LOAD(d') 
	\geq (w(d')-\epsilon) (1-\epsilon) (h_B - 1)
	\geq (w(d)+\epsilon) (1-\epsilon) (h_B - 1)$.
	The factor $1-\epsilon$ is due to height differences inside the groups.
	Thus, the total load of containers in $D'$ is
	$$
	\LOAD(D') \geq \left(\sum_{d' \in D'} \LOAD(d') \right) - \lambda
	\geq (1-\epsilon) (h_B - 1) (W(D)+ q \epsilon) - \lambda \,,
	$$
	where the term $\lambda$ is due to sparse shelfs.
	In the remainder of the proof, we show
	\begin{equation}
	\label{eq:enlargedDContainers1}
	(1-\epsilon) (h_B - 1) (W(D)+ q \epsilon) - \lambda
	\geq h_B + \epsilon + W(D) h_B \,,
	\end{equation}
	\ie all items that have to be reinserted fit into the enlarged
	D-containers of the $q$ aligned levels.
	Note that no D-container in an aligned level has a width greater than 
	$1-3\epsilon$, thus $W(D) \leq (1-3\epsilon)q$.
	Let $z = (1-\epsilon)(h_B - 1)$ and $y = (1-3 \epsilon)(z-h_B)$.
	We show 
	\begin{equation}
	\label{eq:enlargedDContainers2}
	(y + z \epsilon) q \geq h_B + \epsilon + \lambda \,,
	\end{equation}
	which implies Equation\nobreakspace \textup {(\ref {eq:enlargedDContainers1})}, since
	Equation\nobreakspace \textup {(\ref {eq:enlargedDContainers1})} is equivalent to
	$W(D) (z- h_B) + z q \epsilon \geq h_B + \epsilon + \lambda$  
	and we have
	\begin{align*}
	W(D) (z- h_B) + z q \epsilon
	&= - W(D) (h_B - z) + z q \epsilon \\
	&\geq - (1-3 \epsilon) q (h_B - z) + z q \epsilon & W(D) \leq (1-3\epsilon) q\\
	&=       (1-3 \epsilon) q (z - h_B) + z q \epsilon \\
	&=       (y + z \epsilon) q \\
	&\geq h_B + \epsilon + \lambda & eq.\nobreakspace \textup {(\ref {eq:enlargedDContainers2})} \,.
	\end{align*}
	It holds that $y + z \epsilon \geq 12$
	and with $q=2 / \epsilon^2$ it follows that
	$
	(y+z \epsilon) q \geq 12 q = 24 / \epsilon^2 \geq 13 / \epsilon^2
	+ \epsilon + 1 / \epsilon^2 \geq h_B + \epsilon + \lambda$,
	thus Equation\nobreakspace \textup {(\ref {eq:enlargedDContainers2})} is shown.
\end{proof}

\subsubsection{Analysis}
The goal of this subsection is to show that if Algorithm\nobreakspace \ref {alg:insertNarrowItem}
increases the packing height, the packing up to the previous height is well-filled.
As the first step, we analyze the load of C-containers in the following lemma.
The ideas in the proof are similar to those used in Section\nobreakspace \ref {sec:approxGuarantee}.

\begin{figure}
	\centering
	\includegraphics[width=0.7\textwidth]{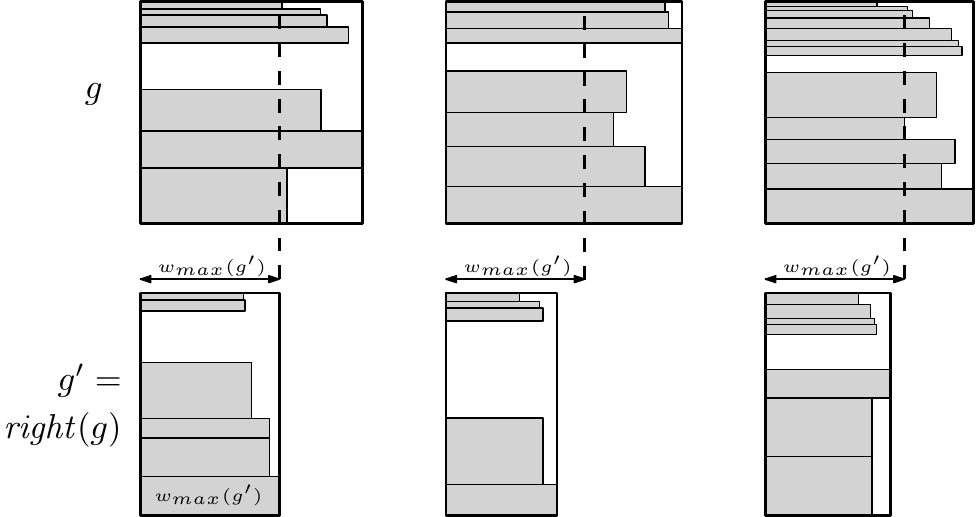}
	\caption{Containers of group $g$ have total load at least 	$h(g) \wmax(g')$.}
	\label{fig:fillingGaps}
	
\end{figure}

\begin{lemma}
	\label{lemma:loadCcontainers}
	Let $\C{\con}$ be a container instance. Then,
	$\LOAD(\C{\con}) \geq (1- 2 \epsilon) \SIZE(\C{\con}) - 
	\bigO{\omega / \epsilon^3}$.
\end{lemma}

\begin{proof}
	We write $C$ instead of $\C{\con}$ for short.
	For a category $l \in W$, the group $g_0(l)$ is defined as $(l,A,0)$ if
	block $A$ is non-empty, and otherwise $(l,B,0)$.
	Let $G_1 = G \setminus \bigcup_{l \in W} g_0(l)$.
	Further, let $g \in G_1$ and $g'=\mathit{right}(g)$. 
	With \invC - \invD we have $K_g \geq K_{g'}$
	(set $K_{g'}=0$, if $g'$ is not defined).
	
	First, we consider the load of containers of the group $g$.
	Let $C(g) = \{c \in C \mid R(c)=g \}$ be the set of such containers.
	By \invB,  each container of group $g$ contains items of 
	width at least $\wmax(g')$	(see Figure\nobreakspace \ref {fig:fillingGaps}) and height $h(g)$.
	We have:
	
	\begin{equation}
	\label{eq:loadCg}
	\begin{aligned}[b]
	\LOAD(C(g)) &\geq h(g) \wmax(g') \\
	&\geq (h_B - 1) (K_g - 1)  \wmax(g') 						& \text{\invE} \\
	&=    (h_B K_g - h_B - K_g + 1) \wmax(g') \\
	&=     h_B K_g \wmax(g') - (h_B + K_g - 1) \wmax(g') \\
	&\geq     h_B K_{g'} \wmax(g') - (h_B + K_g - 1) \wmax(g') 	& K_g \geq K_{g'} \\					
	&\geq  \SIZE(C(g')) - (h_B + K_g - 1)  				& \wmax(g') \leq 1 \\	
	\end{aligned}
	\end{equation}
	Now consider the set of groups
	$G_2 = G_1 \setminus \bigcup_{l \in W} \mathit{right}(g_0(l))$
	(\ie we drop also the $(l,X,1)$-groups from $G$).
	Let $g' = \mathit{right}(g)$ for all $g \in G_1$. Then,
	$\{ g' \mid g \in G_1 \} = \{ \mathit{right}(g) \mid g \in G_1 \}
	= G_2$.
	Furthermore, let $C_1$, $C_2$ be the sets of container rectangles 
	of groups $G_1$, $G_2$.
	Summing over each group in $G_1$ gives the total load of containers in $C_1$:

	\begin{align*}
	\LOAD(C_1)  &=    \sum_{g \in G_1} \LOAD(C(g))  \\
	&\geq \sum_{g \in G_1} \SIZE(C(g')) - (h_B + K_g - 1) 
	& \text{eq.\nobreakspace \textup {(\ref {eq:loadCg})}} \\
	&=    \sum_{g \in G_2} \SIZE(C(g)) - \sum_{g \in G_1} h_B + K_g - 1 \\
	&\geq    \SIZE(C_2) - \card{G_1} (h_B + K_g^* - 1) \,,
	\end{align*}
	where $K_g^* = \max_{g \in G_1} K_g$.
	As the next step, we show 
	$
	\card{G_1} (h_B + K_g^* - 1)
	\leq \epsilon \SIZE(I_L) + 	\bigO{\omega / \epsilon^3} $, implying
	\begin{equation}
	\label{eq:LoadC1}
	\LOAD(C_1) \geq \SIZE(C_2) - \epsilon \SIZE(I_L) - 
	\bigO{\omega / \epsilon^3} \,.
	\end{equation}
	With \invC-\invD we get
	$K_g^* \leq 2^l k \leq \frac{1}{\epsilon} k 
	= \frac{1}{\epsilon} \floor{ \frac{\epsilon}{4 \omega h_B} \SIZE(I_L) }
	\leq \frac{1}{4 \omega h_B} \SIZE(I_L) \,.$
	Since $G_1$ contains less than 
	$\card{G} \leq \omega \left(2+ \frac{16}{\epsilon}\right)$ groups 
	(by Lemma\nobreakspace \ref {lemma:NumberOfGroups}), it follows that
	$$
	\card{G_1} K_g^* \
	\leq \omega \left(2+ \frac{16}{\epsilon}\right) \frac{1}{4 \omega h_B}  \SIZE(I_L)
	= \frac{2 \epsilon + 16}{4 \epsilon h_B} \SIZE(I_L) 
	\leq \epsilon \SIZE(I_L) \,.
	$$
	Further, the term 
	$\card{G_1} (h_B - 1)$ is bounded by $	\bigO{\omega / \epsilon^3}$, as
	$h_B \in \bigO{1 / \epsilon^2}$.
	Thus, Equation\nobreakspace \textup {(\ref {eq:LoadC1})} is shown.
	Like in the proof of Theorem\nobreakspace \ref {theo:approxGuarantee}
	(Equation\nobreakspace \textup {(\ref {eq:LossOfG+})}),	
	all containers of $(l,X,0)$-groups can be placed in $k$
	levels of height $h_B$. Clearly, the same holds for $(l,X,1)$-groups and thus
	$\SIZE(C_2)  \geq \SIZE(C) - 2 \omega h_B k \geq \SIZE(C) - \epsilon \SIZE(I_L)$. 
	Using $\LOAD(C) \geq \LOAD(C_1)$, 
	Equation\nobreakspace \textup {(\ref {eq:LoadC1})},
	and the last inequality with $\SIZE(I_L) \leq \SIZE(C)$
	completes the proof.
\end{proof}
The next lemma shows a similar result for D-containers.

\begin{lemma}
	\label{lemma:loadDcontainers}
	Let $D$ be the set of all full D-containers in a packing of height $h'$
	(see Figure\nobreakspace \ref {fig:fillingGapsTopLevel}).
	Assuming that in each D-container the loss of width is at most 
	$2 \epsilon$, it holds that
	$\LOAD(D) \geq (1-2 \epsilon) \SIZE(D) - 2 \epsilon h' - \lambda$.
\end{lemma}

\begin{proof}
	We first consider a single full D-container $d \in D$
	like shown in Figure\nobreakspace \ref {fig:internalDContainer}.
	The loss in height is at most 1 and the loss of width
	by assumption at most $2 \epsilon$.
	Further, multiplying the height with factor $1-\epsilon$ regards 
	height differences in the groups, thus:
	\begin{align*}
	\LOAD(d) 
	&\geq (1-\epsilon) (w(d)- 2 \epsilon) (h_B - 1) \\
	&=    (1-\epsilon) \left( w(d) h_B - w(d) - 2 \epsilon (h_B - 1) \right) \\
	&\geq (1-\epsilon) w(d) h_B - w(d) - 2 \epsilon(h_B - 1) \,.
	\end{align*}
	Sparse shelfs can waste a total area of at most $\lambda$, therefore
	\begin{align*}
	\LOAD(D)
	&=    \left( \sum_{d \in D} \LOAD(d) \right) - \lambda \\
	&\geq \left( \sum_{d \in D} (1-\epsilon) w(d) h_B - w(d) 
	- 2 \epsilon(h_B - 1) \right) - \lambda \\
	&=    (1-\epsilon) \SIZE(D) - \left(\sum_{d \in D} w(d) \right)
	- \card{D} 2 \epsilon  (h_B - 1) - \lambda \\
	\intertext{and with $\sum_{d \in D} w(d) = \SIZE(D) \frac{1}{h_B}$ 
		it holds further }
	&=    \left(1-\epsilon - \frac{1}{h_B} \right) \SIZE(D) 
	- 2 \card{D} \epsilon  (h_B - 1) - \lambda \\
	&\geq (1-2 \epsilon) \SIZE(D) - 2 \epsilon h' - \lambda \,,
	\end{align*}
	where the last inequality follows from $1 / {h_B} \leq \epsilon$
	and $\card{D} \leq h' / {h_B}$.
\end{proof}

The next corollary uses Lemmas\nobreakspace \ref {lemma:loadCcontainers} and\nobreakspace  \ref {lemma:loadDcontainers}
to bound the total load of a container packing filled with D-containers.

\begin{corollary}
	\label{lemma:heightIaux}	
	Let $C$ be a container instance and let a packing of $C$, filled
	with D-containers, of height $h'$ be given.
	Let $I_{aux}$ be the set of all big, flat, and narrow items
	in the packing.
	If all levels are well-filled and all D-containers are full, then
	$\SIZE(I_{aux}) \geq (1- 4 \epsilon)h' -\bigO{\omega / \epsilon^3}$.
\end{corollary}

\begin{proof}
	According to Lemma\nobreakspace \ref {lemma:loadCcontainers}, C-containers are filled with
	$\LOAD(C) \geq (1 - 2 \epsilon) \SIZE(C) - z$
	for $z = \bigO{\omega / \epsilon^3}$.
	For the sake of simplicity, we can assume that each level is aligned
	and that the remainder is filled by a D-container:
	The load is the same as if the level would be aligned and filled by an 
	empty D-container.
	An unaligned but well-filled level has gaps of total width
	at most $2 \epsilon$. Hence, the D-container would have a loss of
	width of at most $2 \epsilon$ and we can apply	
	Lemma\nobreakspace \ref {lemma:loadDcontainers}.
	
	As $h'$ is the packing height and also the total area of the packing
	(in a strip of width 1), $\SIZE(D) = h' - \SIZE(C)$ and it holds further
	\begin{align*}
	\SIZE(I_{aux}) 
	&=     \LOAD(C) + \LOAD(D) \\
	&\geq (1-2\epsilon) \SIZE(C) - z + (1-2\epsilon) \SIZE(D) 
	- 2\epsilon h' - \lambda \\
	&=    (1-2\epsilon) \SIZE(C) + (1-2\epsilon) (h' - \SIZE(C))
	- 2\epsilon h' - z - \lambda \\	
	&=    (1-4\epsilon) h'
	- z - \lambda \,.
	\end{align*}
	Since both terms $z$ and $\lambda$ are dominated by 
	$\bigO{\omega / \epsilon^3}$, the claim is proven.
\end{proof}

Recall that Algorithm\nobreakspace \ref {alg:insertNarrowItem} only increases the packing height if all levels
are well-filled and all D-containers are full, thus we obtain the following result. 
The main argument and the notation of Lemma\nobreakspace \ref {lemma:finalPackingHeight} 
is similar to \cite{kenyon2000near}.

\begin{lemma}
	\label{lemma:finalPackingHeight}
	Let $h'$ be the height of the container packing. Algorithm\nobreakspace \ref {alg:insertNarrowItem}
	returns a packing of height $h_\mathit{final}$, such that
	$ h_\mathit{final} \leq \max \left\{ h', 
	(1+\epsilon') \SIZE(I)+ \bigO{\omega / \epsilon^3} \right\} $,
	where $\epsilon' \in \bigO{\epsilon}$.
\end{lemma}

\begin{proof}
	If Algorithm\nobreakspace \ref {alg:insertNarrowItem} does not reach Line\nobreakspace \ref {line:packNarrowItemsOnTop},
	$h_{final} = h'$ and thus the claim holds trivially.
	Now suppose that Algorithm\nobreakspace \ref {alg:insertNarrowItem} opens at least one new shelf in 
	Line\nobreakspace \ref {line:packNarrowItemsOnTop} placed on top of the packing.
	By Lemma~\ref{lemma:characteristicsOfInsNarrowItem}, all levels are well-filled and all D-containers are full,
	and thus we can apply \MakeUppercase Corollary\nobreakspace \ref {lemma:heightIaux}.
	Each horizontal segment in the packing is of one of two types:
	Either it is a container-level of height $h_B$ 
	(containing C- and D-containers) 
	or a shelf of width 1 (containing narrow items).
	Let $h'$ be the total height of container-levels and $h''$ the total
	height of shelfs of width 1. Hence, $h_{final} = h' + h''$.
	We partition the set of items $I$ into $I_{aux}$ like in  
	\MakeUppercase Corollary\nobreakspace \ref {lemma:heightIaux}	and ${I_N}'$, the set of narrow items placed in shelfs 
	of width 1.
	Let $\epsilon' = \frac{1}{1 - 4 \epsilon} - 1 \in \bigO{\epsilon}$
	and $z = (1+\epsilon') \bigO{\omega / \epsilon^3}$.
	Then, the conclusion of \MakeUppercase Corollary\nobreakspace \ref {lemma:heightIaux} is equivalent to:
	\begin{equation}
	\label{eq:heightContainerPacking}
	h' \leq  (1+\epsilon') \SIZE(I_{aux}) + z
	\end{equation}
	Further, with Lemma\nobreakspace \ref {lemma:onlyNarrowItems} we can bound the packing height
	for the shelfs of width 1.
	Remember that the N-buffer causes additional height of $h_B$. 
	If we see it as part of the shelf packing, we get from the proof of
	Lemma\nobreakspace \ref {lemma:onlyNarrowItems} with $\alpha=\epsilon$:
	\begin{equation}
	\label{eq:heightShelfPacking}
	h'' \leq \frac{1}{(1-\epsilon)^2} \SIZE({I_N}') + \lambda + h_B
	\end{equation}
	The final packing height $h_{final}$ can be bounded as follows:
	\begin{align*}
	h_{\mathit{final}} &= h' + h'' \\
	&\leq (1+\epsilon') \SIZE(I_{aux}) + z 
	+ \frac{1}{(1-\epsilon)^2} \SIZE({I_N}') + \lambda + h_B 
	& \text{eq.\nobreakspace \textup {(\ref {eq:heightContainerPacking})},\textup {(\ref {eq:heightShelfPacking})}} \\
	&\leq    (1+\epsilon') \SIZE(I_{aux}) + (1+\epsilon') \SIZE({I_N}') 
	+ \lambda + h_B + z \\
	&= (1+\epsilon') \SIZE(I) + \lambda + h_B + z
	& I = I_{aux} \mathbin{\dot{\cup}} {I_N}'
	\end{align*}
	Again, the additive terms $\lambda$, $h_B$, and $z$
	are bounded by $\bigO{\omega / \epsilon^3}$.
\end{proof}

\subsection{Changes in the Container Packing}
So far we assumed a static layout of D-containers. But as the width of C-containers may change
due to inserted or removed items, D-containers also change in width.
Moreover, when the pattern (see LP \ref {lp:binpacking}) of a level changes,
D-containers have to be redefined. In both cases, narrow items may have to be 
reinserted.
This two scenarios are considered in the following.

\subsubsection{Stretch Operation and Narrow Items}
\label{sec:stretchNarrowItems}

The \textsc{Stretch} operation (introduced in Section\nobreakspace \ref {sec:auxFunctions}) 
changes the position of containers in one level such that no item overlaps with 
a C-container. When the stretched level includes a D-container, this may 
overfill a level  (see the example in Figure\nobreakspace \ref {fig:stretchNarrowItems}). 
In the worst case, all items from a maximal filled D-container have to repacked,
\ie items of total size $(1-\epsilon) h_B$. 
Lemma~\ref{lemma:narrowItemsBoundRepacking} bounds the repacking size with $q h_B$ in this case.
\begin{figure}
	\centering
	\begin{subfigure}[t]{0.45\textwidth}
		\centering
		\includegraphics[width=0.7\textwidth, page=1]{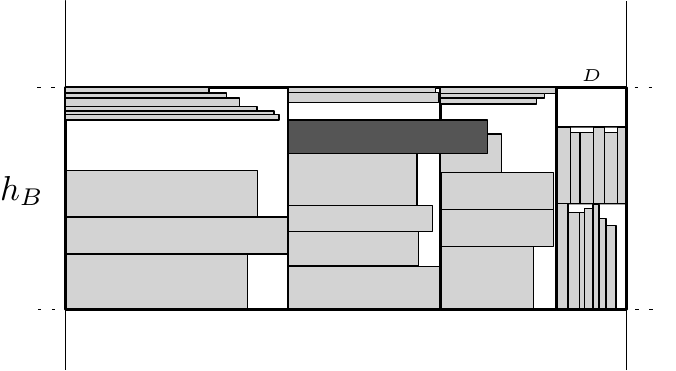}
		\caption{Before \textsc{Stretch}}
	\end{subfigure}
	~
	\begin{subfigure}[t]{0.45\textwidth}
		\centering
		\includegraphics[width=0.7\textwidth, page=2]{stretch_narrow_items.pdf}
		\caption{After \textsc{Stretch}}
	\end{subfigure}	
	\caption{Overlapping items (dark colored) after \textsc{Stretch} operation}
	\label{fig:stretchNarrowItems}
\end{figure}

\subsubsection{Repacking of Levels}

Recall that Algorithm\nobreakspace \ref {alg:insertContainer} makes use of the operation
\textsc{Improve}, which changes some components in the
LP/ILP-solutions that represent the level configuration.
In the packing several levels are repacked, which requires
that the D-containers have to be rebuilt.
The new set of D-containers may have a lower load since the
widths change and thus sparse shelfs can occur.
The next lemma bounds the total area of narrow items that may have to be
packed outside the new D-containers.

\begin{lemma}
	\label{lemma:repackingNarrowItemsImprove}
	Assume that the packing of C-containers changes in $p$ levels. 
	Let $D$ be the set of D-containers in this levels before repacking
	and $D'$ the set of new defined D-containers in the repacked and aligned
	levels. Then,
	$ \LOAD(D) - \LOAD(D') \leq \bigO{p / \epsilon} +~\lambda$.
\end{lemma}
\begin{proof}
	The proof is similar to the one of Lemma\nobreakspace \ref {lemma:narrowItemsBoundRepacking}.
	Let $W(D)$ be the total width of the containers in $D$.
	Clearly, $\LOAD(D) \leq h_B W(D)$.
	The total area of the new containers in $D'$ does not decrease in comparison
	to $D$, but the new defined shelfs can have a smaller load.
	An analog argument like in the proof of 
	Lemma\nobreakspace \ref {lemma:narrowItemsBoundRepacking} gives
	$
	\LOAD(D') \geq (1-\epsilon) (h_B - 1) (W(D) - p \epsilon) - \lambda
	$.
	Since each D-container has width at most $1-\epsilon < 1$, it holds that
	$W(D) < p$. Thus:
	\begin{align*}
	&\LOAD(D) - \LOAD(D') \\
	&\leq h_B W(D) - (1-\epsilon) (h_B - 1) (W(D)- p \epsilon) + \lambda \\
	&=    h_B W(D) - (1-\epsilon) (h_B - 1) W(D) 
	+ (1-\epsilon) (h_B - 1) p \epsilon + \lambda \\
	&=    W(D) (h_B - (h_B - 1 - \epsilon h_B + \epsilon))
	+ (1-\epsilon) (h_B - 1) p \epsilon + \lambda \\
	&<    W(D) (1 + \epsilon h_B) + p \epsilon + \lambda \\
	&=    p (1 + \epsilon h_B + \epsilon) + \lambda 
	\end{align*}
	The claim follows by $1 + \epsilon h_B + \epsilon = \bigO{\epsilon h_B} =
	\bigO{1 / \epsilon}$.
\end{proof}

\subsection{Overall Algorithm}
So far we presented online algorithms
for the case where only narrow items are present (Section\nobreakspace \ref {sec:onlyNarrowItems})
and an algorithm that fills gaps in the container packing with narrow items
(Algorithm\nobreakspace \ref {alg:insertNarrowItem}). Recall that the container packing approach
requires a minimum size of~$I_L$.

Therefore, it remains to show how to handle the general setting, where
narrow items arrive, while $\SIZE(I_L(t))$ is possibly too small for the
container packing. For this purpose, the overall Algorithm\nobreakspace \ref {alg:insertGeneralSetting} 
works in one of two modes:
In the \textit{semi-online mode}, small instances of $I_L$ are packed offline 
and arbitrary large instances of $I_N$ online.
In contrast, in the \textit{online mode}, 
the techniques of Sections\nobreakspace  \ref {sec:containerPacking} to\nobreakspace  \ref {sec:narrowItems} 
are applied.

\begin{figure}
	\centering
	\includegraphics[scale=0.8]{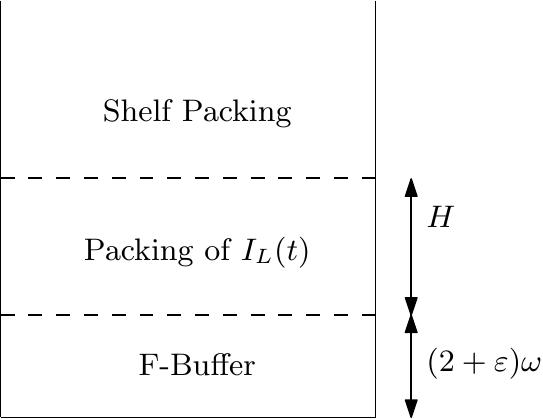}
	\caption{Top-level sketch of the packing in the semi-online mode}
	\label{fig:semionlinePacking}	
\end{figure}

\paragraph*{Semi-online mode}
The top-level structure of the packing in the semi-online mode is
shown in Figure\nobreakspace \ref {fig:semionlinePacking}. It is divided into three parts:

\begin{itemize}
	\item 
	To avoid that flat items cause a huge amount of repacking, 
	we use an F-buffer for flat items like introduced in Section\nobreakspace \ref {sec:insertFlat},
	but with slot height $y = 2 + \epsilon$.
	
	\item
	When the F-buffer overflows or a new big item arrives, 
	a new packing of
	$I_L(t)$, including all items from the F-buffer, is obtained by
	the algorithm of Kenyon and R{\'e}mila (K\&R) \cite{kenyon2000near}.
	In the strip, we reserve a segment of fixed height $H$ for the changing packings of $I_L(t)$,
	where $H$ is the packing height of a maximum instance of $I_L(t)$.
	
	\item 
	Narrow items are packed separately on top of the packing via the shelf algorithm 
	described in Section\nobreakspace \ref {sec:onlyNarrowItems}.
\end{itemize}

\paragraph*{Online mode}
As soon as $\SIZE(I_L(t)) \geq \frac{4 \omega h_B}{\epsilon} (h_B + 1)$,
we go over to use Algorithms\nobreakspace \ref {alg:insertBigItem},  \ref {alg:insertFlat}, and\nobreakspace  \ref {alg:insertNarrowItem}, 
depending on the item type.

\begin{algorithm}
	\caption{Insertion of an item in the general setting}
	\label{alg:insertGeneralSetting}
	
	\SetAlgoLined
	\DontPrintSemicolon
	\SetKwInOut{Input}{Input}
	\SetKw{If}{if}
	\SetKw{IF}{If}	
	\SetKw{Or}{or}
	\SetKwFunction{Shift}{Shift}
	\SetKwFunction{Stretch}{Stretch}	
	\SetAlgoNoEnd
	
	\Input{Item $i_t$}
	\BlankLine
	
	\uIf(\tcp*[f]{Semi-online mode})
	{$\SIZE(I_L(t)) < \frac{4 \omega h_B}{\epsilon} (h_B + 1)$} {
		\uIf{$i_t$ is flat}{
			Place $i_t$ in the F-buffer. \;
		}
		\uIf{$i_t$ is big \Or F-buffer overflows} { 
			Pack $I_L(t)$ (incl. items in the F-buffer and $i_t$) 
			with K\&R \; 
			Replace packing of $I_L(t-1)$ with the new packing
		}
		
		\uIf{$i_t$ is narrow}{
			Place $i_t$ via shelf packing
		}
	}
	\uElse(\tcp*[f]{Online mode}){
		Use Algorithms \ref {alg:insertBigItem}, \ref {alg:insertFlat} 
		or \ref {alg:insertNarrowItem} 
		depending on the type of $i_t$.
	}
\end{algorithm}

The following theorem states that the overall 
Algorithm\nobreakspace \ref {alg:insertGeneralSetting} returns a packing of
the desired packing height.

\begin{theorem}
	\label{lemma:packingHeightOverallAlgo}
	Let $h$ be the height of the packing returned by 
	Algorithm\nobreakspace \ref {alg:insertGeneralSetting}.
	It holds that
	$h \leq (1+\bigO{\epsilon}) \OPT(I) \allowbreak + \mathrm{poly}(1 / \epsilon) $.
\end{theorem}

\begin{proof}
	In the semi-online mode, the packing consists of two segments
	of constant height 
	$\omega y = \bigO{\log \frac{1}{\epsilon}}$ and
	$H = \bigO{\frac{\omega {h_B}^2}{\epsilon}}$
	(see Figure\nobreakspace \ref {fig:semionlinePacking})
	and the shelf packing of height $h'$.
	By Lemma\nobreakspace \ref {lemma:onlyNarrowItems},
	$h' \leq (1+\epsilon) \OPT(I_N) + \bigO{1 / \epsilon^4}$.
	Therefore, the total packing height is at most
	$(1+\epsilon) \OPT(I_N) + \bigO{1 / \epsilon^4} + \bigO{\log \frac{1}{\epsilon}}
	+ \bigO{\frac{\omega {h_B}^2}{\epsilon}}
	\leq (1+\epsilon) \OPT(I) + \mathrm{poly}(1/\epsilon)$.
	
	In the online mode, Lemma\nobreakspace \ref {lemma:finalPackingHeight} bounds the final
	packing height after insertion of narrow items with
	$
	h_\mathit{final} \leq \max \left\{ h', 
	(1+\epsilon') \SIZE(I)+ \bigO{\omega / \epsilon^3} \right\} $,
	where $h'$ is the height of the container packing
	and $\epsilon' \in \bigO{\epsilon}$.
	In the case $h_{final} = h'$, we need the result from Theorem\nobreakspace \ref {theo:heightContainerPackingAfterImprove}.
	Equation\nobreakspace \textup {(\ref {eq:heightContainerPackingAfterImprove})} states
	$
	h' \leq (1+2\Delta) \OPT(I_L) + 2(1+\epsilon)z + m
	$,
	where $\Delta = \bigO{\epsilon}$, $z=\bigO{1 / \epsilon^2}$, and
	$m = \bigO{\frac{1}{\epsilon} \log \frac{1}{\epsilon}}$.
	With $\OPT(I_L) \leq \OPT(I)$, the claim follows.
	In the other case $h_{final} > h'$, the claim follows directly
	by $\SIZE(I) \leq \OPT(I)$ and 
	$\omega = \bigO{\frac{1}{\epsilon} \log \frac{1}{\epsilon}}$.	
\end{proof}

\section{Migration Analysis}
\label{sec:migrationAnalysis}

Let $\Repack(t)$ denote the total size of repacked items at time 
$t \geq 1$, \ie at the arrival of item $i_t$.
Recall out notion of amortized migration: While in \cite{sanders} 
for the migration factor $\mu$ holds $\Repack(t) \leq \mu \SIZE(i_t)$,
we say that an algorithm has amortized migration factor $\mu$ if
$\sum_{j=1}^{t} \Repack(j) \leq \mu \SIZE(I(t))$,
where $I(t) = \{i_1, \ldots, i_t \}$.
A useful interpretation of amortized migration is the notion of \textit{repacking potential},
also used by Skutella and Verschae in \cite{skutellaVerschaeJournal}.
Each arriving item $i_t$ charges a global repacking budget with its repacking potential $\mu \SIZE(i_t)$.
Repacking a set of items $R$ costs $\SIZE(R)$ and is paid from the budget.

If we waant to show that $\mu$ is the amortized migration factor,
it suffices to show that the budget is non-negative at each time:

\begin{lemma}
	\label{lemma:amortizedAnalysisPotentialFunction}
	Let $\Phi(t)$ denote the repacking budget at time $t$ and assume
	$\Phi(0) = 0$.
	If $\Phi(t) \geq 0$ at each time $t \geq 1$, it holds that
	$\sum_{j=1}^{t} \Repack(j) \leq \mu \SIZE(I(t))$.
\end{lemma}
The proof of Lemma \ref{lemma:amortizedAnalysisPotentialFunction} follows by a standard argument
for potential functions, see \cite[Sec. 17.3]{cormen}.

\subsection{Repacking of Operations}
\label{sec:repackingOperations}
As the first step we consider the repacking performed by the two major operations \textsc{Shift}
and \textsc{ShiftA} used for the insertion of a big or flat item.
In the following analysis we use the values $\omega, d, q, h_B, p$ as defined
in previous chapters. See Table\nobreakspace \ref {tab:boundsForValues} for a summary.

\subsubsection{Shift}
In general, shifting an item $i_t$ into a group is followed by a sequence of 
shift operations and thus the repacking performed
in each shift operation sums up to the total repacking.
The maximum amount of repacking occurs in a sequence beginning
in group $(l,B,q(l,B))$ and ending in group $(l,A,-1)$.
Therefore, we consider the sequence of shift operations
$$
\textsc{Shift}(g^0,S^0), \textsc{Shift}(g^1,S^1), \ldots, 
\textsc{Shift}(g^{d-1},S^{d-1}), \textsc{Shift}(g^d,S^d)
$$
with $S^0=\{i_t\}$, $g^0=(l,B,q(l,B))$, $g^{d-1}=(l,A,0)$, and $g^d=(l,A,-1)$. 
See Figure\nobreakspace \ref {fig:shiftSequenceRepacking}.
Like in Section\nobreakspace \ref {sec:shiftSequence}, let $S^j$ and $S_{out}^j$ be the
sets $S$, $S_{out}$ for the call \textsc{Shift}$(g^j,S^j)$.
For the maximum amount of repacking we can assume
\begin{align}
\label{eq:estimateSoutEQ}
h(S^j) = h(S_{out}^{j-1}) \leq h(S_{out}^j) \,.
\end{align}

\begin{figure}
	\centering
	\includegraphics[width=\textwidth]{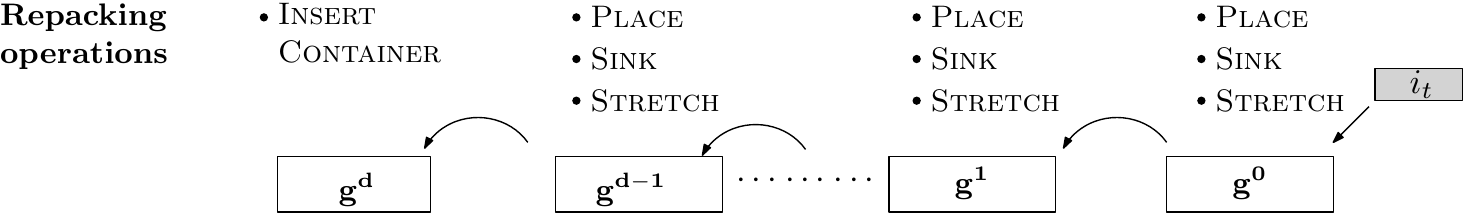}
	\caption{Repacking in a maximum shift sequence}
	\label{fig:shiftSequenceRepacking}
\end{figure}

\begin{lemma}
	\label{lemma:shiftRepacking}
	The total repacking in the above defined shift sequence is at most
	$\bigO{\frac{q h_B d^2}{\epsilon}}$.
\end{lemma}
\begin{proof}
	Repacking in the \textsc{Shift} operation (Algorithm\nobreakspace \ref {alg:shift}) is performed by the
	auxiliary operations \textsc{Sink}, \textsc{Place}, \textsc{Stretch}, and \textsc{InsertContainer}.
	Note that the repacking of the first three operations depends 
	on their position in the shift sequence.
	An operation is at \textit{position} $j$ if it is performed as part of 
	\textsc{Shift}$(g^j,S^j)$.
	First we analyze the repacking of the operations 
	\textsc{Sink}, \textsc{Place}, \textsc{Stretch} at position $j$ 
	with $0 \leq j \leq d$.
	
	\paragraph*{Sink}
	The \textsc{Sink} operation is necessary for each container
	that contains a big item from $S_{out}^j$. 
	Since each big item has minimum height $\epsilon$, 
	the number of big items in $S_{out}^j$ is at most 
	$\floor{h(S_{out}^j) / \epsilon}$.
	Hence, the same bound holds for the number of affected containers.
	Flat items in $S_{out}^j$ do not cause sinking (see Section\nobreakspace \ref {sec:insertFlat}).
	The repacking for each sinked container is at most $h_B$, thus we get
	\begin{align}
	\label{eq:repackShiftSink}
	\mathit{Repack\_Sink}(j) 
	\leq \floor{\frac{h(S_{out}^j)}{\epsilon}} h_B
	\leq h(S_{out}^j) \frac{h_B}{\epsilon} \,.
	\end{align}
	
	\paragraph*{Place (without Stretch)}
	Remember that the placing differs depending on the item type:
	Big items get placed via Algorithm\nobreakspace \ref {alg:placeBig} without repacking,
	while placing flat items via Algorithm\nobreakspace \ref {alg:placeFlat} requires resorting 
	of containers.
	For a worst-case analysis we assume $S^j$ to consist of flat items solely.
	When placing the set $S^j$ with Algorithm\nobreakspace \ref {alg:placeFlat}, the repacking size
	is at most $n h_B$, where $n$ is the size of the partition in 
	Line\nobreakspace \ref {line:partitionShiftFlat}.
	A partition into at most $\ceil{h(S^j) / (1-\epsilon)} + 1$ sets can be found 
	greedily. Using Equation\nobreakspace \textup {(\ref {eq:estimateSoutEQ})} gives
	\begin{align}
	\label{eq:repackShiftPlace}
	\mathit{Repack\_Place} (j)
	\leq \left(\ceil{\frac{h(S^j)}{1-\epsilon}} + 1\right) h_B
	\leq h(S_{out}^j) \frac{h_B}{1-\epsilon} + 2 h_B\,.
	\end{align}
	
	\paragraph*{Stretch}
	First, we bound the number of containers where items from $S^j$ could be placed.
	Let $S^j$ be partitioned into $S^j_b$ and $S^j_f$, the sets of big and flat items.
	Items from $S^j_b$ can be placed in at most $\floor{h(S^j_b) / \epsilon}$ 
	containers since each big item as minimum height $\epsilon$.
	As analyzed in the previous paragraph, flat items get placed in at most
	$\frac{h(S^j_f)}{1-\epsilon} + 2$ containers.
	With Equation\nobreakspace \textup {(\ref {eq:estimateSoutEQ})} we get
	$$
	\floor{\frac{h(S^j_b)}{\epsilon}} + \frac{h(S^j_f)}{1-\epsilon} + 2
	\leq \frac{h(S^j_b)}{\epsilon} + \frac{h(S^j_f)}{\epsilon} + 2
	=   \frac{h(S^j)}{\epsilon} + 2
	\leq \frac{h(S_{out}^j)}{\epsilon} + 2
	$$
	as the total number of containers where new items may be placed.
	Each placing in a container is followed by a \textsc{Stretch} operation, 
	which repacks containers of one level. In addition, a D-container with
	narrow items has to be repacked. Using the result from 
	Lemma~\ref{lemma:narrowItemsBoundRepacking}, the repacking for one stretched level
	is bounded by $h_B + q h_B = \bigO{q h_B}$.
	\begin{align}
	\label{eq:repackShiftStretch}
	\mathit{Repack\_Stretch}(j)
	&\leq \left( \frac{h(S_{out}^j)}{\epsilon} + 2 \right) \bigO{q h_B}
	= h(S_{out}^j) \cdot \bigO{\frac{q h_B}{\epsilon}} + \bigO{q h_B}
	\end{align}
	
	\paragraph*{InsertContainer}
	\MakeUppercase Corollary\nobreakspace \ref {lemma:1-2bins} states that the total height
	of items that have to be packed into new containers is at most
	$h(S^0) + h_B - 2 \leq h_B - 1$. 
	Therefore, exactly one container gets inserted.
	The insertion of a container via Algorithm\nobreakspace \ref {alg:insertContainer}
	makes use of \textsc{Improve}, which repacks $p$ levels.
	In addition, according to Lemma\nobreakspace \ref {lemma:repackingNarrowItemsImprove},
	narrow items of maximum size 
	$\bigO{p / \epsilon} + \lambda$ have to be reinserted.
	With Lemma~\ref {lemma:narrowItemsBoundRepacking}, at most
	$\frac{\bigO{p / \epsilon} + \lambda}{h_B + \epsilon}
	= \bigO{\frac{p}{\epsilon h_B}}$ times $q$ levels get aligned.
	Using $h_B = \bigO{q}$ gives:
	\begin{equation}
	\label{eq:repackShiftInsCon}
	\mathit{Repack\_InsCon} 
	\leq p h_B + \bigO{\frac{p}{\epsilon h_B}} q h_B
	= \bigO{p q} + \bigO{\frac{p q }{\epsilon}}
	=    \bigO{\frac{p q }{\epsilon}}
	\end{equation}

	\paragraph*{Conclusion}
	Sinking, placing, and stretching occurs in the first $d$ of the $d+1$ shift steps.
	With Equations\nobreakspace  \textup {(\ref {eq:repackShiftSink})} to\nobreakspace  \textup {(\ref {eq:repackShiftStretch})} ,
	this repacking is at most of size:
	
	\begin{equation}
	\label{eq:migrationSum}
	\begin{aligned}[b]
	&\quad \sum_{j=0}^{d-1} \mathit{Repack\_Sink}(j) 
	+ \mathit{Repack\_Place}(j)
	+ \mathit{Repack\_Stretch}(j) \\
	&\leq \sum_{j=0}^{d-1} 
	\left( h(S_{out}^j) \frac{h_B}{\epsilon} \right)
	+ \left( h(S_{out}^j) \frac{h_B}{1-\epsilon} + 2 h_B \right)
	+ \left( h(S_{out}^j) \cdot \bigO{\frac{q h_B}{\epsilon}} + \bigO{q h_B} \right) \\
	&=  \left( \frac{h_B}{\epsilon} + \frac{h_B}{1-\epsilon} 
	+ \bigO{\frac{q h_B}{\epsilon}} \right)
	\left(\sum_{j=0}^{d-1} h(S_{out}^j)\right)
	+ d \left( 2 h_B + \bigO{q h_B} \right)
	\\
	&=  \bigO{\frac{q h_B}{\epsilon}} \left( \sum_{j=0}^{d-1} h(S_{out}^j) \right) + d \cdot \bigO{q h_B}
	\\
	&\overset{(*)}{=}  \bigO{\frac{q h_B d^2}{\epsilon}} + \bigO{d q h_B} \\
	&=  \bigO{\frac{q h_B d^2}{\epsilon}}
	\end{aligned}
	\end{equation}
	where ($*$) holds since
	$h(S_{out}^j) \leq j + 3 + \epsilon$ for 
	$h(S^0) = h(i_t) \leq 1$ (Lemma\nobreakspace \ref {lemma:hSout}) and with the Gauss sum
	$
	\sum_{j=0}^{d-1} h(S_{out}^j) 
	\leq \sum_{j=0}^{d-1} j + 3 + \epsilon
	\leq \sum_{j=0}^{d-1} \bigO{j}
	= \bigO{d^2}
	$
	follows.
	
	Additionally, the repacking due to \textsc{InsertContainer} is by Equation\nobreakspace \textup 
	{(\ref{eq:repackShiftInsCon})} at most $\bigO{\frac{p q}{\epsilon}}$.
	As $p = \bigO{d^2}$ (see Table\nobreakspace \ref {tab:boundsForValues}), the 
	first three operations dominate the total repacking.
\end{proof}

\begin{table}
	\caption{Variables used for the migration analysis}
	\label{tab:boundsForValues}
	\centering
	\def\arraystretch{1.2}
	\begin{tabular}{@{}llll@{}}
		Symbol & Meaning & Asymptotic bound & Reference \\
		\midrule
		$\omega$ & \# categories 
		& $\bigO{\log \frac{1}{\epsilon}}$
		& Lemma\nobreakspace \ref {lemma:NumberOfGroups} \\
		$d$      & max. length of shift sequence 
		& $\bigO{\frac{1}{\epsilon} \log \frac{1}{\epsilon}}$
		& Lemma\nobreakspace \ref {lemma:NumberOfGroups} \\
		$q$		 & \# aligned levels in Algorithm\nobreakspace \ref {alg:insertNarrowItem}
		& $\bigO{1 / \epsilon^2}$
		& Lemma\nobreakspace \ref {lemma:narrowItemsBoundRepacking} \\
		$h_B$	 & container height 
		& $\bigO{1 / \epsilon^2}$
		& Section\nobreakspace \ref {sec:approxGuarantee} \\
		$p$	 & \# repacked levels by \textsc{Improve}
		& $\bigO{\frac{1}{\epsilon^2} \log \frac{1}{\epsilon}}$
		& Theorem\nobreakspace \ref {theo:improveApplication}	
	\end{tabular}
\end{table}

\subsubsection{ShiftA}
\label{sec:migrationShiftA}

The steps performed in Lines~\ref{line:shiftAdefS}-\ref{line:shiftAPlace} of Algorithm\nobreakspace \ref {alg:shiftA} 
(\textsc{ShiftA}) are analog to Lines~\ref{line:shiftWidest}-\ref{line:shiftPlace} of 
Algorithm~\ref{alg:shift} (\textsc{Shift}), 
except for the height of the shifted items. 
Therefore, we reduce these parts of \textsc{ShiftA} to \textsc{Shift},
in order to make use of 
Equations\nobreakspace  \textup {(\ref {eq:repackShiftSink})} to\nobreakspace  \textup {(\ref {eq:repackShiftStretch})} .

\begin{lemma}
	\label{lemma:repackingShiftA}
	The total repacking performed in one call of Algorithm\nobreakspace \ref {alg:shiftA}
	is at most $\bigO{\frac{q {h_B}^2 \omega }{\epsilon^2}}$.
\end{lemma}

\begin{proof}
	After the for-loop in Algorithm\nobreakspace \ref {alg:shiftA}, there are three more sources
	of repacking: First, the $2^l$ new containers get inserted with
	Algorithm\nobreakspace \ref {alg:insertContainerShiftA}.
	Then, some items of group $(l,B,q(l,B))$ might be moved to other containers
	within the group. Finally, some containers may get deleted by
	Algorithm\nobreakspace \ref {alg:deleteContainerShiftA}.
	First, we consider the steps analog to \textsc{Shift}.
	
	\paragraph*{Sink, Place, Stretch}
	Let $b = q(l,B)$. 
	For each iteration $j \leq b$ of the for-loop in Algorithm\nobreakspace \ref {alg:shiftA}, let
	$S^j$ be the set $S$ defined in Line\nobreakspace \ref {line:shiftAdefS}.
	Repacking operations are sinking the containers with items from
	$S^j$ (in group $g_{j+1}$), placing $S^j$ (in group $g_j$) and
	stretching of levels affected by the placing of $S^j$.
	Therefore, we can apply 
	Equations\nobreakspace  \textup {(\ref {eq:repackShiftSink})} to\nobreakspace  \textup {(\ref {eq:repackShiftStretch})} 
	if we set $S_{out}^j = S^j$.
	With the definition of $u$ in Line\nobreakspace \ref {line:shiftAdefU} of Algorithm\nobreakspace \ref {alg:shiftA}
	we get
	\begin{multline*}
	h(S^0) \leq s_{min} + 1
	=     u - h(g_0)
	\overset{\text{\invE}}{\leq} 2^l k (h_B - 1) - ((2^l (k-1) -1) (h_B - 1)  \\
	=    (2^l + 1) (h_B - 1)
	\end{multline*}
	and for $1 \leq j < b$
	\begin{multline*}
	h(S^j)  \leq s_{min} + 1
	=    u - h(g_0)
	\overset{\text{\invE}}{\leq} 2^l (k-1) (h_B - 1) - ((2^l (k-1) -1) (h_B - 1)
	\\
	=    h_B - 1 \,.
	\end{multline*}
	Since $b \leq q(l,A) + q(l,B) = \bigO{\frac{\omega}{\epsilon}}$ (by Lemma\nobreakspace \ref {lemma:NumberOfGroups}), it follows
	$2^l + b = \bigO{\frac{\omega}{\epsilon}}$, and thus
	$
	\sum_{j=0}^{b-1} h(S^j) 
	\leq (2^l + 1) (h_B - 1) + (b-1) (h_B - 1 )
	=    (2^l + b) (h_B - 1)
	= \bigO{\frac{\omega h_B}{\epsilon}} \,.
	$
	Analogously to Lemma\nobreakspace \ref {lemma:shiftRepacking} 	we get
	\begin{equation}
	\label{eq:migrationSumShiftA}
	\begin{aligned}[b]
	&\quad \sum_{j=0}^{b-1} \mathit{Repack\_Sink}(j) 
	+ \mathit{Repack\_Place}(j)
	+ \mathit{Repack\_Stretch}(j) \\
	&=  \bigO{\frac{q h_B}{\epsilon}} \left( \sum_{j=0}^{b-1} h(S^j) \right) + \bigO{b q h_B}
	~~=  \bigO{\frac{q {h_B}^2 \omega}{\epsilon^2}} \,.
	\end{aligned}
	\end{equation}
	
	\paragraph*{Repacking of $q(l,B)$}
	At most $\ceilS{\frac{\Delta}{h_B - 1}}$ containers have to be emptied
	and by Equation\nobreakspace \textup {(\ref {eq:delContainersToRepairE})} this term is at most
	$\ceilS{2^l + 1}$. Therefore,
	\begin{align}
	\label{eq:repackShiftAEmptyContainer}
	\mathit{Repack\_{EmptyCon}} \leq \ceil{2^l + 1} h_B
	= \bigO{\frac{h_B}{\epsilon}} \,.
	\end{align}
	
	\paragraph*{Insert / Delete Containers}
	Algorithm\nobreakspace \ref {alg:insertContainerShiftA} needs one call of \textsc{Improve},
	analogously to the proof for \textsc{Shift} (Equation\nobreakspace \textup {(\ref {eq:repackShiftInsCon})})
	the size of repacking is at most 
	$\bigO{\frac{p q }{\epsilon}}$.
	Before containers get deleted, contained items are repacked into other
	containers, which is already treated above.
	There are at most $\ceilS{2^l + 1}$ containers that may get deleted
	(see above), each of them requires a single call of Algorithm\nobreakspace \ref {alg:deleteContainerShiftA}.
	\begin{align}
	\label{eq:repackShiftAInsDelCon}
	\mathit{Repack\_InsDelCon} \leq 
	\bigO{\frac{p q }{\epsilon}} \ceil{2^l + 1}
	= \bigO{\frac{p q}{\epsilon^2}} \,.
	\end{align}
	
	\paragraph{Conclusion} Using 
	Equations\nobreakspace  \textup {(\ref {eq:migrationSumShiftA})} to\nobreakspace  \textup {(\ref {eq:repackShiftAInsDelCon})} ,
	we notice that the operations \textsc{Sink}, \textsc{Place}, and
	\textsc{Stretch} dominate the overall repacking:
	$$
	\bigO{\frac{q {h_B}^2 \omega}{\epsilon^2}} + \bigO{\frac{h_B}{\epsilon}} 
	+ \bigO{\frac{p q}{\epsilon^2}}
	= \bigO{\frac{q {h_B}^2 \omega}{\epsilon^2}} \,.
	$$
\end{proof}

\subsection{Putting it together}
Now we are able to analyze the amortized migration factor of the presented
algorithm. Let $i_t$ be the item inserted at time $t \geq 1$.
For the amortized analysis we use the notion of repacking potential, thus let
$\Phi(t) = \Phi(t-1) + \mu \SIZE(i_t) - \Repack(t)$ be the total budget
at time $t$ and assume $\Phi(0)=0$.

\subsubsection{Insertion of a Big or Flat Item}
\label{sec:migrationFactorBigFlat}

Let $\varGamma$ be the maximum total size of repacking in a shift sequence.
By Lemma\nobreakspace \ref {lemma:shiftRepacking}, 
$\varGamma = \bigO{\frac{q h_B d^2}{\epsilon}}$.
First we need the following result, regarding the amount of repacking caused by the 
block balancing algorithm:

\begin{lemma}
	\label{lemma:blockBalancingConclusion}
	The block balancing procedure (Algorithm\nobreakspace \ref {alg:blockBalancing}) for a set of items $M$
	causes repacking of size at most 
	$\bigO{\frac{q {h_B} \omega }{\epsilon^2}} \SIZE(M)$.
\end{lemma}
\begin{proof}
	Lemma\nobreakspace \ref {lemma:movedGroupsBlockBalancing} states
	that the total number of moved groups for the insertion of $M$
	is at most $\frac{8+\epsilon}{2 h_B} \SIZE(M) + 1$.
	According to Lemma\nobreakspace \ref {lemma:repackingShiftA},
	moving one single group via \textsc{ShiftA} causes repacking
	of size $\bigO{\frac{q {h_B}^2 \omega }{\epsilon^2}}$.
	That is, the total repacking is at most
	$\bigO{\frac{q {h_B}^2 \omega }{\epsilon^2}} 
	\left(\frac{8+\epsilon}{2 h_B} \SIZE(M) + 1 \right) \allowbreak
	= \bigO{\frac{q {h_B} \omega }{\epsilon^2}} \SIZE(M)$.
\end{proof}

Since big items have a minimum size we immediately get the following result even without
amortization:

\begin{lemma}
	\label{lemma:migrationFactorBigItem}
	Inserting a big item via Algorithm\nobreakspace \ref {alg:insertBigItem} has 
	migration factor $\mu = \bigO{\frac{q h_B d^2}{\epsilon^3}}$.
\end{lemma}
\begin{proof}
	The first repacking source in Algorithm\nobreakspace \ref {alg:insertBigItem} is \textsc{Shift},
	whose repacking size is at most $\varGamma$.
	Additionally, the block balancing (Algorithm\nobreakspace \ref {alg:blockBalancing}) causes
	repacking. Considering a single big item $i_t$, 
	Lemma\nobreakspace \ref {lemma:blockBalancingConclusion} with $M= \{i_t\}$ bounds the repacking
	size with $\bigO{\frac{q {h_B} \omega }{\epsilon^2}}$. 
	Since this term is dominated by $\varGamma$, the
	total repacking in Algorithm\nobreakspace \ref {alg:insertBigItem} is $\bigO{\varGamma}$.
	Each big item has minimum size $\epsilon^2$, thus
	$ \mu = \frac{\bigO{\varGamma}}{\epsilon^2} 
	= \bigO{\frac{q h_B d^2}{\epsilon^3}}$.
\end{proof}

The insertion of a flat item is closely related to the procedure for big items.
But since flat items get buffered before the \textsc{Shift} operation, here
the amortized analysis comes into play.

\begin{lemma}
	\label{lemma:migrationFactorFlatItem}
	Inserting a flat item via Algorithm\nobreakspace \ref {alg:insertFlat} has amortized migration
	factor $\mu = \bigO{\frac{q h_B d^2}{\epsilon^4}}$.
\end{lemma}
\begin{proof}
	Let $i_t$ be the inserted item of category $l$ and let 
	$\mu = \bigO{\frac{1}{\epsilon^3}} \varGamma$.
	If $i_t$ can be inserted without overflowing the F-buffer, no repacking is performed.
	Now suppose that an overflow occurs. Then, each of the $2^l$ slots contains items
	of width at least $2^{-(l+1)}$ and total height at least $y - \epsilon = \card{G} \epsilon^2$.
	Consequently, all slots of category $l$ contain items of total size at least
	$2^l 2^{-(l-1)} \card{G} \epsilon^2 = \frac{1}{2} \card{G} \epsilon^2$. 
	None of those items used its repacking potential so far, thus 
	$\Phi(t-1) \geq \mu \frac{1}{2} \card{G} \epsilon^2
	= \bigO{1 / \epsilon} \card{G} \varGamma$.	
	
	The repacking in Algorithm\nobreakspace \ref {alg:insertFlat} is due to \textsc{Shift}, which is 
	called $\card{G(l)} n$ times, where $\card{G(l)} \leq \card{G}$ 
	denotes the number of groups of category $l$ and $n$ the size of the partition
	(see Line\nobreakspace \ref {line:partitionShiftFlat}).
	By \MakeUppercase Lemma\nobreakspace \ref {lemma:sizeShiftPartitionFlat}, $n \leq \bigO{1 / \epsilon}$.
	Like in the case of big items (Lemma\nobreakspace \ref {lemma:migrationFactorBigItem}), the repacking
	of \textsc{ShiftA} is dominated by the repacking for \textsc{Shift}.
	Hence, $\Repack(t) \leq \bigO{1 / \epsilon} \card{G} \varGamma$
	and $\Phi(t) \geq 0$.
	
\end{proof}

\subsubsection{Insertion of a Narrow Item}

\begin{lemma}
	\label{lemma:migrationFactorNarrow}	
	The insertion of a narrow item with Algorithm\nobreakspace \ref {alg:insertNarrowItem} has
	amortized migration factor $\mu= \bigO{q} = \bigO{1 / \epsilon^2}$.
\end{lemma}
\begin{proof}
	Let $\mu = \frac{(q+1) h_B}{(1-\epsilon)^2 (h_B - 1 - \lambda)}$.
	If $i_t$ can be placed in the N-buffer or via \textsc{Shelf-First-Fit},
	no repacking is performed.
	Now assume that $i_t$ can not be placed in the N-buffer. 
	Recall that the N-buffer has width 1 and height $h_B$, where sparse shelfs have
	total height at most $\lambda$.
	Thus, a full N-buffer contains
	items of total size at least $(1-\epsilon)^2 (h_B - 1 - \lambda)$.
	All buffer items and $i_t$ did not use their repacking potential so far, thus
	$\Phi(t) \geq (1-\epsilon)^2 (h_B - 1 - \lambda) \mu - \Repack(t)
	= (q+1) h_B - \Repack(t)$.
	
	On a buffer overflow, Algorithm\nobreakspace \ref {alg:insertNarrowItem} aligns at most $q$ levels
	of height $h_B$. Repacking the items from the buffer into the aligned levels
	(respectively on top of the packing) causes additional repacking of 
	size at most $h_B$.
	Hence, $\Repack(t) \leq (q+1) h_B$ and thus $\Phi(t) \geq 0$.
	As $\lambda \leq 1 / \epsilon^2$ and therefore
	$(h_B - 1 - \lambda) = \Omega (h_B)$,
	it holds that $\mu = \bigO{q}$.
\end{proof}

\subsubsection{Migration Factor of the Overall Algorithm}
Recall that the overall Algorithm\nobreakspace \ref {alg:insertGeneralSetting} has two modes
(semi-online and online).
In the following lemma we consider the semi-online mode.

\begin{lemma}
	\label{lemma:migrationSemiOnline}
	In the semi-online mode, Algorithm\nobreakspace \ref {alg:insertGeneralSetting} has 
	amortized migration factor 
	$	\mu = \bigO{\frac{\omega {h_B}^2}{\epsilon^3}}$.
\end{lemma}
\begin{proof}
	Let $\mu = \frac{4 \omega h_B}{\epsilon^3} (h_B + 1)$.
	For flat and narrow items we again use the notion of repacking potential.

	\paragraph*{Case: $\mathbf{i_t}$ is big}
	Since the size of the instance is bounded from above,
	$\Repack(t) \allowbreak \leq \SIZE(I_L(t)) \allowbreak \leq \frac{4 \omega h_B}{\epsilon} (h_B + 1)$.
	With the minimum size $\epsilon^2$ of big items, the migration factor
	$\mu = \bigO{\frac{\omega {h_B}^2}{\epsilon^3}}$ follows.
	
	\paragraph*{Case: $\mathbf{i_t}$ is flat}
	When $i_t$ can be placed without buffer overflow, no repacking is performed.
	Now suppose that $i_t$ can not be placed in any buffer segment of category $l$.
	Then, each of the $2^l$ slots of category $l$ is filled with items of width
	at least $2^{-(l+1)}$ and total height greater than $y - \epsilon = 2$,
	hence there are items of total size at least
	$2^l (2^{-(l+1)} \cdot 2) = 1$.
	None of the items from the buffer used its repacking potential so far,
	thus $\Phi(t) \geq 1 \cdot \mu - \Repack(t)$.
	We have 
	$\Repack(t) \leq \SIZE(I_L(t)) 
	\leq \frac{4 \omega h_B}{\epsilon} (h_B + 1) \leq \mu$ and 
	thus $\Phi(t) \geq 0$.
	After the repacking, all slots of each category in the buffer are empty.
	
	\paragraph*{Case: $\mathbf{i_t}$ is narrow}
	The narrow item gets placed into an existing shelf or into a new shelf
	on top of the packing. In both cases, no repacking is performed and
	thus nothing is to show.
\end{proof}

The next theorem restates the main result if this work (Thereom~\ref{abs:theo:mainResult}) in more precise
terms.

\begin{theorem}
	\label{theo:overallMigrationFactor}
	Algorithm\nobreakspace \ref {alg:insertGeneralSetting} is a robust AFPTAS for the online strip packing
	problem with amortized migration factor
	$ \mu = \bigO{\frac{1}{\epsilon^{10}} 
		\left( \log \frac{1}{\epsilon} \right)^2 }\,. $
\end{theorem}

\begin{proof}
	The approximation guarantee is already shown in Theorem\nobreakspace \ref {lemma:packingHeightOverallAlgo},
	and the running time is clearly dominated
	by \textsc{Improve}, see the proof of Theorem\nobreakspace \ref {theo:heightContainerPackingAfterImprove}.
	The migration factor is analyzed in 
	Lemma\nobreakspace \ref {lemma:migrationSemiOnline} for the semi-online mode
	and in Lemmas\nobreakspace  \ref {lemma:migrationFactorBigItem} to\nobreakspace  \ref {lemma:migrationFactorNarrow} for the online mode. Using the asymptotic bounds for the values $\omega, q, d, h_B$
	given in Table\nobreakspace \ref {tab:boundsForValues} we get the following
	amortized migration factors:
	
	\begin{tabbing}
		\hspace{5cm}\=\hspace{4cm}\=\kill
		Semi-online mode	\>  
		$\displaystyle{\bigO{\frac{1}{\epsilon^7} \log \frac{1}{\epsilon}}}$\> 
		Lemma\nobreakspace \ref {lemma:migrationSemiOnline}
		\\[1em]
		Online mode, big item	\>  
		$\displaystyle{\bigO{\frac{1}{\epsilon^9} 
				\log \left( \frac{1}{\epsilon} \right)^2}}$\> 
		Lemma\nobreakspace \ref {lemma:migrationFactorBigItem}
		\\[1em]
		Online mode, flat item	\>  
		$\displaystyle{\bigO{\frac{1}{\epsilon^{10}} 
				\log \left( \frac{1}{\epsilon} \right)^2}}$\> 
		Lemma\nobreakspace \ref {lemma:migrationFactorFlatItem}
		\\[1em]
		Online mode, narrow item	\> 
		$\displaystyle{\bigO{\frac{1}{\epsilon^2}}}$\> 
		Lemma\nobreakspace \ref {lemma:migrationFactorNarrow}
	\end{tabbing} 
	Thus, the insertion of a flat item via Algorithm\nobreakspace \ref {alg:insertFlat} 
	dominates the overall migration factor.
\end{proof}

\paragraph*{Acknowledgement} 
We would like to thank Marten Maack for helpful discussions on the lower bound
proven in Section~\ref{abs:sec:lowerBound}.

\bibliography{literature}

\newpage

\appendix

\section{Proof of Theorem~\ref{theo:improveApplication}}
\label{app:improve}

In order to proof Theorem~\ref{theo:improveApplication}, we need two other
results from \cite{jansen2013robust}.
For a vector $x$, let $\nnz{x}$ denote the number of non-zero components.
\begin{theorem}[Theorem 8 from \cite{jansen2013robust}]
	\label{theo:improve2}
	Let $\delta > 0, \alpha \in \NN$, and $x$ be a solution of the LP with 
	\begin{align}
	\norm{x} &\leq (1+\delta) \LIN \,,		\label{eq:improve:theo2A} \\
	\norm{x} &\geq 2 \alpha \left( \tfrac{1}{\delta}+1 \right ) \,. \label{eq:improve:theo2B}
	\intertext{Let $y$ be an integral solution of the LP and $D \geq \delta \LIN$ such that}
	\norm{y} &\leq \LIN+ 2D	 \,, \label{eq:improve:theo2C} \\
	\norm{y} &\geq (m+2) \left(\tfrac{1}{\delta}+2 \right) \,. \label{eq:improve:theo2D}
	\intertext{Further,}
	\nnz{x} &= \nnz{y} \,, \label{eq:improve:theo2E} \\
	x_i &\leq y_i \hspace*{12pt} \forall 1 \leq i \leq n \,. \label{eq:improve:theo2F}
	\end{align}
	The algorithm \textsc{Improve}($\alpha$,$x$,$y$) then returns a
	fractional solution $x'$ with
	$\norm{x'} \leq (1+\delta) \LIN - \alpha$ and an integral solution
	$y'$ where one of the two properties hold:
	$\norm{y'} = \norm{y} - \alpha$ or $\norm{y'} = \norm{x'} + D$.
	Further, $\nnz{x'},\nnz{y'} \leq D$, and
	the	distance between $y'$ and $y$ is bounded by 
	$\norm{y'-y} \in \bigO{\frac{m+\alpha}{\delta}}$.
\end{theorem}

\begin{corollary}[Corollary 9 from \cite{jansen2013robust}]
	\label{cor:improve}
	Let $\delta' \geq \delta$ and $x$ be a solution of the LP with
	\begin{align}
	\norm{x} &= (1+\delta') \LIN  \,, \label{eq:improve:cor2G} \\
	\norm{x} &\geq 2 \alpha \left( \tfrac{1}{\delta}+1 \right) \,. \label{eq:improve:cor2H}
	\intertext{Further, for some $D \geq \delta' \LIN$ let}
	\norm{y} &\leq \LIN + 2D \,, \label{eq:improve:cor2I} \\
	\norm{y} &\geq (m+2) \left( \tfrac{1}{\delta}+2 \right) \,, \label{eq:improve:cor2J}
	\intertext{and it holds that}
	\nnz{x} &= \nnz{y} \,, \label{eq:improve:cor2K}\\
	x_i &\leq y_i \hspace*{12pt} \forall 1 \leq i \leq n \,. \label{eq:improve:cor2L}
	\end{align}
	Then, algorithm \textsc{Improve}($\alpha$,$x$,$y$) returns a fractional solution
	$x'$ with $\norm{x'} \leq \norm{x} - \alpha = (1+\delta') \LIN - \alpha$
	and an integral solution $y'$ where one of the two properties hold:
	$\norm{y'} = \norm{y} - \alpha$ or 
	$\norm{y'} = \norm{x} - \alpha + D$.
	Further, $\nnz{x'},\nnz{y'} \leq D$, and
	the distance between $y'$ and $y$ is bounded by 
	$\norm{y'-y} \in \bigO{\frac{m+\alpha}{\delta}}$.
\end{corollary}

\subsection*{Proof of Theorem~\ref{theo:improveApplication}}
\begin{proof}
	
	By \MakeUppercase Theorem~\ref {theo:approxGuarantee}, we have
	$\OPT(\CR{\con}{R}) \leq (1+\epsilon') \OPT(I_L) + z$
	for $\epsilon' = 4 \epsilon$ and $z = \bigO{1 / \epsilon^4}$.
	Let $\Delta = (1 + \epsilon') (1+\delta) - 1$
	and $m \leq \frac{16 \omega}{\epsilon} + 2 \omega$ be the number of groups
	(Lemma\nobreakspace \ref {lemma:NumberOfGroups}).
	Finally, let $D = \Delta \OPT(I_L) + m + (1+\delta)z$.
	With the above definitions, it holds that
	\begin{align}
	\label{eq:improveOptCrOptIlA}
	(1+\epsilon) \OPT(\CR{\con}{R}) \leq (1+\Delta) \OPT(I_L) + (1+\epsilon)z
	\,.
	\end{align}
	Depending on the case $\delta' \leq \delta$ or $\delta' > \delta$,
	we apply Theorem\nobreakspace \ref {theo:improve2} or \MakeUppercase Corollary\nobreakspace \ref {cor:improve}.
	\paragraph*{Case: $\delta' \leq \delta$}
	At first, we check the prerequisites of Theorem\nobreakspace \ref {theo:improve2}.
	Properties \textup {(\ref {eq:improve:theo2B})} and\nobreakspace   \textup {(\ref {eq:improve:theo2D})} to\nobreakspace  \textup {(\ref {eq:improve:theo2F})} 
	hold by condition.
	Property \textup {(\ref {eq:improve:theo2A})} holds since
	$$ \norm{x} = (1+\delta') \LIN(\CR{\con}{R}) \leq (1+\delta) \LIN(\CR{\con}{R}) 
	\,. $$
	For eq.\nobreakspace \textup {(\ref {eq:improve:theo2C})}, we first show that $D$ can be bounded
	from below:
	\begin{align*}
	D &=     \Delta \OPT(I_L) + m + (1+\delta) z \\
	&=     ((1+\delta)(1+\epsilon')-1) \OPT(I_L) + m + (1+\delta) z \\
	&=     ((1+\delta)(1+\epsilon')-1) 
	\left( \frac{\OPT(\CR{\con}{R})}{1+\epsilon'} - \frac{z}{1+\epsilon'} \right) 
	+ m + (1+\delta) z \\
	&=     (1+\delta) \OPT(\CR{\con}{R}) - \frac{\OPT(\CR{\con}{R})}{1+\epsilon'} + \frac{z}{1+\epsilon'} + m \\
	&>     \delta \OPT(\CR{\con}{R}) \\
	&\geq  \delta \LIN(\CR{\con}{R})
	\end{align*}
	Now we can show the second part of eq.\nobreakspace \textup {(\ref {eq:improve:theo2C})}:
	\begin{align*}
	\norm{y} &\leq (1+2\Delta) \OPT(I_L) + 2 (1+\delta)z + m \\
	&\leq 2 \Delta \OPT(I_L) + \OPT(\CR{\con}{R}) + 2 (1+\delta) z + m \\
	&\leq 2 \Delta \OPT(I_L) + \LIN(\CR{\con}{R}) + m + 2 (1+\delta) z + m \\
	&=    2(\Delta \OPT(I_L) + m) + \LIN(\CR{\con}{R}) + 2 (1+\delta) z \\
	&=    2D - 2 (1+\delta)z + \LIN(\CR{\con}{R}) + 2(1+\delta) z \\
	&=    \LIN(\CR{\con}{R}) + 2D
	\end{align*}
	Hence, Theorem\nobreakspace \ref {theo:improve2} is applicable.
	Therefore, \textsc{Improve}($\alpha$,$x$,$y$) returns solutions $x'$, $y'$ such that
	\begin{align*}
	\norm{x'} &\leq (1+\delta) \LIN(\CR{\con}{R}) - \alpha \\
	&\leq (1+\delta) \OPT(\CR{\con}{R}) - \alpha \\
	&\leq (1+\delta) ((1+\epsilon') \OPT(I_L) + z) - \alpha \\
	&=    (1+\delta)(1+\epsilon') \OPT(I_L) + (1+\delta)z - \alpha \\
	&=    (1+\Delta) \OPT(I_L) + (1+\delta)z - \alpha
	\,.
	\end{align*}
	Moreover, in the case $\norm{y'} = \norm{x'} + D$ it holds that
	\begin{align*}
	\norm{y'} &=    \norm{x'} + D \\
	&\leq (1+\Delta) \OPT(I_L) + (1+\delta)z - \alpha + D \\
	&=    (1+\Delta) \OPT(I_L) + (1+\delta)z - \alpha
	+ \Delta \OPT(I_L) + m + (1+\delta)z \\
	&=    (1+2\Delta) \OPT(I_L) + 2(1+\delta)z + m - \alpha
	\,.
	\end{align*}
	In the case $\norm{y'} = \norm{y} - \alpha$, the claim follows immediately
	by condition \textup {(\ref {eq:improve:theo3e})}.
	The remainder of the claim follows by the corresponding
	implications of Theorem\nobreakspace \ref {theo:improve2}.

	\paragraph*{Case: $\delta' > \delta$}
	Here, we apply \MakeUppercase Corollary\nobreakspace \ref {cor:improve} and again show the
	prerequisites first. By condition, Properties 
	\textup {(\ref {eq:improve:cor2G})},  \textup {(\ref {eq:improve:cor2H})}, and\nobreakspace   \textup {(\ref {eq:improve:cor2J})} to\nobreakspace  \textup {(\ref {eq:improve:cor2L})}  hold and it remains to show 
	eq.\nobreakspace \textup {(\ref {eq:improve:cor2I})}. Note that $\norm{y} \leq \LIN(\CR{\con}{R}) + 2D$ 
	holds for the same reason as in the case $\delta' \leq \delta$.
	Therefore, we only have to show $D \geq \delta' \LIN(\CR{\con}{R})$:
	\begin{align*}
	\norm{x} &=    (1+ \delta') \LIN(\CR{\con}{R}) \\
	&\leq (1+ \Delta) \OPT(I_L) + (1+\delta) z \\
	&=    \OPT(I_L) + \Delta \OPT(I_L) + (1+\delta)z \\
	&\leq \LIN(I_L) + m + \Delta \OPT(I_L) + (1+\delta)z \\
	&\leq \LIN(\CR{\con}{R}) + m + \Delta \OPT(I_L) + (1+\delta)z \\
	&=     \LIN(\CR{\con}{R}) + D
	\end{align*}
	Hence, $D \geq  (1+ \delta') \LIN(\CR{\con}{R}) - \LIN(\CR{\con}{R})
	= \delta' \LIN(\CR{\con}{R}) $.
	Thus, we can apply \MakeUppercase Corollary\nobreakspace \ref {cor:improve}.
	\textsc{Improve}($\alpha$,$x$,$y$) returns solutions $x'$, $y'$ such that
	$$
	\norm{x'} \leq \norm{x} - \alpha \leq (1+\Delta) \OPT(I_L) + (1+\delta)z - \alpha
	\,. $$
	For the integral solution $y'$, in the case $\norm{y'} = \norm{x} - \alpha + D$
	it holds that
	\begin{align*}
	\norm{y'} &=    \norm{x} - \alpha + D \\
	&\leq (1+\Delta) \OPT(I_L) + (1+\delta) z + D - \alpha \\
	&=	(1+\Delta) \OPT(I_L) + (1+\delta) z + \Delta \OPT(I_L) + m
	+ (1+\delta) z - \alpha \\
	&=	(1+2\Delta) \OPT(I_L) + 2(1+\delta)z + m - \alpha
	\,.
	\end{align*}
	And in the case $\norm{y'} = \norm{y} - \alpha$, condition 
	\textup {(\ref {eq:improve:theo3e})} gives the claim directly.
	Again, the remaining properties of the claim follow by the corresponding
	implications of \MakeUppercase Corollary\nobreakspace \ref {cor:improve}.
\end{proof}

\end{document}